\documentclass[nofootinbib,preprintnumbers,amsmath,amssymb,twocolumn]{revtex4}
\usepackage{amsthm}

\makeatletter
\renewcommand{\@seccntformat}[1]{}
\makeatother

\newtheorem*{postulate}{Postulate}
\newtheorem{theorem}{Theorem}
\newtheorem{lemma}{Lemma}

\begin{document}

\title{The hard problem and the measurement problem:\\a no-go theorem and potential consequences}
\author{Igor Salom}
\affiliation{Institute of Physics, Belgrade \\ University, Pregrevica 118, Zemun, Serbia }

\begin{abstract}
The ``measurement problem'' of quantum mechanics, and the ``hard problem'' of cognitive science are the most profound open problems of the two research fields, and certainly among the deepest of all unsettled conundrums in contemporary science in general. Occasionally, scientists from both fields have suggested some sort of interconnectedness of the two problems. Here we revisit the main motives behind such expectations and try to put them on more formal grounds. We argue not only that such a relation exists, but that it also bears strong implications both for the interpretations of quantum mechanics and for our understanding of consciousness. The paper consists of three parts. In the first part, we formulate a ``no-go-theorem'' stating that a brain, functioning solely on the principles of classical physics, cannot have any greater ability to induce subjective experience than a process of writing (printing) a certain sequence of digits. The goal is to show, with an attempt to mathematical rigor, why the physicalist standpoint based on classical physics is not likely to ever explain the phenomenon of consciousness -- justifying the tendency to look beyond the physics of the 19th century. In the second part, we aim to establish a clear relation, with a sort of correspondence mapping, between attitudes towards the hard problem and interpretations of quantum mechanics. Then we discuss these connections in the light of the no-go theorem, pointing out that the existence of subjective experience might differentiate between otherwise experimentally indistinguishable interpretations. Finally, the third part is an attempt to illustrate how quantum mechanics could take us closer to the solution of the hard problem and break the constraints set by the no-go theorem.
\end{abstract}

\maketitle

\section{Introduction}

Some critics rightfully argue that majority of attempts to relate the quantum measurement problem with the hard problem of consciousness follow a pretty shallow reasoning pattern: both problems are difficult and open, so thus they must be somehow related. Indeed, we are not aware of many serious scientific papers (clearly set apart from pseudoscience) which discuss the precise connection of the two problems. And yet, the existence of this connection is nevertheless beyond any dispute, at least formally. As this might sound to some as a (slightly at least) controversial statement, we will right away mention an example that clearly proves the assertion. In Everett's seminal paper \cite{EverettThesis, EverettThesis2} (doctorate thesis, to be more precise) where he introduces the ``relative state formalism'' -- the basis of now quite popular many world's interpretation -- he explicitly writes: ``As models for observers we can, if we wish, consider automatically functioning machines, possessing sensory apparata and coupled to recording devices capable of registering past sensory data and machine configurations". Therefore, by equating observer with an automaton, he decisively takes a very definite side in the ongoing debate in the contemporary cognitive science and also in AI research (presuming such automaton would possess ``strong AI"\footnote{Term coined by J. Searle for a machine with mind, i.e.\ consciousness.}), seemingly not even noticing that he is introducing a highly nontrivial statement. This cannot be easily brushed aside as he then completely relies on this assumption to infer, as he writes a bit later, that the wavefunction collapse thus appears merely as a sort of illusion in the \emph{subjective experience} of the observer.

We will return to this aspect of the many-worlds interpretation (MWI) when dealing more comprehensively with the relation between the hard problem and the interpretations of quantum mechanics (QM). Of course, it is not only the many-worlds interpretation that takes sides in the cognitive science debate, nor it is only the hard-line physicalist views that enjoy the support of the physicists. Quantum mechanics has been notorious for flirting with ideas on consciousness that were not generally considered as mainstream in the western school of thought. Since the advent of quantum mechanics, scientists were perplexed by the new role that the act of observation had seemingly gained in physics. Basic (Copenhagen) formulation of QM carried an inherent, though maybe unspoken, subjective character that was hard to avoid in the measurement postulate, palpable in the vagueness of (or lack of) the definitions of the ``measurement apparatus'' and the ``moment of measurement". The formalism allowed the Heisenberg cut -- that was dividing the objective realm of superpositions governed by deterministic Schrodinger equation and the subjective perception of a single well-defined stochastic measurement outcome -- to be positioned fluidly and arbitrarily along the von Neuman's chain. This prompted some of the founding fathers of QM to assume crucial interconnection between the phenomenon of subjective experience and the basic laws of nature: from Heisenberg's ``the concept of the objective reality of the elementary particles has thus evaporated...'' \cite{Heisenberg}, via Planck's ``I regard consciousness as fundamental, I regard matter as derivative from consciousness'' \cite{Planck}, to Schr\"odinger's ``mind has erected the objective outside world...out of its own stuff” (from his book of an appropriate title ``Mind and Matter'' \cite{Schrodinger}) and all the way to the most radical Wigner's stance that ``consciousness'' \emph{objectively} causes reduction of the wavepacket \cite{Wigner}.

The dilemma of whether the wavefunction collapse is objective or subjective (if existent at all) still lingers, being a divergence point of many interpretations of QM. The views that hold the collapse to be truly objective, in some experimentally measurable way, actually should not be seen as interpretations of QM, but as scientific proposals which predict new physics (we prefer to denote this type of ideas as ``mechanical collapse'' models). Namely, in whichever way the objective collapse is there envisioned, such a collapse event can be in principle enclosed in an isolated box and the following noted: the Schr\"odinger equation dictates that the evolution within the box must be linear and unitary, while the proposed objective wave-packet reduction demands the opposite. The expected violation of unitarity is, in principle, a measurable effect and thus requires both experimental confirmation (so far in all cases lacking) and a mathematical modification of the QM formalism that would replace and extend Schrodinger's equation in a consistent way. Very diverse ideas belong to this same category of objective (mechanical) collapse, some of them proposing certain relations to ``consciousness'' and some not at all: from purely physicalist GRW collapse models \cite{GRW}, over Penrose-Hameroff orchestrated objective reduction hypothesis \cite{OrchOR}, to Wigner's expectation that ``consciousness'' should somehow cause this objective collapse \cite{Wigner}. While all suffer from the same lack of a slightest experimental indication in their favor, those attempting to relate collapse with mental phenomena in an experimentally confirmable way are further hampered by difficulties to construct a consistent underlying mathematical theory (most evident in the Wigner's case: whose consciousness is entitled to cause collapse, would a mouse, or an amoeba, or single neuron suffice, and how to write down an analogon of the Schr\"odinger equation which would break down unitarity only for complex enough systems representing conscious observers?)

There is a fairly common misconception that relating QM with subjective, i.e.\ observer-dependent elements must lead in the direction of the previously mentioned ``consciousness causes objective collapse'' variations. On the contrary, interpretations that firmly stick to the existing QM formalism (predicting no new physics) are the most interesting from this aspect. The original Copenhagen interpretation, strictly adhering to both the postulate of unitary evolution and the measurement postulate, can be seen as a representative: the two postulates are not contradicting as we will discuss in more detail, but can be reconciled if we take the collapse to occur subjectively, whereas the perpetual unitary evolution is applicable from the ``third-person view". Apart from the explicitly ``subjectivist'' reading of Copenhagen interpretation, but similar to it, in the group of interpretations that fully acknowledge both postulates at the expense of introducing observer-dependent elements (such as subject/observer/system relative collapse) are also e.g.\ Wheeler's ``participatory universe'' concept \cite{Wheeler1, Wheeler2}, QBism \cite{QBism}, relational interpretation \cite{Rovelli} and Brukner's variant of the latter \cite{Brukner, Brukner2}. Some of these views advocate a relatively radical denial of the existence of objective (non-observer-relative) facts in general (Brukner's ``no facts of the world'' doctrine), but this is not a necessity \cite{ToTheRescue}. While all these views see collapse to some extent as subjective only, the collapse nevertheless does happen, and a subject exists always in a single well-defined state of reality (from his reference point at least). This is to be contrasted with the many-world interpretations (or, somewhat more generally, Everett's relative state interpretation) where the subject himself keeps being in a superposition and subjective perception of a single reality is a sort of illusion itself (the result of a certain brain state being correlated with only the corresponding state of the environment).

Therefore, from the speculations of the fathers of QM to the present-day debates on the enduring issues, attempts to understand quantum mechanics have often led out of the boundaries of the purely objective realm, touching upon the problems of cognition and subjective perception.

More recently, in a last few decades, we have witnessed a push from the opposite direction: researchers dealing with the ``hard problem of consciousness'' are increasingly looking in direction of quantum mechanics, feeling that some extension of the classical physics paradigm (either into quantum mechanics or into something else) is a prerequisite for further progress in their field. Namely, in spite of the advances in the understanding of brain functions and of the neural correlates of consciousness -- and maybe even more so due to these advances -- many scientists and philosophers find that there is an ``explanatory gap'' between understanding of how the brain functions and why does that functioning produce any subjective, first-person experience. Most notably, this problem with a long philosophical history has been, some 20 years ago, rephrased by David Chalmers as the ``hard problem of consciousness": even if we understand all brain functions and explain in detail how the physical processes inside it produce behavior, there remains a puzzle of ``why should physical processing give rise to a rich inner life at all?'' \cite{Chalmers95}. Why are not we ``philosophical zombies", biological mechanisms performing all human functions, but devoid of any subjective experience? (Throughout the text we use notions ``subjective experience'' and ``consciousness'' interchangeably -- possible subtle differences are highly definition-dependent and are of no interest here.)

In attempts to solve this biggest problem of cognitive science (as proponents of the ``hard problem'' see it), many researchers in the area have sought possible answers in physics, e.g.\ \cite{Eccles, Lockwood, Chalmers96, Hoffman, CognitiveToQuantumExample1}. Already in his seminal paper \cite{Chalmers95} and the follow up \cite{Chalmers97}, Chalmers discusses what QM potentially has to offer to elucidate or at least ameliorate the ``hard problem", rightfully pointing out that not all interpretations of QM are born equal in this regard. This remark leads to the following question: if explaining the phenomenon of subjective experience really requires a departure from classical physics, can the ability to potentially explain this phenomenon be used as a further probe into different interpretations of QM -- and tell apart even those which yield absolutely identical objective-experimental predictions? In other words, can the ``subjective-experimental'' fact that we possess the first-person experience succeed where the ``objective-experimental'' facts have failed: to differentiate between mathematically identical interpretations of QM?

It should be noted that some of the researchers working on the hard problem are seeking extensions of classical physics even farther than QM. Ideas of panpsychism, that attribute to matter another fundamental property often called proto-consciousness, are gaining popularity of recent \cite{Koch, Tononi, PanpsychismOverview}. As dubious from the perspective of physics as they are, these research directions provide further incentive to seek arguments that QM, properly understood, is alone sufficient to account for the existence of subjective experience.

However, not everybody from the cognitive side of the fence agrees that any departure from cold minded rational materialism of classical physics is needed at all. Often, as prominently exemplified by Daniel Dennett's stance, there is already disagreement on whether at all there is any ``hard problem'' that awaits for explanation \cite{Dennett, TypeA2, TypeA3}. For, if neuroscience can entirely explain why we behave in certain ways, why we report things as ``I see red", ``I subjectively feel that I exist", or ``I am puzzled by the hard problem of consciousness'' and if we can explain all that without need to invoke any ``consciousness'' thing, then it must be wrong to artificially postulate existence of such ``consciousness". Moreover, if we do not need any consciousness to explain any objective phenomenon or experiment, then acknowledging its existence would be tantamount to acknowledging as real something that objectively does not exist -- and no wonder that such a mistake would lead to a ``hard problem". To proponents of such views, that something which we subjectively experience and which, as it seems to us, instigates us to utter ``I feel that I am conscious'' is only some sort of illusion. In addition, there are also somewhat different purely materialist positions that acknowledge the phenomenon of subjective experience as real and yet hold that this first-person perspective is only a byproduct of the functioning of the material brain, a sort of ``emergent phenomenon'' \cite{TypeB1, TypeB2, TypeB3, TypeB4}. This emergence is then usually seen as a nontrivial (sometimes even irreducible) consequence of the extreme complexity and peculiar structure of the brain organ. Often, but not always, it is implied that a similar ``subjective experience'' phenomenon should accompany information processing of the comparable scale and type, even if performed in a different medium (e.g.\ in a computer, by proponents of the ``strong AI'' idea).

Proponents of these strictly physicalist views, by a rule, maintain that classical physics seems to be perfectly sufficient to explain the phenomenology of brain processes and thus also our behavior and what we call ``consciousness". Classical physics here of course also includes chemistry of the processes in brain, but our emphasis on ``classical'' implies that neither explicit quantum effects (e.g.\ interference, entanglement over distance), nor, more importantly, underlying quantum principles (indeterminism, indefiniteness of system properties manifest through superpositions) should play any role in explaining of the subjective experience. The good old XIX century mechanistic and deterministic picture of reality should suffice since the direct chain of causes and effects can be followed through a brain and result in complete explanation of its outputs given the corresponding sensory inputs and memory content. Not only that neuroscience supports the view that the inner workings of the brain can be fully explained in a deterministic causal manner (without resorting to QM), but also progress in AI research indicates that deterministic classical computers should in principle be able to replicate computational powers of the brain. So, if the paradigm of classical physics is indeed sufficient to explain all objective brain functions, is it really necessary to extend this paradigm in order to encompass also this murky objectively-non-definable phenomenon of subjective experience? Is not it more rational to expect that the so-far-eluding subjective phenomenon will be also somehow explained within the basic mechanically-deterministic framework? And even if we reach the opposite conclusion and decide as necessary to extend the paradigm on which physics rests, how could that help?

Of course, essential is the main question of whether subjective experience can be explained solely based on the ontology of classical physics -- as long as this direction seems fairly plausible, going further into discussions of possible paradigm extensions (in cognitive science) is hard to justify. However, arguments are abundant both pro and contra, though a reader going through discourses on the subject can sometimes get an impression that proponents of opposing views either speak in different languages or really have different personal levels of subjective experience -- so often reasons that seem too obvious to one side are quite unintelligible to the other. It is not surprising that, on this extremely abstract, by-definition-non-objective topic where even the central concept of ``consciousness'' eludes strict definition, a good part of the arguments reduces either to hand waving or to pointless analogies. Several inspiring thought experiments/mental abstractions have become very influential (e.g.\ Chinese room argument \cite{ChineseRoom}, Mary the color-blind scientist \cite{ColorblindMary}, philosophical zombies \cite{PhilosophicalZombies}), but in the end they at most served as good intuition pumps, not leading to any convergence of the opposing opinions.

We believe that some progress can be made by bringing the discussion to a more formal level. In spite that the topic here obviously does not belong to mathematics, any related analysis certainly employs elements of logic, which motivated us to formulate our statements here in the form of ``Theorems'' (regardless of the mathematical jargon, the text does not contain formulas and is fairly easy to follow). In the next section, we formulate and prove our main ``no-go'' theorem.  The theorem is the final step in a succession of a number of previously defined theorems, each derived from the predecessors. This no-go theorem claims that a human brain, functioning on principles of classical physics, cannot have any greater ability to induce subjective experience than a process of writing (printing) a certain sequence of digits (printing can be done in an arbitrary way and on the arbitrary surface).\footnote{Formaly precise statement is postponed until the next section.}

This no-go theorem per se does not directly say anything about whether the subjective experience can or can not be explained based on processes governed by classical physics. We do not think that such formal conclusion can be ever given, simply because there is no generally accepted definition of consciousness (and there cannot be, as some find it to be a mere illusion, while others find it so real to pose the hard problem). Instead, we try to avoid the trap of defining the consciousness by confining our reasoning only to establishing qualitative ``equivalence classes'' of subjective experiences. Surprisingly, as we show below, this can be done even in the absence of a precise definition. In this way, by building equivalences step-by-step, we finally conclude that a human brain, active for a certain period, must produce the same or similar type of subjective experience as the process of writing of a certain huge number.

Once we establish this equivalence, we will introduce the notion of mapping between the evolution of the system (functioning of the brain or writing of the digits) and the content of the subjective experience (i.e.\ the scene and qualia that are being experienced). If we assume that this mapping exists and is nontrivial -- in other words, that some subjective experience (with a clear meaning) emerges as the consequence of the dynamics of the physical system -- then, due to the established equivalence, we will face a grave problem to explain where, when and why exactly this one mapping occurs and is somehow experienced while writing (or chanting) the numbers. Alternatively, to avoid logical inconsistencies and openly magical connotation, we may assume a Dennett-like position and claim that there is no such mapping at all and that the dynamics of the system (i.e.\ behavior) is all there is. In other words, that the equivalence stated in the no-go theorem is of form 0=0, i.e.\ that neither brain nor writing of a number can induce any experience and that whole consciousness thing is an illusion. But the translation from a human brain to a number being written helps to clarify, as we argue below, that the zero in this case, i.e.\ the absence of emergence of anything, must be a true nothing and cannot be even an illusion. Namely, we will clarify that denying any (nontrivial) mapping implies that there is no any meaning of the experience, nothing to specify the particular (illusionary or not) experienced content or the qualia, and we can no more speak even of an illusion, simply because it cannot be said ``what is that illusion of". (After that, it is up to the reader to perform a subjective experiment of checking if possesses or not any sort of inner experience, even an illusion of, and to compare the results with the hypothesis.)

Of course, one can also try to logically dispute the theorem. Contrary to the theorem, one can try to argue that the premise of classical physics can be maintained while still insisting that activity of a human brain produces subjective experience whereas writing of any number does not give rise to any, let alone similar experience. In this case, it must be precisely explained and elaborated which logical step in the no-go derivation was unwarranted, and why. By splitting the main statement into a number of steps and by assuming a math-like approach to discussion we hope to keep counterarguments and pondering about the theorem to the point, forcing any analysis to pinpoint the exact disputed step and detail (including the obvious need for a clarification of why is that point relevant and acknowledgment of the logical consequences of taking a different turn at that point).\footnote{It is not likely than anyone would even consider hand-waving or an analogy as a counterargument to a theorem.} If reasonable, such a counterargument could shed substantial light on what is essential for the emergence of the consciousness -- history of science is full of examples where finding holes in no-go theorems has provided new directions for progress. (Interestingly, some preliminary discussions about the theorem have shown that even people who disagree about the final conclusion of the no-go theorem are looking for logical ways out in very different directions).

Finally, the logical option that we find most plausible is that the basic premise of the sufficiency of classical physics is simply wrong. The rest of the paper is devoted to the implications of such a conclusion. Pursuing the hope that, unlike classical physics, QM could be sufficient, in the second part we investigate which interpretations of QM are likely to sustain the test of the no-go theorem. The underlying idea is to extend the corpus of the empirical data a good theory should take into account in a way that it includes not only external, third-person perspective experiments, but also the most elementary first-person empirical data -- i.e.\ the fact that we possess some sort of subjective experience. While it would require full resolution of the hard problem to truly establish which interpretations positively can account for the consciousness in a consistent way, in that section we will concentrate on a more modest problem: which interpretations have potential merely to defy the conclusions of the no-go theorem, when the latter is generalized to QM. (We will also restrict our analysis only to a few more common interpretations of QM.)

Important in this regard will be the recognition that the ``measurement problem'' is essentially a direct rewriting of the ``hard problem'' in the language of QM. Namely, those cognitive-science positions that maintain that the third-person perspective is sufficient for description and understanding of the entire universe, including understanding of the (illusion of) subjective perspective, naturally translate into no-collapse (many-worlds) interpretation(s), where objective unitary evolution of the complete wavefunction is taken to provide the full picture of reality. On the other hand, cognitive-science views sceptical of the prospect that the phenomenon of subjective experience can be derived from the third-person explanation of the world, naturally map into QM interpretations that hold linear unitary evolution governed by Schr\"odinger's equation as insufficient to explain our subjective perception of the single measurement outcome (i.e.\ those that find the collapse postulate also necessary). Investigation of this connection will become the central issue of Part II. However, once this mapping is understood, it is not surprising that interpretations of QM inherit, at least partially, the same consequences of the no-go theorem as the corresponding ``interpretations'' of the hard problem.

Finally, in the Part III, we will try to address the more difficult question: whether and how the introduction of quantum mechanic can truly change the conclusions of the no-go theorem and at least leave some room for a plausible explanation of the subjective experience. The considerations in this section will be more speculative, aimed to demonstrate that QM, with a correct reading, has this potential. At this point, it is only worthwhile to stress that neither ``consciousness causes objective collapse'' (Wigner-like) approaches nor invoking explicit quantum phenomena in the high-temperature noise dominated brain (along Penrose-Hameroff lines) is necessary for this. It is the mere turning upside-down of the ontology of classical physics, and recognition of the subjective experience as a primary non-derivable entity while preserving the laws of QM unchanged that may provide the ground to truly understand both the measurement and the hard problem. Nevertheless, as we will see, even after such paradigm-shift the essence of the no-go theorem still largely constrains the set of available hypotheses.

\newpage
\section{Part I: the no-go theorem}

In this section, we explore in detail and with logical rigor the assumption that the concepts of classical physics are sufficient to explain the phenomenon of subjective experience. This assumption is commonly taken for granted in the neuroscientific brain research, and thus we label our starting position as the ``neuroscience postulate'' (NP):

\begin{postulate}
Neuroscience postulate (NP). Classical physics, based on ideas of mechanical-determinism and ontology of physicalism, is sufficient to explain the functioning of the brain and the emergence of consciousness.
\end{postulate}

Of course, the classical chemistry is included above as a part of classical physics (the latter is understood in broad terms). By ontology of physicalism, we imply the paradigm that only matter (including fields) and material things are considered to be real. In addition, we assume by the postulate that matter follows certain causally closed deterministic dynamics, the details of which are unimportant. In particular, the postulate excludes dualistic and idealistic ontologies, as well as the essential traits of quantum mechanics: the matter properties are here always well defined and observer-independent. However, we will also take into consideration physicalist views on consciousness which incorporate (explicitly or in-between the lines) mapping from physical states of the brain to the ``emerged'' conscious states -- in spite that such views might be sometimes classified as property-dualistic.

Note that we deliberately do not attempt to strictly define the word ``consciousness". It can be either, following ideas of D. Dennett, that ``illusion'' which accompanies certain brain responses and behavior, or whatever complex phenomenon that somehow emerges during the functioning of the brain in what is commonly seen as a conscious state. Or it could be any other (more or less common) understanding of consciousness which complies with the rest of the postulate presumptions (but it cannot be anything dualistic in common sense, like qualitative states of non-material spirit or mind). In this sense, any time we mention ``consciousness'' in the context of this analysis, one is free to insert words ``an illusion of'' beforehand.

Now we proceed to the formal analysis of a certain chain of logical consequences of the Postulate, by formulating a sequence of Theorems which end in the already announced ``no-go'' theorem. Here, in the main text, we will describe only the outlines of the ``proofs"\footnote{In spite that the subject here is not of mathematical nature, and that the weight of the logical arguments given in favor of some statement here cannot be compared with the strength of the ``proof'' in the mathematically strict sense, we nevertheless prefer and dare to use this word, albeit with quotation marks.}, whereas more detailed arguments are carefully considered in the appendix. (Reader is encouraged to skip the detailed ``proofs'' in the appendix of those theorems that (s)he personally finds self-evident -- in this way we hope to facilitate perusal of the otherwise maybe overly lengthy paper).

The first question we have to deal with is whether this ability to generate subjective experience (or an illusion of it), which we assign to a human brain, can be granted to anything else apart from the brain. In particular, whether we could consider as conscious another being, behaving in all (relevant) respects indistinguishable from a human (e.g.\ demonstrating emotions, reporting self-awareness, claiming to experience qualia, even being capable to autonomously reach the ``hard problem'' and express its puzzlement over it \cite{ConsciousnessTest}), but whose physiology of the cognitive organ (brain or some equivalent) is possibly different? Apart from being of philosophical significance, this question has important ethical, as well as potential legal (in the wake of the AI era) implications.

We hold that the answer to such a question must be positive, at least under the presumption of the NP. Namely, we must first note that we have no objective means to infer the existence of anyone's subjective experience apart from our own since, already by its name, we generally here denote a subjective phenomenon. Even assuming consciousness of other human beings formally requires a ``leap of faith", but one that in general must be made to avoid the blind alley of solipsism (the so called ``problem of other minds'' \cite{OtherMinds}). Besides, solipsism is certainly not an option allowed by the Neuroscience postulate.

Within the paradigm of physicalism (and thus also when deprived of the solipsism option), it is very hard to motivate why the attribute of being conscious should be given to another human, but not to another agent autonomously behaving in a qualitatively indistinguishable way. For example, we may imagine a perfect human-replica android, e.g.\ like the popular Data from Star Trek science fiction series and an identically looking human, e.g.\ the actor Brent Spiner, next to each other (unlike in the original Star Trek plot, we shall assume that this Data is also programmed to show emotions just like a human). We have absolutely no means to verify if Brent Spiner possesses any inner experience, but he claims to be conscious, behaves so in every way, and we agree that he must be conscious. We have no more and no less chance to objectively confirm or deny that Data has subjective experience, and he behaves in every way like the human Spiner, reporting his conscious/emotional states and reacting to stimuli in a human-like fashion. Attributing consciousness to one while not to the other, in spite of admitting that there is no objective way in which such an assertion could be tested, is to actually acknowledge the real existence of objectively not measurable (yet well defined) properties. While this would be in principle logically acceptable e.g.\ in a dualistic setup (where, theoretically, one could even arbitrarily assign that certain agents have soul/spirit attached to them while some other are philosophical zombies), such assignment is incompatible with physicalist presumption implied by NP. (Formally, one could misuse our freedom in definition of consciousness to define it as ``possession of human-like nervous system'' but such a definition would completely miss the common meaning of the word and fail to grasp the essence of the ``subjective experience'' problem.) Therefore, we formulate the following theorem:

\begin{theorem}
\label{Theorem1}
(Behavioral theorem). Under the premise of NP, two agents behaving effectively in the same way when it comes to reacting to external stimuli, reporting about internal conscious/emotional states, having the ability to autonomously introspect and recognize its subjective experience and engage into related discussions -- must possess the same qualitative level of subjective experience.
\end{theorem}

By saying ``the same qualitative level'' we leave the freedom that the two agents may experience qualia in (somewhat) different ways. These details are not essential here. Simply stated, the crucial conclusion of the theorem is that of the two such agents either both have subjective experience or none has. Note also that the theorem does not state the converse: that the two agents with different behaviors cannot have the same subjective experience (this need not be so and anyhow is not relevant for our further reasoning). Another remark is that behavior should be here and throughout the text taken in a broader sense, which includes not only obviously visible behaviour, but also responses to stimuli in general: in a hypothetical case of a patient with impaired motoric abilities who could process sensory data but would require NMR scan of the brain to somehow read out responses (e.g.\ yes or no answers to our verbal questions), these measurable brain responses would also then fit into our definition of behavior.

While we have provided, above, some basic arguments for the Theorem 1, in the Appendix we give a more elaborate ``proof", arguing also that its conclusion is almost unavoidable even in a general setting (i.e.\ without relying on the Neuroscience postulate and, in particular, for a certain class of property dualistic hypotheses). Besides, we also note that, to our knowledge, there are hardly any cognitive scientists or philosophers that adhere to the neuroscience postulate and yet deny the conclusion of the Theorem 1.\footnote{A likely exception is J. Searle, though even he seems to occasionally imply necessity of certain behavioral differences, as, for example, in his review \cite{SearleOnRebellion} of N.\ Bostrom's book ``Superintelligence: Paths, Dangers, Strategies''. There he concludes that there is no risk of AI rebellion since AI is unable in principle to ``engage in motivated behavior'' -- by which he implicitly but obviously admits unavoidable behavioral differences.}

The statement of the Theorem 1, combined with the mechanical-deterministic assumption for the brain functioning, has an immediate consequence. Namely, causally closed mechanism of the brain organ can be in principle simulated on a powerful enough classical computer to arbitrary precision. This means that given the current brain state and the sensory inputs, brain responses can be, in principle at least, calculated on a standard type of computer. In turn, this means that a computer can, in theory, predict the same behavior that would be the result of the brain functioning. Given the appropriate artificial body, and fed with sensory perceptions that the body receives, a powerful classical computer (based on silicone or any computationally equivalent technology) can thus assume the role of the brain and control an agent that is behaviorally indistinguishable from a human, in the context of Theorem 1. This leads to:

\begin{theorem}
Under the NP premise, a sufficiently advanced android controlled by a classical type of computer may possess the same level of subjective experience as a human.
\end{theorem}

In the Appendix, we give a more formal ``proof'' of the theorem, and discuss some generalizations (e.g.\ to include also stochastic processes in the brain).

We ``derive'' also the following two theorems that will help us in the later stages:

\begin{theorem}
(On temporal correlation): Subjective experiences (or illusions of) that accompany dynamical evolution of a certain system (e.g.\ brain) are synchronized with that evolution, and do not depend on evolution of the system either before or after the considered period of the dynamical evolution (apart from due to the contents of memory).
\end{theorem}

In other words, if we consider a brain activity within a given 5 minutes period, there will be (assuming the brain is conscious) the corresponding subjective experience unfolding during that same time, whose content and quality do not depend either on the past or on the future of that brain, apart from, possibly, due to recollection of previous events. For example, if the physical brain in that period receives and processes stimuli that correspond to pleasure (or pain), it is during that same time, i.e.\ simultaneously, that the pleasure (or pain) is subjectively experienced (or illusions of). Also, smashing the brain with a hammer after these 5 minutes or not, will not influence the prior experience. While being a simple corollary of the Neuroscience postulate and the behavioral theorem, and something that is generally taken for granted, it still deserves to be mentioned. Since subjective experience is not objectively measurable, without the physicalist NP premise this synchronization is an additional assumption (which, as argued in the paper \cite{ToTheRescue}, is not necessarily warranted in a more general, e.g.\ QM setting).

\begin{theorem}
(Repetition theorem): If a physical system that gives rise to consciousness undergoes identical dynamical evolution more than once (fed by the identical stimuli and starting from the same initial state each time), each time the same subjective experience arises. Alternatively, if another system identical in every detail to the first one, undergoes the same dynamics, both evolutions must be accompanied by identical subjective experiences.
\end{theorem}

This seems to be unavoidable: either as a consequence of the previous theorem or noting that if the consciousness is a consequence of system dynamics, the fact that a period of identical dynamics in the same system state has already occurred before cannot preclude or influence the emergence of the same experience again. Same goes if we consider a different system, but identical with the first in all relevant properties. Nevertheless, in the Appendix, we will also address some non-standard views on this line of reasoning.

Now we proceed to establish further equivalence relations between different systems, proceeding in small steps.

We take as a starting point an android satisfying presumptions of Theorem 1, named Data (as we already discussed, unlike the Star Track Data, this one is indistinguishable from the humans also when it comes to expressing emotions). By Theorem 1, Data belongs to the same consciousness class as humans. From Theorem 2 it follows that Data's CPU can be, in principle, based on some standard, e.g.\ silicone, computing. This CPU receives inputs from Data's sensors and body parts, does the calculation and feeds back the results to various actuators in Data's body (speech synthesizer, artificial muscles, etc).

\begin{lemma}
The physical location of the CPU controlling the behavior is inessential, as long as its function is not diminished.
\end{lemma}

This hardly requires special proving: it might be in the android's head, but also it could be localized in some computer nearby which is (e.g.\ by some WiFi equivalent) connected to to the rest of the android. According to Theorem 1, this should not affect its level of sentience, as it would not influence its behavior.

Next, we take a slightly bigger step:

\begin{theorem}
(Simulation theorem): Simulated android (or androids), existing only in a virtual reality realm simulated on a classical computer, can be, in principle, programmed to be as conscious as humans.
\end{theorem}

The outline of the proof is as follows. Just as one android can be remotely controlled from an external CPU and be as conscious as a human, so can be two or more androids controlled from the same computer (or a powerful enough computer cluster). Furthermore, sensory input arriving at the computer that governs androids' behavior is certainly digital from some point on, and thus, in principle, can be simulated, i.e.\ replaced by computed input that would correspond to a computer-generated virtual surrounding. Also, the behavioral feedback that android bodies provide to the environment (by their motions and sounds they produce) can be taken into account when computing the evolution of the virtual surroundings. It can be also imagined that the androids, until some initial moment $t_0$, possess actual physical remotely controlled bodies and receive environment information from real sensors in their bodies, but that after that initial moment real sensory input is replaced by a simulated one (virtual surrounding is generated to closely mimic the real surrounding at $t_0$). Even if the androids would be able to tell the difference after some time (e.g.\ by noting some artifacts of the simulation, due to its finite precision), this switch of the source of sensory input has no way to immediately influence their level of consciousness (at least under the Neuroscience postulate). A somewhat more detailed argument can be found in the Appendix.

Therefore, we must conclude that under NP, also android(s) entirely simulated in a virtual reality can have qualitatively the same level of consciousness as human beings. In spite that there are no more any physical entities around, apart of some information being processed inside the computer, the subjective experiences of the simulated androids would be no less real (or no lesser kind of illusion) than our subjective experiences in flesh and blood. To visualize events in the simulated realm, we could attach a display to the computer and have it render some parts of the virtual reality on it. In this way we can observe the android(s) ``inside'' behaving as conscious beings, and still reporting their subjective states. Of course, presence or absence of the display cannot influence subjective experiences of the androids (according to NP our observation is certainly inessential).

We note that a conclusion that simulated beings in virtual reality can be conscious is nothing new -- moreover, it is part and parcel of all variations of the ``simulation hypothesis'' \cite{SimulationHypothesis} and alike. But we had to go through this reasoning carefully, to show that under the presumption of Neuroscience postulate these conclusions basically turn out to be a logical inevitability, rather than any arbitrary belief.

The ``hypothesis'' commonly conjectured in this context is that we are actually living in one such simulation. But, unlike most of the proponents of that idea, we are next to consider problematic aspects of such and similar hypotheses. Namely, as we did the job of isolating the agents from the outside world, deeper problems already begin to get more palpable, at least a bit.

For the sake of concreteness, let us assume that inside such simulation there is the android Data, existing and living through his subjective experiences. Let him be at first content, quite tranquil, pondering about something pleasant while calmly wandering through the simulated environment. All of a sudden, while passing through a door in the virtual reality, the door unexpectedly gets shut (e.g.\ due to wind) and his finger gets severely slammed. Momentarily, the feeling of content is replaced by the sudden shock and shortly replaced again by the feeling of excruciating pain in the finger. We remind that this Data was programmed to have emotions and feel and react to pain in the same way as a human -- in particular, his cognitive units might work by simulating the functioning of a human brain which is exposed to the same stimuli. If required for having experience of pain, we may simulate also the reaction of the entire human nervous system when the finger is exposed to extreme mechanical pressure. In this way, we ensure identical behavioral reactions, which, according to Theorem 1 and under NP assumption means that also some corresponding subjective experiences must be (qualitatively) there. (After all, if for any reason someone finds relating an android and pain too stretched, the same argument could be made by considering any other type of subjective experience/quale, the choice of pain makes the conclusions only slightly more intuitive.)

Therefore, there inside the computer, must ``somewhere be'' that feeling of content, of shock and that familiar feeling of the pain in the finger. In a similar sense as when any human experiences these feelings -- if the latter is an illusion, then so is the former. Even if the qualia are not experienced identically, there still must be a couple of distinct subjective experiences (or illusions of), emerging in a succession. Besides, there ``inside'' is someone that has experienced the sudden closing of the door, who saw them close, who has subjectively perceived that. And yet, there inside the computer, we objectively find nothing but sequences of binary zeros and ones, evolving and changing according to certain rules. How comes that this familiar feeling of pain can appear in the processing of all those zeroes and ones within that computer? And, even leaving the feelings (and qualia in general) aside, do these zeros and ones really have interpretation per se, do they really have that unique meaning of a human-like android Data that has just hurt his finger? Where is objectively that door inside? When there is no longer our interpretation of the bits in the computer as this virtual environment with Data (which disappears when we remove the monitor), does this information inside still, without us, external observers, really and objectively has this same meaning, this and only this interpretation? When the monitor is on, it is us who endow with the meaning the endless streams of zeroes and ones in the computer registers: the conversion of the stream of data into the pixels on the monitor is tailored to translate the information to be readable to us, and even there, that pattern of differently colored pixels on the screen, does it have any meaning per se? Does it really have the distinctive meaning of the Data feeling pain, and the door being shut, even without us who correlate colors of the pixels on the monitor to form the picture of an android in our heads? One should take care not to unintentionally infiltrate the syntax of the bits with the meaning that we ourselves provide to it, if we are to investigate the possibility of ``self-emergent'' meaning. And furthermore, even if we accept that these zeros and ones have the unique meaning of Data android in his virtual world, is that enough that these zeros and ones become alive and feel the hurting finger? A book has meaning (at least the one that we give to the written symbols), but we do not usually take that the characters in the plot really experience the narrated events, neither as the book is written, nor when it is read.

Intuitively, we get a glimpse of the problems, but our intention is not to rely on intuition. Instead, we proceed in the manner as formal as possible, to make the problems more obvious and harder to deny.

Indeed, at this point, it is not at all yet clear if there is an actual problem, at least if there is an unsolvable one. The gap that some believe to exist between the binary digits and the subjective experiences might be just a consequence of wrong intuition and over-simplification. Unlike the static symbols in a book, we here have a dynamical computation process of immense complexity, ongoing on an advanced supercomputer of huge capacity and extreme speed. There is a widespread opinion that the combination of such complexity and delicate structure can bring about the first-person experience as well.

Thus, we proceed further in establishing the chain of equivalences. We start by noting that the supercomputer performing the computation need not be fast at all, only of a sufficient memory capacity:

\begin{lemma}
The computer running a closed simulation that contains conscious agents can be arbitrarily slow.
\end{lemma}

This is an important direct consequence of the fact that the entire simulation is completely isolated from the outer world, resulting in that agents within can perceive the passage of time only relative to other events in the virtual reality. Overall slowing down of the pace of simulation (as seen from the outside world) cannot be detected from the inside.

No longer requiring calculation speed, we can choose computer architecture much more freely (according to Theorem 2, the technical realization was not essential). We decide to run the simulation on one from the variety of universal Turing machines, with a binary alphabet \cite{TMgeneral}. It has been almost a century of how we know that any computer calculation, even those usually associated with the most powerful computers imaginable, can be performed just as well also on a very long tape of check-boxes which can be either ticked or unticked by a ``head'' that moves over the tape one step at the time, following a given set of instructions (program). In computational theory, the tape is usually of infinite length, but we are concerned with a well defined finite problem: running the simulation of the android(s) and virtual environment for a finite virtual time (e.g.\ one hour). As this is doable on a super-powerful but finite-resources computer cluster, it is also doable using a finite length tape.

As Turing machines are abstract concepts, we will also pick a concrete realization. We will imagine a strictly mechanical wooden Turing head\footnote{Much like the ingenious craft described in \cite{WoodenTuringMachine}.} that reads and writes ``X'' and ``O'' symbols in the sand (e.g.\ on a huge stretch of beach or in some desert). For example, the head can be wind propelled (though this is completely irrelevant). The symbols are equally spaced, imprinted deep enough so that machine can also read them by sensitively ``touching'' to check the depth of the sand at a few points. With a universal Turing machine the program of the simulation is also embedded in the sand symbols (much like the compiled source code in a conventional computer is stored in the same physical memory as the data) and in this way we maximally reduce the complexity of the head itself, while effectively all relevant information is encoded in the sand. In the beginning, a row of ``X'' and ``O'' symbols that represents the simulation program together with the simulation initial state (state of the android and the virtual environment at some given moment) is imprinted in the sand, after which the wooden head is positioned and activated.

As this setup is computationally equivalent to the supercomputer performing the simulation (apart from the processing speed, which is irrelevant according to Lemma 2), the same level of subjective experiences must emerge now as when running the simulation on silicone architecture. This conclusion is summarized in the following theorem (some more formal technical details of the lemma and theorem ``proofs'' are presented in the Appendix):

\begin{theorem}
A system comprising a pattern of X/O symbols written in sand and a wooden Turing head mechanism performing certain automatic operations (of ``reading'' the symbol beneath, flattening out the sand, inscribing a new symbol and moving left/right one step) according to the instructions embedded also in the sand, can contain conscious agents living through human-like subjective experiences (i.e.\ agents belonging to the same equivalence class of consciousness level as humans).
\end{theorem}

The reasons behind this switching to sand-symbols Turing machine are twofold. Firstly, the word ``computer'' (especially ``supercomputer") is nowadays quite emotionally charged: as ever more powerful computers/smartphones enter all spheres of our life, we are left with an intuitive impression of their omnipotence, at least in principle (especially omnipotence of hypothetical futuristic computers). It is hard to keep in mind that, in regard of the ability to perform calculations, they are in principle no different from the first computers ever made -- the only advances were quantitative, in speed and capacity. Restating the problem in the form of sand symbols and a wooden head we see as a sobering way to counter these prejudices, especially in the case like this one where speed does not matter. This is also the reason why we opted for a purely mechanical wooden head -- any involvement of electrical devices may unintentionally lead to mystifications.\footnote{Alternatively, instead of a wooden head, we could imagine a middle-aged fat warehouse worker, unshaven, with a cigar in the corner of his mouth, strolling down a huge beach and unenthusiastically carving by his finger X and O symbols in the sand, according to a table of instructions that he carries along. Or, in a tribute to Douglas Adams, we might entrust the same job to two well trained white mice, though it does sound less realistic.}

The second reason is far more important. Though we have now stripped away illusion that it could be the advanced architecture/structure of the computer, combined with its power expressed in billions of operations in second that might lead to the emergence of consciousness, there still remains the mystification that computational process itself can somehow bring about the emergence of the subjective experience. However, once we have gone from relatively abstract electronic circuits to simple motions of some amounts of sand, we can now closely follow the computational process in every mechanical detail. In this way, we set the stage for the next step aimed to show that relying upon ``computational process'' to induce consciousness is no less illusory (at least within the realm of classical physics).

A few comments beforehand. If the equivalence of the two ``computers'' (the silicone one and the sand one) is too counterintuitive to grasp it, we can again resort to a variant of a ``computer display'' aid. The simulation program can be written so as to reserve one million symbols (at some preselected positions) for a ``screen". The algorithm will be such that, when we arrange that million of X and O symbols in a 1000 x 1000 matrix, we get a rendered picture, a ``monochromatic'' snapshot of the current state of the simulation, e.g.\ as a camera focused on Data would see. It is not quite millions of colors, but the resolution is fine enough to glimpse on how Data fares ``inside". The algorithm can even ``refresh'' these million symbols each tenth of a second of the simulation time, so, if we are patient enough, we may follow Data calmly roaming around, responding to stimuli in the virtual world and behaving in every respect as a sentient being (the corresponding real-time refresh rate might well be one frame in a thousand years, but that does not take away from the argument). Just as before, a few virtual minutes into the simulation, after a quite while of writing/rewriting of the sand symbols, the Data's finger will get smashed.

And the same intuitive problems arise again, yet even more pronounced. Simulated Data must feel a lot of pain, just as would a human in the same situation, and that familiar feeling of pain has to originate somewhere there, on that beach. He will also have a stream of thoughts, and that inner dialogue must be there too -- thoughts having a very precise meaning, about the pain and the slammed door -- thoughts that have nothing to do with the sand or the X/O symbols. Are the scene of Data's finger caught in the closing door and the qualia of his feelings, really embedded objectively there in the endless X/O sand symbols, irrespective of our interpretation?

But first and foremost: what is, on that huge stretch of beach, with no one else there but the mechanical wooden head and the immense number of X/O symbols, feeling the pain? Is it the symbols? Could hardly be the head: it has only a few internal states and is comparatively very simple (as we discuss in more detail in the Appendix, some universal Turing machines can have as little as two internal states, and some other only mere 22 instructions overall). However, a determined physicalist relying on the Neuroscience postulate probably will not flinch: the consciousness and the pain are emergent (or illusionary) properties of the system, so it is the system comprised of symbols in the sand and the wooden head, combined as a whole, that feels the pain (or has that illusion) during the dynamical process of the computation. Whatever that exactly means. But it is far from clear what this should mean, especially in the context of classical physics: ``system as a whole'' is an abstract and at best an emergent property, just as the ``pain'' or ``feeling of red'' must be under the NP. Real should be only matter and its motions (interactions). There is nothing that ``holistically'' binds the wood and the sand together in a system, essentially no physical interaction apart from that reading and imprinting (i.e.\ apart of some sand motions caused by the head), as no kind of miraculous ``entanglement'' can occur in the world of classical physics.

This is where we turn to the final step of our no-go theorem. We intend to exploit the fact that Turing machine computation, being a mathematical idealization, is not something which exists per se in the real world (especially not in the world of classical physics). It is merely our invented name for a real physical process -- in this particular case, the process of relocation of certain amounts of sand. Therefore, it is pointless to explore how the running of a Turing machine can cause consciousness. The goal is to try to explain, if possible, how exactly this moving of some sand here and there can result in a subjective experience of serenity, pain or color. What is essential in these motions that should have such power?

To investigate this, we will first concentrate on the dynamic of the wood-sand computation process. Indeed, so far, symbols are frequently changing during the computation: some are being erased and replaced with others -- and that gives an impression that something is ``alive'' here. But this ``dynamics'' turns out to be less essential than it seems. To show this, we will slightly modify our Turing machine. Instead of replacing the symbols by new ones in the same row, the machine head will now write each new state of the ``Turing tape'' one row below the previous (we deal with the technical details of this modification in the Appendix). While certainly not being optimal with respect to the computation speed and spatial (i.e.\ sand) resources, in this way we attain the following: nothing is any more deleted, no symbol is ever replaced and the symbols in the sand no longer represent only the current state of the Turing machine, but also we now have, in the rows above, the entire log of the calculation up to that point. (In this operation mode the head merely copies basically all of the symbols into the next row, apart possibly from the one that needs to be modified -- instead of repeating that symbol, the new value is written.) We can embed in the sand also the current state of the head so that each line contains all relevant information. Obviously, nothing is significantly changed by this modification that could diminish the ability of the entire system to compute and to elicit subjective experiences. We thus let the simulation run in this way for, e.g., an hour of the virtual time -- long enough that Data has suffered the smashed finger and went through all discussed transitions of his subjective experiences.

A number of new conceptual problems now arise from the fact that there, in the sand, in this ``entire log'' setup we have not only the state of X/O symbols that encode the present state of Data, but also all the previous states. What is Data then experiencing at any given moment? If we assume that the simulation has just progressed up to the moment when Data smashed his finger, he must feel that shock, and pain. But a few million rows above, there is another line, encoding the same Data who feels content and joy. If we discard somewhat ridiculous possibility that each row in the sand gives rise to one Data that lingers frozen in existence experiencing that instant (it is easy to identify many inconsistencies if one is to follow this line), then what is Data actually feeling (or having that illusion of feeling) at the moment, as it should be just one thing? One is tempted to say that it obviously must be the last row, as, after all, the position of the ``head'' is there and the process of calculation is taking place there, while the rows above are just the system log.

But, the situation is actually much less clear. What we know ``for sure", is that if we flatten out all sand symbols but the first row, and position the head again at its initial place, then the Repetition theorem guarantees that the same simulation, with the same subjective experiences of the agents inside, must happen once more, identical in all details. But what if we forget to erase the old symbols? We just reset the head, that is, put it on the initial position and reset its internal state (so-called ``m-configuration") to the first instruction, but we do not explicitly erase the ``memory". The head is now moving the same way as the first time, mechanically carving symbols in the sand in the row below -- only that this time it does not change anything in the row below, all same symbols are already there. There seem to be essentially two options.

The first one is that the subjective experiences will reappear again as the head is moved to its initial state and its voyage restarted. But what is so special with the few winded pieces of wood that we call the head, to produce this effect? If we make another identical head and put it at any other of the rows, then that another head must produce the same effect. In such a case, there must coexist subjective experiences of Data who is at the moment experiencing joy and content (emerging due to the head moving along the first row), and of Data who is going through the stress of hurting the finger (due to the motion of the second head, many rows below). Of course, we can add more heads, but this is not a problem by itself. The tricky part is to understand what in that head would have this power, since it is doing nothing?! The head is only traversing the rows, rotating its gears and changing its internal state. The internal state is just the ordinal number of the current instruction, indicated by the orientation of some gearwheel, and the overall number of instructions isn't necessarily that huge -- maybe a couple of dozen. It is not truly calculating even its ``m-configuration", i.e.\ it is not generating truly new information about its internal state as it progresses -- this information about the ``m-configuration'' is already written in the sand. Head is not moving any sand, though it is reading it by sensitive touching, a sort of as a blind man is reading Braille alphabet. The clicking of the wooden parts, directly correlated with the sand symbols but ending nowhere and affecting nothing, does not seem likely to produce consciousness. Otherwise, it is very difficult to motivate why then, a blind man, or a seeing one, would not elicit the same effect of bringing the subjective experiences of Datas into existence, also by tenderly touching rows of X/O in the sand? And if that sort of Braille reading of symbols carved in the sand would bring a consciousness of someone into the being, the proper name for that would be simply -- magic. Besides, if we have established that it is the ``reading'' of the symbols by the head that matters, we might now replace the ``reading'' mechanism of the head by some optical instrument instead of the mechanical ``depth detector". The same reasoning would lead to the conclusion that even optically reading (without disturbing) the sand symbols would induce the subjective experience of the encoded agents. But, since observation plays no role in classical physics, this would further mean that simply having the symbols written somewhere makes Data alive (while it might be magically necessary to point to and follow the symbols by a finger, just as the head sort of points to symbols in its motion).

The other option is that having the head move through the lines but without any effect on the sand would not cause any conscious existence of the simulated agents. Not until the head finally reaches the last row again and starts carving symbols in the sand anew, which then causes lives and experiences of the Datas to resume. In this view, writing of the symbols by the head must be essential -- that is, the ``motion'' of the sand, instead of clicking of the head's gears and switching of its internal states. Thus, in both the previous and in this case, writing of all symbols in the correct order would cause the emergence of the subjective experiences.

And, in principle, though seemingly weird, there is still the possibility that the sand motions per se are not enough, even if they completely and properly mimic motions during the computational process -- they maybe need to be produced by the the wooden mechanism ``clicking'' in a proper way, for the consciousness to emerge.

To explore this already pretty stretched assumption, and to complete our our analysis of the previous options, we consider another run of the identical computation. A few rows below the result of the previous run, we start a fresh simulation: we again manually initialize the first row in an identical manner as before and then let the head do its deterministic computation anew. Of course, it will again produce exactly the same output in the sand as before (the new symbol rows positioned slightly below the last row of the first run). Furthermore, as the result of the Repetition theorem (Theorem 4), the same subjective experiences must arise again in this run, just as they did in the first run. Besides, according to Theorem 3, the abrupt termination of the simulation cannot influence the subjective experiences that must have emerged simultaneously with the computation.

Finally, we again consider running the second round of the simulation, but this time we would be using a slightly different wooden head: instead of a Turing head that calculates the next sand row based on the previous row, now we envisage a simpler head -- one that generates the next row by merely copying the corresponding row of the previous simulation run (it goes ``up'' a few billion rows, reads the symbol, returns back and imprints it at the proper position). Our claim is that the ``copying head'' must produce the same type of subjective experiences as the ``computing head". In slightly more general terms the statement can be formulated as:

\begin{theorem}
(The no-go theorem). The Neuroscience postulate implies that a human brain cannot have any greater ability to induce subjective experience than a process of writing a certain sequence of digits (writing can be done in an arbitrary way and on the arbitrary surface).
\end{theorem}

Since this final step in our chain of theorems probably goes against some widespread expectations, we devote to it a lengthy and detailed ``proof'' in the Appendix. The essence of the argument is in the following. First, we note that the objective effects of the operation of both heads (computing and copying) are identical, at least as far as the motion of the sand is concerned. To analyze this in detail, and to account for a possible difference in produced experiences induced by different functioning of the head mechanism, we consider a series of ``hybrid'' heads, that incrementally bridge the technical gap between the computing and the copying head. We discuss each of this incremental steps and argue that there is no any objective element which could account for any possible qualitative difference in emergent subjective experiences (i.e.\ there is nothing that could explain why the operation of one head would produce the consciousness and of the other not). Note that we, of course, do not claim \emph{computational} equivalence, in the mathematical sense, between a calculating Turing machine and one doing a copying process -- such assertion would be obviously false. Our conclusions here stem from the fact that mathematical idealizations such as ``Turing machine'' or ``computing'' do not exist as such in the universe of classical physics, and that all there is are motions of some sand and some wood -- and we compare the potential of these motions to elicit subjective experience, in these two cases. In the Appendix, we also offer a few additional variations of the argument.

The remainder of the logical steps towards the conclusion of Theorem 7 is fairly obvious. Once we have ascertained that the copying of the log of the previous run has no lesser potential to produce the consciousness than ``computing'' of the sand-symbols again, it is also immediate that: i) copying the proper X/O sequence on any surface by any means would necessarily elicit the same subjective experiences (the entire chain of theorems could be repeated by using the new surface and a different head); ii) by putting the conclusions of all the theorems together, we find that a human brain, under the NP premise, has no greater ability to produce consciousness than the copying of the pattern, which, in turn, can be interpreted as a huge number (rows of symbols O and X can be concatenated and replaced by digits; besides, a Turing machine with a larger alphabet could also have been used, resulting in more than two digits).

Taking a sober look at the last dozen of paragraphs (plus taking into account the detailed ``proof'' in the Appendix), it might seem strange that we had to devote this much space to an analysis of the hypothetical side-effects of a piece of wood moving some sand. Naively, one could expect that everybody will agree that a moving piece of wood can cause some moving of the sand and that anything beyond that can be only weird imagination. Why waste time? Surprisingly, there are too many people who are ready to take very seriously the prospect of this piece of wood generating consciousness, vivid perceptions, and feelings of some agents simply by moving this sand. To address this belief equally seriously, we had to go through all these details and demonstrate that, if the belief is correct, then also writing of a specific huge number must have the same strange power.

We now turn to the discussion of no-go theorem implications, first for the Neuroscience postulate and in the context of cognitive science, and later on also in the context of physics.

\subsection{Discussion of the no-go theorem and of its cognitive science implications}

In very general terms, there are a few positions that one can assume with respect to the result of the no-go theorem.

The first possibility is to dispute the theorem conclusion. To that end, one (or more) of the logical steps (i.e.\ theorems 1 through 7) leading to the final result must be invalidated.

The conclusion of the Behavioral theorem (Theorem 1) -- that agents which behave indistinguishable in principle must possess a similar level of consciousness -- is, in general, a source of much debate. However, it must be kept in mind that in Theorem 1 this inference is reached under the premise of Neuroscience postulate. While we believe (and argue in the Appendix) that there is a strong case, even in general, for the plausibility of the behavioristic conclusion, disputing it in the context of NP would certainly be very hard. Any such attempt must carefully address the arguments in the Appendix.

To refute the Theorem 2, one needs to deny that the brain can be simulated even in principle, despite that Neuroscience postulate is assumed. It is not clear how this could be supported for a system of the size of a brain, whose relevant part of dynamics is governed by classical physics (Neuroscience postulate). Even putting aside the option to simulate physical dynamics of the brain, one must also argue that it is in principle impossible of to develop sufficiently advanced general AI based on classical computing -- an assertion which seems unsupported by the current progress in that area.

We have no idea how Theorems 3 and 4 could be disputed in any truly relevant manner (within the classical physics paradigm).

Theorem 5 is peculiar, as we believe that at this very link the logical chain breaks in a crucial way when we switch from classical physics to quantum mechanics (as we shall discuss in the third part of the paper). However, as long as we stick to principles of classical physics, insisting on determinism and definiteness of matter properties irrespectively of the observation, we see no possible argument why the simulated input from the sensors could be relevantly different from the actual one.

Contesting Theorem 6 would require an explanation of why two mathematically equivalent realizations of a computer would differ in the ability to produce subjective experience, giving a special privilege to electronic (silicone) architecture over some others. We have no idea how this stance could be plausibly supported.

Finally, we went at great length in the Appendix to argue that there cannot be any relevant difference between the computing and the copying head, differing only in very subtle correlations of head and sand movement. It must not be lost from sight that we were not proving computational equivalence of a Turing machine and a copying machine -- these two are obviously computationally inequivalent, but that is completely irrelevant for our discussion. Instead, we only showed that motions of the sand caused by one particular design of a certain wooden machine have no more potential to give rise to subjective experiences than the same motions of the sand caused by a slightly modified design of the same wooden machine. To challenge Theorem 7, one should propose what of these delicate differences in motions of the sand and wood is exactly responsible for the emergence of subjective experience and how.

None of these seems to us very promising. While it cannot be excluded that we have overlooked some additional possibilities (and finding loopholes in the theorem might give important insights into the phenomenon of consciousness), we should now concentrate on options that seriously acknowledge the result of the no-go theorem.

The first such option is to accept that both the human brain and the process of printing of the appropriate huge number do elicit subjective experiences, and of the same type.

Given the huge number (that is the simulation log), it is far from obvious whose subjective experiences and with what content such number represents. The hypothesis that printing of the number causes the experiences to emerge and to be subjectively lived through, immediately implies two things:

i) There is a well-defined mapping between the numbers and the contents of subjective experiences\footnote{One could argue that the existence of such mapping already presumes sort of property dualism. While avoiding to get lost in classification issues, it is our position that the no-go theorem also poses a problem to such philosophical views, if they are dualistic to matter governed by classical physics.} (i.e.\ the meaning of the number is objectively there, encoded in the number);

ii) The writing of the number causes these particular experiences to become subjectively experienced by encoded agents, and that should happen merely as the consequence of this writing process.

One of the problems evident in the second of the two implications is its strong magical connotation. It points that mere writing ``brings to life'' the subjects weirdly (magically) encoded in the number, out of nothing. Indeed, it is hard to avoid analogy that such a number is then some huge equivalent of ``abracadabra'' which summons spirits into existence. Certain numbers would then encode and thus bring to life subjects going through great pains and agonies, and it would be highly immoral and unethical to write any such number down. On the other hand, writing some other numbers would bestow immense pleasures to agents that would spring into existence, as the digits are laid out on surface -- and writing such numbers as many times as possible would become a sort of a moral imperative.

If one is nevertheless willing to further pursue such hypothesis, there are even more serious problems awaiting.

First, what sorts of writing (i.e.\ copying) result in the emergence of subjective experiences? Is reading of the number aloud sufficient? It should be, since it is effectively ``imprinting'' the number in sound waves. In principle, we could devise a (technically demanding) Turing machine that operates on sound waves instead of in the sand and then repeat the proof. The magical analogies with chanting of the correct spell to summon the spirits would be even stronger. Actually, if we follow this route, we must acknowledge that it is (repeating of) the information that causes consciousness to emerge. And allowing information per se to have a certain effect, not via any physical mechanism, is usually considered as a sort of ``magical thinking".

But the problem actually goes even deeper and is related to vague definitions of information and of the copying process in classical physics, which are necessarily related to some (arbitrary) level of coarse-graining. For example, if we write the number in some reflective ink, turn off the lights, then take a flashlight and illuminate a digit by digit in the correct order, it is reasonable that the subjective experiences should arise again -- by this process, we have copied the number into a sequence of light signals. Alternatively but to the same effect, we could, instead of pointing the flashlight from digit to digit, briefly turn on an omnidirectional light source positioned to the left of the first digit. The geometry itself, combined with finite speed of light, will ensure that the digits become illuminated in the correct order and thus this light will accomplish the same effect as the flashlight. If so, then a flickering light will induce numerous experiences of the encoded agents -- each time the light goes off and on again, they will relive their hour of simulated time over and over again. But there is even no need for the light to flicker. For example, if it only changes the color every now and then, the new round of subjective experiences should frequently emerge since the photons (or light waves) of new color are copying the information, and these should have nothing do with the photons of the previous color. Proceeding in this way, only slight variations of the light frequency should be sufficient to each time produce a new instance of agents emerging together with their subjective experiences. (It is possible to repeat the same reasoning with reflected sound instead of reflected light -- if one objects that due to the theory of relativity a Turing machine operating on reflected light signals might not be feasible even in principle.) On the other hand, a question of ``how many times has Data gone through the experience of slamming his finger'' should have a well-defined answer (we are dealing with classical physics, not many-worlds theory). But the answer seems to depend upon arbitrary coarse-graining for the light color -- how small variation is sufficient to make another instance of agents and their experiences emerge.

These issues are further exacerbated if we use some more common configurations of sand (or soil) instead of X/O symbols: for example, a small heap of sand could serve instead of ``X'' and shallow hole/depression instead of ``O". And now we consider billions of billions of solid bodies throughout our Universe, with myriads of heap/hole patterns on them, illuminated from various directions and in an unsteady manner. When some of these appropriately align and happen to produce the correctly evolving pattern, will some Data, or someone with experience of one of us, pop in and out of existence for a split of a subjective second? Or, why take into account only heap and hole patterns -- anything that can carry information should suffice, our choice of particular coding must anyhow be irrelevant. On the other hand, some researchers already take seriously the Boltzmann's brain problem\footnote{Estimate which shows that it is statistically far more probable (under certain reasonable assumptions) for randomly distributed matter to form, by mere chance, an isolated brain with an illusion of the external reality that we see than it is the probability that such external reality indeed exists.} in its original form \cite{BoltzmannBrain}, and here we have its analogon on steroids, making it almost certain, under these premises, that our subjective experiences are actually elicited by a flickering light illuminating a seemingly random pattern in a universe close to a thermodynamical equilibrium, rather than being products of biologically evolved creatures.

Problems of this type arise because the information has no meaning per se in the classical physics and instead becomes meaningful only in the context and through some process for which this information is pertinent, as these naturally dictate what is the relevant coarse-graining level. Ambiguities in defining the information are usually no reason for concern in classical physics paradigm: information per se does not affect anything, so we are free to define it as we like (suited to context). Normally, within this paradigm, ``information'' is only a helpful emergent property itself, not something that has inherent ontological meaning (same as ``Turing machine'' or ``computation"). Attempts to attribute independent meaning to information separated from any context, and, moreover, to attribute nontrivial effects (such as causing the emergence of consciousness) to the information alone and by itself, likely results in logical inconsistencies as the ones just considered.

And, even if we put aside the problems of coarse-graining and of the definition of information, and take for granted that a well-defined string of bits (which is a mathematical abstract) can exist in the world of classical physics, yet another layer of potential difficulties arise. It is the old conundrum of whether syntax alone can produce semantics (meaning), just on its own, per se. Can a sequence of bits possess a well defined unique meaning, irrespective of the external interpretation and observers? Does this huge number inherently carry this one and only one meaning of the Data's finger getting hurt, five minutes into the simulation? To probe this question in our framework, we can do the following short thought experiment. We may imagine that no androids ever existed, and no fingers, and no pain, and that even this four-dimensional universe has never existed. That the only thing that had ever existed was a one dimensional (time)line, enough to support a string of bits (represented in whichever way) randomly appearing. Sooner or later, the proper binary number (i.e.\ the proper X/O sequence) will appear -- is it really plausible that, at that moment, the concepts of the androids, colors, pain, and of the four-dimensional space-time will all of sudden appear and spring into existence as the specific bit pattern appears? Or is it more likely and plausible that the sequence of bits in that one-dimensional world will remain what it is -- just a sequence of bits?

We could easily extend the listing of logical complications and inconsistencies tied to any attempt to assign objective mapping from the numbers to the subjective experiences, and especially to the claim that copying of a certain number causes the corresponding experiences to be actually lived through. However, taking that these few arguments were sufficient to render the whole idea absurd, we now we turn to the option that settles all these issues in the logically simplest way.

\subsection{Neither is conscious}

Indeed, by far the easiest and the most parsimonious way to avoid the previous logical mess is to deny any ``emergence'' at all. When the digits are printed, they are printed and that is all that happens -- it could not be more obvious and simple. The motion of the wooden head causes motion of the sand, and that is all to it. It trivially solves all the riddles posed by ``emergence'' hypothesis: there is neither need to explain what exactly emerges, nor what motions of sand and wood make it emerge, nor when it emerges, nor how many times it emerges. Anyhow, it is very unlikely that a consistent set of answers can be given to these questions, even in principle. In this way, there is no need to invoke any magic of agents with subjective experiences popping out of thin air, only because someone has uttered the correct chant.

But, the price to be paid is that, by the equivalence established through the no-go theorem, nothing can emerge also as the byproduct of the brain functioning. The brain is still perfectly capable to produce behavior -- i.e.\ basically the motions of our muscles. External stimuli, which are in the essence also some types of motions (either of photons, or of air in the sound, or of some other molecules/hard bodies in case of other senses), after some processing turn into our responses which are again some motions (either of our bodies or of our vocal apparatuses). No mystery there. Motions can and will cause other motions, but nothing apart from that. In particular, since we deny that printed symbols can evoke subjective experiences, by the logical inference we must deny the same magic to the brain this time.

And here this hypothesis, in spite of being logically consistent and simple, runs into a grave problem -- it cannot sustain the test of an elementary subjective experiment of introspection. There is hardly anything of which we can be as sure as we are confident of having some subjective experience -- be it an illusion or not, it is there. This subjective and yet undeniable observational fact cannot be accounted by the hypothesis that motion causes other motions and nothing else (and that HP holds).

While this conflict with observational data might be slightly obscured when we analyze brain functioning, the established equivalence (in regard of the subjective experiences) between the brain and the symbols in the sand makes this problem far more obvious. There is clearly no ``pain", at least not just like that, in the rows of ``X'' and ``O'' symbols imprinted in the sand (with or without a wooden machine). There is also no ``red'' in the yellowish sand and brown wood. These concepts must somehow emerge, which requires the mapping of the symbol configurations to the experiences -- the option that we have just previously discarded due to too many ``magical'' inconsistencies -- or otherwise these notions simply would have never existed in the world of sand symbols. Sure, there might be a pattern of ``X'' and ``O'' that would correspond to virtual Data virtually reporting ``I feel pain and see red'' in the virtual world. But these would be only some additional symbols on top of the rest of markings in the sand, unless there is some ``internal'' vantage point of Data that somehow emerges and is led into the illusion that he is feeling the pain and seeing the red (which is emergence of subjective experience and the mapping again). Otherwise, red and pain would have just never existed in such a world, and neither would the meaning of the symbols. But we know too well that pain and red, whatever these are, do exist in reality, and that they are not literally certain configurations of sand symbols.

Irrefutability of the fact that we do possess some sort of subjective experience has led some proponents of the physicalist views to seek the refuge in the word ``illusion'' (most notably D. Dennet, promoting, for example, the concept of the ``user illusion'' \cite{Dennett}). Indeed, no matter how logically consistent and appealing by the simplicity that would be, no one would seriously consider a hypothesis that clearly and openly states that there is no any subjective experience, not even an illusion of. It would be in too obvious conflict with results of personal observation (i.e.\ introspection). Importance of our no-go theorem in this context is that it helps us get rid of the illusion that the word ``illusion'' can miraculously save the day.

Namely, the proponents of the ``illusion'' solution always had the problem to explain ``whose illusion it is?", i.e.\ a problem with the subject of illusion. But we are here pointing to another problem, maybe even more severe: there is also a problem of ``illusion of what?", that is, with the object of illusion. Namely, if the sand symbols have the illusion of ``pain in the finger", then nonetheless there must be a mapping of the symbols into this subjective experience and this mapping has to be somehow and at some moment ``activated", i.e.\ the mapped illusion must be at some moment experienced by the agent having the illusion. The mapping tells us it is the illusion of the finger pain, and not of an ice cream flavor. The ``illusion of pain'' is not made of sand, and thus it must emerge and occur in the sand as the symbols are written. And since the mapping is still there, all the magical problems already discussed arise nonetheless. Fact that it is ``the mapping of sand configuration into the content of illusion'' instead of ``into the content of subjective experience'' is a mere renaming of the problem and does not diminish it a bit. But if we decide to consistently apply the clear-cut solution that there is no mapping at all, then we also entirely lose the object of the illusion, and nothing remains, apart from the dead carvings in the sand. And all of us know too well the pain of a smashed finger certainly is not ``nothing".

An additional argument why it would be logically consistent to replace ``illusion'' with ``nothing'' is the following. The essence of the ``illusion-like'' view of the consciousness, is something like this: we can in principle explain everything we objectively see, measure, observe, including also the talk about consciousness -- solely on cause and effect grounds, without need to invoke any ``consciousness'' thing. Therefore it is wrong to artificially introduce any ``consciousness'' concept in the explanation of the world because that would automatically presuppose the existence of non-objective facts. And we do not need the latter for the explanation of any experiment. But the trouble is that the same logic applies also for the concept ``illusion of consciousness". By the same criteria, we no more need the ``illusion of consciousness'' than the ``consciousness'' itself, so then we should introduce neither of these, and stick to ``nothing".

Besides, why is there an urge to replace the plain word ``nothing'' by the word ``illusion'' and make the explanation more complicated? The only possible motive can be: in order to fit some experimental data. But what experimental data -- there is no objective experiment that points to the necessity of introducing any ``illusion'' into the explanation. Insisting on ``illusion'' means acknowledging a new class of experiments whose outcomes do not follow from and cannot be inferred from any combination of objective experiments. Admitting the existence of such independent ``subjective experiments'' is equivalent to admitting the existence of phenomena which are not explained even if we successfully account for all objective events -- in other words, it here means accepting the ``hard problem". Thus, the opponents of the ``hard problem'' should pay the utmost attention not to involve the word ``illusion'' into their explanations and to use the fair-and-square word ``nothing'' instead.

But, hypothetical beings that would experience absolutely ``nothing", not even any sort of ``illusion", are easy to distinguish from humans via experiment of introspection. The former beings are commonly denoted as ``philosophical zombies". The world inhabited by such beings would be very much like a movie projected on a cinema screen, where everybody leaves an impression of a sentient being, but it is only an animated picture, ``no one inside", not even having an illusion.\footnote{It is unclear in what sense such a world could be said to exist -- if it is not experienced by anyone -- apart as a pure abstract mathematical possibility.} And we are quite certain at least of the illusion.

In any case, the ``illusion'' thing must either be ``something", in which case it induces mapping between the symbols and that something, or it is plain and true ``nothing", which is then easily discarded due to conflict with observations. The certain appeal of the word ``illusion'' lies in its ostensible ability to exist in a limbo between ``something'' and ``nothing". But the source of this false impression is easy to track -- ``illusion'' manages not to really exist in spite of not being ``nothing'' only because it is commonly defined as someone's belief (existing contrary to the evidence), or as (a distortion of) our perception of reality. In other words, it exists subjectively, while not objectively, i.e.\ as a difference of subjective and objective. Therefore, strictly speaking, if we use it as a word different from both ``nothing'' and ``something'' then we implicitly recognize the true existence of the ``subjective realm'' different from objective phenomena in the first place. But if our intention from the outset was to avoid this, then we must pick: illusion is either something or nothing -- and we have just demonstrated why neither is acceptable.

To sum up, none of the options that maintain the Neuroscience postulate seems plausible. This actually should not surprise us a lot. The classical physics is basically all about motions (and interactions that influence these motions). It was never clear how exactly motions, connected in a cause and effect chain, could eventually result in a phenomenon such as subjective experience, which is so much qualitatively different from the low-level physical motions in the brain that supposedly cause it. Conventional wisdom is that enormous complexity and delicate structure of the brain organ somehow turn these low-level motions and interactions into our subjective experiences such as pain or red. This always required a leap of faith, as there was never any real explanation of how this qualitative change can occur. What we did by the series of theorems was to put this hypothesis under formal and systematic scrutiny, we followed logical inferences from start to end and finally concluded, based on quite strong arguments, that this faith was never rationally substantiated.

Therefore, we turn to the remaining option: that the premise under which the no-go theorem was derived -- namely, the Neuroscience postulate -- must be wrong. We deal with this possibility and its implications in the rest of the paper.

\newpage
\section{Part II: Neuroscience postullate is wrong}

The negation of the Neuroscience postulate, at its face value, simply implies that the paradigm of classical physics is insufficient to explain the phenomenon of consciousness. Should this truly surprise us, more than a century since we have discovered that all the basic tenets of the classical world picture were deeply wrong? From where comes such confidence, expressed via the prevailing credo of the Neuroscience postulate, that the deepest and yet elementary truths about our Universe must be irrelevant for the explanation of such an elusive, perplexing and profound phenomenon as the consciousness is?

A part of the answer is certainly irrational, reflecting our prejudices about how the Universe should be, the lack of intuitiveness of quantum mechanics and the sharp fragmentation of science in which researchers from one field are often quite unfamiliar with even the basics concepts from the others. However, there are also strong rational arguments that we have already mentioned in the introduction: i) hot, noise-dominated environment of the brain seems unsupportive of manifestly quantum phenomena such as interference and entanglement over distance; and ii) deterministic, cause and effect reasoning of classical physics, as well as artificial intelligence progress based on classical computing, look promising to explain the power of the brain to compute and produce behavior -- without need to invoke quantum mechanics or anything beyond.

Could the mere paradigm shift from the ontology and logic of the classical physics to that of the quantum mechanics, even without invoking any not-so-plausible hypothesis about explicit quantum effects in the brain, close the gap from the explanation of the human behavior to the explanation of consciousness? We believe it can, but this obviously depends upon what this ontology of the quantum mechanics really is. In other words, it depends on the interpretation of quantum mechanics.

We will now discuss a few of the more popular interpretations of QM, in the light of the problem to understand consciousness and with respect to their potential to account for the existence of this phenomenon (for which, as we have argued above, the classical physics seems incapable of). We concentrate only on ``pure interpretations", i.e.\ ones that are mathematically equivalent to the textbook formulation of QM and thus indistinguishable in principle by the standard experiments (we leave to the experimental test the hypotheses which predict new physics).

First of all, we note that the interpretations which tend to recover as much as possible of the classical worldview are the least promising in this regard. De Broglie-Bohm pilot-wave interpretation is paying a hefty price (of renouncing covariance of special relativity, of difficulties to produce quantum field formulation and of presupposing, against the reasoning of the Occam's razor, unnecessary and unobservable properties such as particle positions), only to regain some intuitive features of the classical physics such as determinism and definiteness of reality at the level of particle trajectories. Yet, these are the same features that disqualify the classical physics from the list of frameworks capable to account for the existence of the subjective experience. In turn, due to these very properties that pilot-wave interpretation strives to reestablish (at the cost of running into inherent mathematical obstacles), it falls in the same category of world-views that likely cannot explain the existence of consciousness. Since the essential difference between the de Broglie-Bohm interpretation and the classical physics is in the much more complicated and nonlocal dynamics (where the unitarily evolving wavefunction governs the motions of particles with hypothetical sharply defined positions), and the no-go theorem does not seem to crucially depend on the details of dynamics or on the locality features, it is unclear how replacing classical physics with this interpretation can evade the conclusions of the theorem. Unless a relevant argument of this type is found, the basic motivation for formulating the pilot-wave interpretation (i.e.\ resurrection of the main classical features) can ironically be, at the same time, a serious argument which disqualifies it.

Next we are going to consider two classes of interpretations that are particularly interesting due to their profound relation to the hard problem of consciousness and which require to be analyzed jointly and comparatively, in more detail.

The first class encompasses interpretations that incorporate, on equal footing and without exceptions, both of the standard postulates: i) postulate of unitary state evolution governed by Schr\"odinger equation; and ii) the wavefunction collapse postulate (i.e.\ the measurement update rule) which includes the Born rule for the probability of measurement outcomes. Such interpretations do not predict any new physics (i.e.\ do not assume any specific collapse dynamics) and stick to what is commonly known as the ``textbook formulation". Note that in all such interpretations the collapse event must be in one or the other way observer or ``system of reference'' dependent -- it cannot be objective in the fullest sense, that is, experimentally verifiable by an arbitrary observer. Namely, if the collapse would occur in an entirely objective, observer-independent way, then we might enclose the collapse event (together with the entire measurement setup) in a big enough insulating box and note that, due to the collapse, the evolution in the box is not linear and unitary -- in contradiction with the first postulate and requiring some extension of the formalism to describe the dynamics of the collapse. (This also holds for Wigner or Penrose like ``consciousness causes collapse'' or vice-versa hypotheses -- they all require new physics to explain details of the objective collapse, and we do not consider such proposals here.) The standard interpretations belonging to the first class are the Copenhagen interpretation (with an appropriate reading), Rovelli's relational interpretation, Brukner's variation of the relational interpretation, QBism and Wheeler's ideas of participatory universe. These interpretations deal in different ways with the scenarios where the two postulates seemingly come into contradiction (e.g.\ of the Wigner's friend type \cite{Wigner}).

The second class on which we focus covers the interpretations that hold that the collapse postulate is unnecessary and can be removed, without ascribing, as a replacement, objective reality to anything apart from the wavefunction (unlike, e.g.\ the De Broglie-Bohm or modal interpretations which discard the collapse postulate but introduce some new elements of reality; the same is true for dualistic ad hoc ontology extensions of the Many-minds type). Historically the first representative of this class was Everett's relative state formulation. As this view of quantum mechanics required further interpreting, it gave birth to many variations of this basic idea, of which to this day the best known is De Witt's many-worlds interpretation.\footnote{Some related ideas, such as the consistent history approach, can also be, to some extent, seen as members of this class.} The underlying idea of these formulations is that the objective, third-person perspective -- represented by the global wavefunction -- must be sufficient to account for everything that is going on in the universe, while the subjective, first-person accounts have to be derivable from this objective picture. In particular, this also includes the collapse: it is in principle not possible to pinpoint why, where or when it happens so it may exist only as a sort of not truly real, subjective illusion of the involved agents, one which in principle can be explained on the basis of properties of the global wave-function.

Already from this overview it seems natural to regard the views from the second class as straightforward quantum mechanical generalizations of the hardline cognitive science positions which deny existence of the hard problem in the sense that they maintain that subjective experience is merely a sort of illusion and can be (as everything else), in principle, entirely deduced from the objective properties of matter (while not attributing ontological reality to subjective phenomena). On the other hand, the interpretations from the first class insist on the necessity of the collapse postulate in order to account for the actual measurement outcomes, obtained by the experimenter and from the viewpoint of the experimenter. No matter the collapse is observer-dependent and thus in some sense subjective (as it is not possible to objectively say when or if it occurs), it is still necessary for the complete description of reality. Thus, in these views, the third-person perspective is not sufficient per se and must be supplemented by taking into account the internal perspective of the system. The latter is done by postulating an entire additional basic law of nature -- the collapse postulate. Hence, the interpretations from the first class can be seen as QM equivalents of the philosophical position that the external, third-person perspective is insufficient to fully account for all aspects of reality, i.e.\ for the subjective experiences (by a straightforward extrapolation these cognitive science positions then hold that there is indeed a hard problem of consciousness -- i.e.\ something that remains even if we entirely understand the dynamics of the brain from the external viewpoint).

Furthermore, denying the essence of the hard problem, while advocating a QM interpretation from the first class would basically be a logically problematic stance. Namely, if one assumes the position that the evolution of the physical brain state (as seen from an objective, external perspective) causes, per se and as a necessary byproduct, the emergence of the corresponding subjective experience (or the illusion of that), then by the most straightforward generalization it could be concluded that a superposition of brain states should correspond to superposition of subjective identities each having the appropriate subjective experience. Consequently, the collapse postulate would not be necessary even to account for the subjective perception of the single measurement outcome (from the experimentalist's perspective), and it would become completely superfluous.

In more detail: presume that we embrace the position that brain, while being merely a lump of matter (albeit a vastly complex one), directly and by itself produces both the objective behavioral responses to the external stimuli and the (illusion of) the subjective experience, while the latter is essentially reflecting and accompanying the former. According to the standard physicalist, no-hard-problem position, a conscious agent should remain conscious even if we enclose her in an informationally sealed box. Besides, we can later confirm (e.g.\ by watching a video record of the box interior, after the box was opened) that the agent inside the box behaved as a conscious being, reporting his subjective experiences in a usual manner, supporting the assertion. So far, the overall state of the agent and the entire system in the box is effectively classical, there is no any (at least not significant) superposition of macroscopic states. In any case, no collapse can play any role in such setup: according to interpretations of both classes, no collapse objectively occurs anywhere in the box (otherwise there would be in principle objectively measurable breakdown of some interference) and what happens outside of the box anyhow cannot be of any relevance. For an external observer, the wavefunction of the system in the box is evolving unitarily, and yet such evolution produces consciousness of the agent inside. By the most direct line of reasoning, we can conclude: the particle system which we call ``brain'' produces consciousness as a side-effect of the unitary evolution of its wavefunction. If we now expose the agent in the box to a visual stimulus of seeing the pointer of a measurement device (of a well-defined location), by the previous premise this should cause the agent's subjective experience of perceiving the pointer position. Besides, in the absence of any superposition of the stimuli (of the pointer position), the system dynamics should closely mimic the classical brain evolution. Now, by the linearity of QM, a superposition of two different stimuli (caused this time by superposition of two pointer locations as a result of some quantum measurement) must produce a superposition of two, mutually orthogonal evolving brain states (entangled with different pointer locations). Since we have concluded that the evolving wavefunction of the brain is sufficient to elicit subjective experience, and that collapse plays no role in that, then each of the two orthogonal brain states should evoke the experience of a separate personal identity having a different subjective perception of the pointer position.\footnote{The fact that we now have a sum of two terms in the wavefunction cannot easily influence the conclusion: any wavefunction can be written as a sum of many terms.} In this way, we have explained the subjective experience of perceiving the single measurement outcome without invoking the collapse postulate, which makes the latter essentially redundant\footnote{Up to a presumption that the remaining problems of the many-worlds view are indeed satisfactorily addressed, as the MW proponents claim.}. 

In yet another words, physicalist presumption that brain ``secretes'' consciousness\footnote{Here we allude to the influential claim by 18th century French materialist Pierre Cabanis that ``brain secretes thought as the liver secretes bile".} leads most naturally to a conclusion that a superposition of two brain states, differing in some memorized information, corresponds to a superposition of two ``secreted'' conscious personal identities differing only in this memorized information. Each personal identity then must have a subjective perception of a well-defined measurement outcome, and introduction of the collapse postulate becomes difficult to justify. This is not to say that introducing the collapse postulate is logically forbidden, but only that it seems unreasonable in the light of the Occam's razor: any attempt to argue that the introduction of the collapse postulate accomplishes anything at all then must prove that some other aspect of the MWI is wrong (commonly it is the derivation of the Born rule that comes under scrutiny, but it is difficult to argue that this problem per se requires something as drastic and vague as is the measurement postulate; besides, we will soon return to the problem of probabilities and point out that it is not separate from the issue of subjective experience).

The above analysis agrees with the reasoning Everett himself used to justify rejecting of the collapse postulate. He heavily relied on the nontrivial assumption that the existence of a term in the global wavefunction that corresponds to a brain state is sufficient to give rise to a subjective experience associated with this brain state. However, in his works, Everett did not recognize this as a new postulate. Instead, he aimed to derive that the subjective experiences of an observer arise from the objective and deterministically evolving wavefunction, by resorting to another assumption \cite{EverettThesis2}: ``As models for observers we can, if we wish, consider automatically functioning machines, possessing sensory apparata and coupled to recording devices capable of registering past sensory data and machine configurations". Then he takes for granted that from the correlation of states of such automata with different experimental outcomes in the superposition one may directly infer: ``...the usual assertions of [the collapse of the state on measurement] appear to hold on a subjective level to each observer described by an element of the superposition.'' His conclusion thus exactly agrees with ours: the hard-line physicalist position on the hard-problem logically implies redundancy of the collapse postulate. Crucially, we also see that the entire relative state formalism explicitly (and openly stated) depends upon the belief that evolving states of these automata, alone and inevitably, cause the emergence of the subjective experiences.\footnote{This is not to say that many-world views do not admit non-physicalist ontological extensions of the Many-minds type, but these are, as in the Many-minds case, commonly related to conceptual difficulties which usually make the logical price to be paid higher than the gain obtained by removing the collapse postulate. More importantly, we have constrained our analysis to the standard and ontologically minimalistic version of the MWI, and, reading their papers, there is very little doubt that it is the one that both Everett and DeWitt had in mind.} In other words, Everett here implicitly presumes, without any further discussion, that there can be no true hard problem. In a similar fashion though in a conditional form DeWitt writes, in an attempt to explain subjective experience of the splitting of the worlds in his many-worlds interpretation: ``to the extent to which we can be regarded simply as automata and hence on a par with ordinary measuring apparatuses, the laws of quantum mechanics do not allow us to feel the splits'' \cite{DeWitt}. We again see that the very idea of removing the collapse postulate entirely relies on the assumption that subjective experiences arise, per se, as a byproduct of the functioning of the automata.

The alternative to MWI reasoning is to deny that the unitary evolution of a brain wavefunction, as seen from outside the box, is sufficient to account for the emergence of the subjective experience. A supporter of one of the first class interpretations can naturally assert that subjective experience can be accessed and taken into account only by considering the internal perspective from the reference point of the brain system. More specifically, this internal perspective must then be a quality that is not a consequence of the external properties (but merely correlates with them), and it should incorporate and invoke the collapse postulate when dealing with superpositions. But this is tantamount to acknowledging the hard problem: the objective, i.e.\ external viewpoint is insufficient to explain subjective experience -- starting with the subjective perception of the single measurement outcome (and as a remedy, we must introduce something qualitatively different, an entirely new physical postulate). In this way we are not violating Occam's principle -- we are indeed introducing another postulate, yet not in extravagance but from a necessity to account for the existing data (i.e.\ for subjective observation of any experimental outcome).

By the previous reasoning, we have pointed out to a more direct relation of the subjectiveness of the collapse postulate (its observer/reference frame dependence) and the subjectiveness of the experience. We provided arguments that a strict physicalist position on the emergence of the subjective experience implies also redundancy of the collapse postulate and thus logically almost necessitates many-world view. Consequently, any consistent opinion which tends to retain the collapse postulate should acknowledge the hard problem in some way.

In this light, whether the subjective experience is reducible to objective phenomena or not, to much extent naturally translates to the question of whether the collapse postulate is reducible to unitary evolution, or must be introduced as a separate postulate. Therefore, the answers to these two questions are tightly interrelated. As for the latter question, volumes have been written both in support and in refutation of Everett's claim that the collapse postulate is redundant.\footnote{Even by the proponents of the many-world ideas, it is nowadays commonly accepted that some additional postulate must be introduced as a replacement, at least to technically reproduce the Born's rule \cite{MWrequiresPostulate}. But these technical details, as well as problems related to the preferred basis, are of a secondary nature to our discussion.} Here we can afford only to briefly illustrate the main aspect of the problem which is relevant to us now and related to the notion of the measurement outcome probability.

For this, it is sufficient to consider a quantum equivalent of a coin toss, let that be a spin z measurement of a particle with spin $\frac 12$ oriented along the x-axis, and an agent, e.g.\ Alice, performing the measurement. Common wisdom, as well as the textbook formulation of quantum mechanics, tells us that after Alice carries out the measurement the observed outcome will present new information to her, that is, she will get to know something she did not know before. Namely, quantum mechanics, invoking the collapse postulate with Born rule, predicts that Alice will observe z spin projection to be either $+\frac 12$ or $-\frac 12$, both with probabilities 1/2. There is no way to predict the result in advance, and as Alice finds out the outcome she learns a piece of information she did not know a minute before (in certain sense a new bit of information, not existing before, has been created). And indeed, this account of the events is what actual experiments confirm, at least in the sense that the experimenters in such situations always report a single outcome obtained with probability one half.

But this is not how the many-worlds (or relative state) interpretation describes the events. In that view, the state of the system (Alice + measurement device + particle) before the measurement is:
\begin{eqnarray}&|{\rm Alice}_0\rangle|{\rm device}_0\rangle|X+\rangle =&\nonumber \\
& \frac 1 {\sqrt 2}|{\rm Alice}_0\rangle|{\rm device}_0\rangle(|Z+\rangle+|Z-\rangle),& \label{Alice1}\end{eqnarray}
where $|{\rm Alice}_0\rangle$ and $|{\rm device}_0\rangle$ correspond to initial states of Alice and the device, while $|X\pm\rangle$
and $|Z\pm\rangle$ correspond to spin projection $\pm \frac 12$ along the given axis. The measurement, in this view, now only induces interaction between the subsystems which results (after a period of unitary evolution) in the entangled state:
\begin{eqnarray}& \frac 1 {\sqrt 2}|{\rm Alice}_+\rangle|{\rm device}_+\rangle|Z+\rangle + &\nonumber \\
& \frac 1 {\sqrt 2}|{\rm Alice}_-\rangle|{\rm device}_-\rangle|Z-\rangle,& \label{Alice2}\end{eqnarray}
where $|{\rm device}_\pm\rangle$ represents the state of the measurement device indicating that the spin projection is $\pm \frac 12$, while $|{\rm Alice}_\pm\rangle$ represents the state of Alice and her brain after registering the corresponding device reading. No collapse has ever occurred in any de facto sense, not even in Alice's reference system alone, and anything that might \emph{look like} a collapse from the viewpoint of Alice is an illusion entirely derivable and explainable by considering the complete wavefunction. Alice persists in the superposition, occurring in both terms of (\ref{Alice2}).

Important point is that there is no element of chance in the above description: the evolution of the system is fully deterministic and easily predictable. After the measurement, we read from (\ref{Alice2}) that in the reality, represented in MWI by the overall system state, there exists both an agent Alice who has perceived projection $+\frac 12$ and an Alice who has perceived $-\frac 12$ outcome. Both are equally real, and the many-worlds view stipulates that both subjective experiences (of perceiving $+\frac 12$ and $-\frac 12$ outcomes, corresponding to two subsystem states $|{\rm Alice}_+\rangle$ and $|{\rm Alice}_-\rangle$ in the superposition) are equally real. This is, as we have already noted, based on unsubstantiated presumption (the fact often neglected) that unitary evolution of the brain state causes per se the emergence of subjective experience, something that is tacitly deduced from the, again unproven, assumption that functioning of an automaton gives rise to consciousness.

But, there is even an additional quantum-mechanical twist to such assumptions about the consciousness, that complicates matters further than in the context of classical physics: namely, the problem to explain where the element of chance enters the story. The account of Alice's measurement scenario given by the many-worlds interpretation does not include any element of chance. In this view, the final state (\ref{Alice2}) is entirely predictable, and there is nothing new that anyone learns from the experiment result, at least certainly not ``objectively new". Following Everett, advocates of many-world theories insist that probability enters as a subjective experience (illusion) of the agents involved. But even this is difficult to establish: the subjective experience of the Alice that recorded spin up, and of the Alice that recorded spin down are equally real, and none is preferential to the other. Yet, subjectively Alice feels that she got to know something new as she learned the measurement outcome, i.e.\ that now she has a piece of information she did not possess before -- therefore, that there still is an element of unpredictability. Alice has a subjective experience that she ended as only one of these Alices in the superposition, and if everything is deterministic and predictable, she would like to be able to calculate in advance as which of these Alices she will turn up. Even if that it is merely an illusion that she turned up as only one of these Alices, she nevertheless should be able to compute everything in a fully deterministic universe, including what the illusion will be. It is difficult to deny that there exists some randomness that defies our ability to predict, since that randomness is so obvious in the practice. If the theory cannot predict the content of the ``single outcome illusion", there still seems to remain an element of chance, not accounted for by the deterministic many-world hypothesis. It is, at least to say, philosophically unclear whether this can be indeed successfully explained by simply saying that there is another Alice in the second branch pondering over the opposite measurement outcome. In an attempt to solve this riddle, S. Saunders, one of the MW proponents, contends: ``...Alice does know everything there is to know: she knows (we might as well assume) the entire corpus of impersonal, scientific knowledge. But what that does not tell her is just which person she is or where she is located in the wave-function of the universe"\cite{MWchance}. But this reasoning seems to indicate that there is a sort of hidden variable, not accounted by the theory, which determines her ``location in the wave-function of the universe". Consequently, such lines of reasoning explain why many-world theory has motivated such a departure from physicalism as is the many-minds variant \cite{ManyMinds, Squires}, and how comes that arguments about a formally realist and deterministic interpretation of a physical theory necessitate discussions about subjective experience, problems of identity (whether minds exist only instantaneously with no continuous identity extending over time), whether account of mind states supplies a natural definition of a future self, whether it is possible for agents to formulate a rational welfare strategy (and whose welfare?) and so on \cite{MWproblems1, MWproblems2}.

In the end, more than 60 years after the first formulation of Everett's ideas, all these discussions have not so far led to any consensus and have left us in the dark whether the very foundations of this class of interpretations are consistent and in agreement with observations or not. In our opinion, however, all these arguments fail to identify the main problem, which is: critically relying on the belief that functioning of an automaton can, per se, induce subjective experiences. Another way to recognize this problem is to note that there is a very straightforward and easy solution to all above logical conundrums posed by the many-word view: all of these foundational problems plaguing formulation of MW theories immediatelly disperse if we assume that there is simply no subjective experience, i.e.\ that we are purely philosophical zombies. All events then happen only objectively: there are objectively two terms in the global wavefunction corresponding to complex particle system which form a certain rigid body, a biomechanical automaton that we denote as Alice; in one term this configuration of particles has been correlated with the spin-up outcome, whereas the configuration in the other term ended up entangled with the spin-down outcome. There is absolutely no mystery of how the initial system containing (in a factor space) a single subsystem state $|{\rm Alice}_0\rangle$ describing particular configuration of particles that form Alice, unitarily evolves into two distinct terms, each containing (mutually orthogonal) subsystem vectors $|{\rm Alice}_+\rangle$ and $|{\rm Alice}_-\rangle$ which correspond to slightly differently arranged configurations of atoms (both atom configurations we, by a convention, call Alice).

Nothing in that case ever occurs subjectively: Alice never feels anything and never experiences anything, the configuration of the atoms which constitute Alice is only capable of ``reporting'' stuff. Reporting is an automated process by which the wavefunction of the particle subsystem that we call Alice, via unitary evolution correlates with the wavefunction of the surrounding air (again a many-particle subsystem) in the manner that can be interpreted as the air carrying vibrations of Alice's voice (and the voice is again correlated with the quantum state of her brain subsystem). Therefore, the subsystem of particles representing Alice can, in both terms of the superposition, further evolve into a state in which these particle configurations are ``reporting'' the corresponding spin outcomes (i.e.\ getting accordingly correlated with air molecules). But none of the Alices is experiencing anything. Under these assumptions there is no longer any problem of ``which of the Alices she would subjectively feel to have become", that is ``which of the outcomes she would subjectively perceive". The answer, in all cases, would be simply -- none. Asking about the subjective experience of philosophical zombies would be plainly an ill-posed question.

Furthermore, since any randomness in these interpretations is only a subjective illusion, removing the phenomenon of subjective experience also removes all the problems related to probability. There is no more even need to derive the Born rule -- some terms in the superposition simply have larger associated multipliers than the others, and that is all to it. On the level of philosophical zombies, i.e.\ of automata without any internal perspective, it is difficult even formulate the notion of probability: we can imagine a zombie tossing an unfair quantum coin, e.g.\ with head to tail probabilities of 25 versus 75 percents, but it would only lead to two equally real superposition terms with coefficients of magnitudes $\sqrt{\frac 14}$ and $\sqrt{\frac 34}$ -- asking what was the probability of a given outcome makes no sense, since we are not allowed to ask what outcome the zombie (or automaton) has actually subjectively perceived. We may have the zombie repeat the tossing many times and record the outcomes, but that would not change much -- unitary evolution would deterministically lead to a superposition of terms corresponding to each of the possible outcome sequences, in general with different multipliers. The zombie in each of the branches can \emph{report} a conclusion based on his outcome sequence: from uttering ``the ratio of heads to tails was approximately 1:3'' in superposition terms with larger multipliers, to exclaiming ``there are only head outcomes in this universe'' in the term with a tiny coefficient.  However, this has nothing yet to do with probability or chance. Another postulate can be introduced (and usually is) to somehow relate the ``number'' (or measure) of identical worlds existing within each branch with the corresponding coefficient in the superposition. This is standardly done to explain ``why are we more likely to find ourselves'' in a branch with a higher multiplier. But this question already implies ``us'' and our subjective experience, and such a step in neither necessary nor easy to motivate in the philosophical zombie scenario. In such a scenario, there is actually no need to bother with the derivation of Born's rule since all \emph{objective and impersonal} experimental data can be explained without it. Therefore, we arrived at an important recognition that problems in MWI related to probability and the Born rule are again the same problems of subjective experience and of the unwarranted assumption that functioning of automata by itself \emph{produces} consciousness (i.e.\ of the negation of the hard problem).

It should be noted that also all ``preferred basis'' problems plaguing (to some degree) MW interpretations immediately go away once we replace agents with subjective experience by philosophical zombies, even without resorting to decoherence: as the global wavefunction is real and it anyhow does not induce subjective experience, it is as real in any of the bases and there is no need to look for a preferred one (singled out by decoherence or not).

Thus, remarkably, replacing agents with philosophical zombies in the many-world scenario solves all its problems. The wavefunction only evolves and changes its form. This certainly produces ever new and more entangled forms of the global wavefunction, but of course, does not produce any experience -- actually a very natural and logical assumption. There is no need to introduce additional postulates, not even to derive Born rule.

But it is that ``illusion'' of subjective perception of a single outcome which we experience in practice and which introduces the notion of probability and demands the Born rule: the rule is necessary in order to explain outcomes of \emph {subjective} experiments, to account for ``what I've perceived", and is not needed for explanation of \emph {impersonal, objective} experiments. Curiously, or maybe expectedly, what then introduces the problems in MWI is actually insisting that these wavefunctions of zombie-agents (i.e.\ of some particle configurations) should be miraculously endowed with the subjective perception, i.e.\ insisting that it must be possible also to assume the first-person perspective and have experimental predictions from this personal vantage point. We saw that there is one clear-cut logical solution of how to consistently establish the relative state (or many-worlds) formulation -- i.e.\ the one with philosophical zombies. However, the proponents of these QM interpretations would like to assume that there is yet another solution, which incorporates subjective experiences (but without postulating consciousness as some fundamental entity as in Many-minds interpretation). However, for this latter part, they do not provide any arguments (just take it for granted), whereas this very part turns out to be quite tricky, as we have pointed out and as is manifest through the ongoing debate on MWI foundations.

As a side note, we see that the hard problem itself in the quantum context is deeper and more palpable than the same problem in the classical case -- the difference is that in QM switching from the objective perspective into the first-person account must, unlike in the classical physics, introduce also breakdown of the determinism of continuous evolution and result in appearance of subjective discontinuity and chance. As we have pointed out above, without the notion of subjective experience in MW interpretations it is difficult to introduce notions of probability and randomness. On the other hand, it is this discontinuity and randomness that we observe in the experiments, since all experiments can be only performed and analyzed subjectively in the end. This aspect of the hard problem was not present in the classical physics: opponents of the hard problem could far more easily negate the very existence of the subjective experience since there was not this qualitative difference between determinism and chance which subjective experience was supposed to bridge. Is much harder here than it was in the classical paradigm to reconcile the internal, personal viewpoint with the ``impersonal'' objective one where nothing is left to chance and everything is continuous. For all these reasons, it is overall more difficult to ignore the hard problem in the quantum case, and this reflects in the philosophical obstacles to discard the collapse postulate and in the history of the subjective elements of interpretations of the quantum mechanics since its advent.

Once we have established this essential connection between the hard and the measurement problem, we may also recognize that some other related concepts from cognitive science have their counterparts in the QM. For example, the ``other minds problem'' (the impossibility to prove whether other agents possess subjective experiences, already discussed in Part I) translates into impossibility to objectively prove if and when the collapse occurs from the vantage point of someone else, who is informationally isolated from us (i.e.\ in spite that Alice will later report perceiving a single spin projection measurement outcome, we expect all objective experiments to confirm her to exist in the superposition state (\ref{Alice2}) ). Also, the Searle's problem of how syntax could ever produce semantics (i.e.\ how ``meaning'' can arise from manipulations with symbols or from complex but per se meaningless motion of some matter), in QM context can be traced to the puzzle of how is it possible that a bunch of particles, having superposed, smeared properties and undergoing unitary evolution can form an observer from whose vantage point the collapse occurs and turns probabilistic potentialities into actualized reality. But we will not further elaborate on these connections here.

To sum up our previous conclusions: the hard physicalist line on the subjective perception implies that the collapse postulate is superfluous, but it also goes the other way round -- MWI explicitly requires (and thus implies) presumption that subjective experience is a necessary byproduct of functioning of an automaton, therefore, one position does not go without the other; accepting of the collapse postulate demands also recognition of the hard problem in some form; due to the gap between the subjectively perceived randomness and the determinism of the objective (from an external viewpoint) wavefunction evolution, the hard problem is more difficult to ignore in QM; illustration of the previous conclusion is the fact that philosophical zombies automatically solve all the riddles posed by MW formulation, while it is the introduction of the subjective experience hypothesis that lies in the core of the foundational controversies of the MW theories; the established relation between the hard and the measurement problems allows us also to translate some other concepts, such as ``other minds problem'' and ``syntax is not semantics'' from cognitive science into QM.

Of these conclusions, it is particularly indisputable (not only because it is explicitly stated by both the Everett and DeWitt) that QM interpretations of the second class, which seek to discard the collapse postulate, rely heavily on the assumption that functioning of an automaton can, per se, elicit subjective experiences. Even without our no-go theorem, which we did not yet take into account in the above analysis, the last hypothesis opens many conceptual problems in MW formulations. But the derived no-go theorem makes the case for MW theories much worse, since the theorem's most probable implication is that this vital tenet of MW views is entirely irrational and hard to logically sustain. It is true, however, that the no-go theorem was derived within the paradigm of the classical physics, whereas, strictly speaking, Everett's class of formulations relies on the idea that it is the wavefunction corresponding to the automaton which should be capable to produce subjective experience. But, unless a strong argument could be constructed as to why would a wavefunction of an automaton have more potential to produce subjective experience than the automaton itself, the conclusion of the no-go theorem provides a serious counterargument for this class of QM interpretations. In fact, proponents of the De Broglie-Bohm interpretation find the opposite more reasonable: in the pilot-wave formulation the objective unitary evolution of a brain wavefunction alone does not give rise to subjective experience -- it is only the term of the superposition occupied by hypothetical ``real'' particles which corresponds to actual experiences.\footnote{While the many-world supporters believe that any term in superposition containing an active brain state generates subjective experiences, the pilot-wave adherents hold that all problems related to such assumption are somehow solved by a claim that filling the wavefunction branch by real particles is what brings the subjective experience into existence. Both groups imply that these, mutually contradicting, positions go without saying and it is particularly puzzling that none of them feels prompted to provide any arguments in their favor. It seems that each group erroneously takes for granted that neuro and cognitive science have already solved these issues with a conclusion that somehow supports exactly their hypothesis.} In the end, not only that the class of MW interpretations fails to offer any novel potential for the explanation of the existence of the subjective experience (a quality that must be on the list of attributes of a plausible interpretation), but it turns out that the no-go theorem likely invalidates a pivotal presumption on which these interpretations hinge.

Finally, what about the first class of interpretations, containing views which incorporate the collapse postulate -- do these interpretations provide a solution to the hard problem? That would be quite peculiar, since not only that almost none of the proponents of these interpretations advocates these views as solutions to the problem of subjective experience, but most of them are seemingly totally unaware of the entire context of the cognitive hard problem.

Nevertheless, we first need to notice that the line of reasoning of the no-go theorem definitely can be no longer applied within the paradigm of any of these interpretations. It is not just that the dynamics of the classical physics is replaced by the unitary evolution of the wavefunction given by the Schr\"odinger's equation -- as we have already argued, this difference alone could hardly bring us closer to any resolution.

The first potentially relevant difference is the introduction of indeterminism. However, this step alone also does not seem to be sufficient to allow for the existence of consciousness. If we simply supplement mechanistic materialism of classical physics with an element of randomness, then the resulting non-deterministic dynamics of the brain could be accounted for, in the theorem(s), by its simulation on a non-deterministic computer, and later represented by a nondeterministic Turing machine. It is known that a non-deterministic Turing machine is equivalent to a deterministic Turing machine endowed with an additional tape with predefined random numbers (each time an evaluation should depend upon a chance, a random input is obtained from this auxiliary tape). In this way, it seems to be possible to repeat the same proof of the no-go theorem also for a non-deterministic equivalent of the classical paradigm. The point is that such a stochastic mechanicism is difficult to differentiate from deterministic mechanicism with (nonlocal) hidden variables: every element of chance can be, in principle, explained by (nonlocal) interaction with some degrees of freedom inaccessible to us. It is very hard to motivate why such dynamics would be qualitatively different from the one discussed in the no-go theorem (even if the source of this randomness would be external to the physical universe).

Nonetheless, an important aspect of abandoning the determinism lies in the informational content of the universe. Namely, common for all causally-deterministic models is that the existing information encoded in the state of the universe never changes, and is already contained in the initial conditions.\footnote{No new information (in Kolmogorov's sense) can appear if the universe is in every future moment uniquely determined by its present state. However, strictly speaking, for the information to be unchanging it is also necessary that past can be uniquely reconstructed from the present, which is a requirement that was satisfied by all major deterministic models in physics so far. Without this, the amount of information could reduce as time progresses. Some other subtleties which may arise here in the context of General relativity we find of no relevance.} In certain (informational) sense, in any deterministic universe, nothing really new ever happens. In every new moment, we have the same old content merely represented (rewritten) using slightly different coding rules than before. This was very much evident in our deterministic-brain simulator in the form the sand-symbols Turing machine (in Part I): all the rows of calculation progress were already contained in the first row with the initial data, and this entire simulated-brain evolution was no more that a sort of ``unzipping'' of the information packed in the first line, producing essentially nothing new (this is exactly why the information content, in the Kolmogorov's sense, by definition remains the same during computation). Formally quantum-mechanically, it is even less than that -- time evolution can be accounted for by mere unitary change (a sort of rotation) of the Hilbert space basis vectors. Any appearance of new information, as already discussed in the context of class two interpretations, can be at most just a local illusion of an agent which perceives only her branch of the global wavefunction (and, if so, it must be explained how and why this illusion arises).

Contrary to that, the collapse postulate axiomatically introduces the emergence of new information, as the result of a measurement, whenever the probability of the outcome was less than one. Which of the possible outcomes will actualize -- that piece of information which we obtain upon the measurement -- was not previously present at all in the universe. This emergence is seen by Wheeler (and often repeated by A. Zeilinger) as ``an elementary act of creation'' \cite{ActOfCreation}. Still, without more serious modifications of the ontology, the new information might be simply coming from some pre-supplied source, e.g.\ like from that additional random Turing tape, and whether the ``location'' of this information source (tape) would be within our universe or external to it, could hardly influence our earlier conclusions. Therefore, to change the no-go theorem's verdict about consciousness and to perhaps grant some real significance to indeterminism itself, we need a much more momentous change than the mere introduction of randomness.

And the room for a true and relevant paradigm shift lies in the fact that the interpretations of the first class effectively (if not always explicitly) give up the idea that the first-person account of events is logically derivable from the third-person perspective, i.e.\ from the evolution of the wave function alone. An entirely new postulate is introduced to explain the internal perspective, relative to the observer/system. When and whether the collapse occurs is only relative to the observer (i.e.\ subjective), and even the mere existence of the generally agreeable facts (i.e.\ of the ``facts of the world'' in Brukner's formulation\footnote{It is there disputed that ``one can jointly assign truth values to the statements about observed outcomes (`facts') by different observers'' \cite{Brukner2}.}) is thus brought under dispute. In general, for the views of the first class, it is then also the wavefunction which becomes relative to the observer, not having an objective ontological reality which is otherwise ascribed to it by the interpretations from the second class. It is the actualized quantum events, as seen and as existing from the relative perspectives of observers/systems, which are granted ontological existence, while the wavefunction turns more into a bookkeeping aid to predict probabilities of future events, given the previous ones. It should be also stressed that even with the use of the collapse postulate an observer can only infer things about his exterior: any attempt even from him to closely and mathematically follow own internal state and rigorously deduce how his own subjective experiences arise, seems to be fundamentally hampered by the impossibility of complete self-measurement (the later impossibility is derived in the context of relational interpretations \cite{NoSelfMeasurement}). These differences are now sufficient for the complete breakdown of the no-go theorem derivation. We can no longer speak at all about the ``emergence'' of the internal, subjective observer/system perspective from the objective, third-person perspective system dynamics. This internal perspective is here taken to be fundamental and non-derivable, seemingly both from the external and from the internal viewpoint.

However, it should be immediately made clear that interpretations of the second class do not necessarily impose any anthropocentric element, that the word ``subjective'' might resonate with. One must either accept that a human can be also seen as a complex system of atoms, or must seek some extensions of physics (and we have many times expressed our belief that the latter is not necessary). It is thus a reasonable position that the right to this fundamental, non-derivable internal perspective thus should be granted in some way also to other systems apart from humans (or animals), notwithstanding the fact that the complexity of the human system (and of the brain in particular) introduces new qualitative features to this internal perspective. Relational interpretation is a good example: it is there mostly understood that every (quantum) system is entitled to a referential role and the collapse (i.e.\ quantum events) can happen from its frame of reference. The exact word ``subjective'' hardly ever appears in the context of this interpretation -- the role of this word is usually taken by the notion of observer relative events/facts, where observer can be in principle any system and the relativity is more understood in a similar way as in Einstein's theory of relativity.\footnote{We actually do not think that there should be any fundamental, low-level difference between these two concepts. Below we even discuss a view that subjectivity of the collapse event, i.e.\ the fact that it may occur for one observer but not for the other, can be replaced by temporal relativity of when the subjective experience arises.} For example, while the length of a rod is observer (inertial frame) dependent in the special relativity, this is not (at least not commonly) seen as an indication that the rod or its length is in any way subjective or not objectively real. Though this is in part a matter of terminology, there are some important differences between the relativity of the rod length and the observer-dependence of the collapse, which lend more justification to the use of the word ``subjective'' in the latter than in the former case. In the former case the description of the reality will be to some extent different for different observers -- however, each of these descriptions is sufficient to account for all elements of reality: each observer can reconstruct how the reality looks from the aspect of another one since each of the perspectives in principle contain the same mutually-nonconflicting information.\footnote{We will not here discuss special cases in General relativity, such as observers falling into a black hole and alike, but it can be argued that their case is also different from the relativity of collapse.} Was it the case that the rod did not exist at all in one referential frame, while existing in the other, that would certainly raise more doubts about the rod's objective existence. But that exactly seems to be the case with the collapse: in configurations of the Wigner's friend type \cite{Wigner}, when an observer is performing a measurement of a quantum system while being sealed in a box, the collapse she observes does not occur at all (i.e.\ does not exist) from the outside-the-box perspective. Also: while in relativity the perspective and the properties of reality depend on kinematic (geometric) properties of the referential (coordinate) system, in the case of the collapse we can only vaguely relate it to the ``internal'' perspective of a physical system. Combined together, these differences intuitively give more grounds for the use of the word ``subjective'' in the case of collapse, though, in the end, it certainly comes down to the precise definition of this word.

In any case, in all interpretations of the first class lurks a possibility to slightly intuitively extend this notion of the fundamental non-deducible and irreducible observer/quantum system right to internal perspective in order to recognize that the same phenomenon might be essential for the subjective perspective and the subjective experience in common sense. And in this possibility lies the potential to not merely avoid the conclusion of the no-go theorem but also to somehow account for the existence of consciousness. Therefore, in our opinion, the least what can be said for the interpretations from the first class is that they leave a room for some future explanation of consciousness, unlike, it seems, the rest of here considered QM interpretations. In the modern terminology, we could dub the first class as hard-problem ``solution ready'' interpretations, in the sense that the formalism of physics is potentially ready to handle, i.e.\ incorporate an explanation of the existence of subjective experience in a sort of natural, seamless way.

This main idea can be pursued further, towards some more concrete hypotheses on how consciousness can be understood. That is the subject of the next section.

\newpage
\section{Part III: Getting closer to the hard-problem solution}

The main motive in this section is to demonstrate, by providing a sketch of some logically consistent solutions, the potential of the first class interpretations to take us closer to the resolution of the hard problem. The ideas presented in this section are of a more speculative nature and should be rather seen only as possible options on how to proceed from here.

As we have already announced, an obvious path to follow is to recognize the existence of this non-derivable internal perspective, inherent to the collapse postulate, as the essence also of the phenomenon of subjective experience in the every-day meaning of the word. Therefore, the subjective experience, i.e.\ the consciousness, in its core has to be then seen as a fundamental, non-derivable concept, too.

Of course, the inner perspective of the type we were referring to so far cannot be alone sufficient to manifest high-level consciousness of human type (with its standard traits such as the internal monologue, sense of self, phenomenal experiences, emotions...). While considering the functioning of the brain as seen from the external perspective seems to be insufficient to explain or infer the existence of the subjective experience of the brain-owner (without having to postulate the internal perspective of systems), brain certainly does play a huge role in shaping the form of these experiences. The postulated existence of the ``internal perspective'' appears here only a necessary, but certainly not sufficient, prerequisite for the existence of human-like subjective experience.\footnote{Some researchers suggest that this internal perspective, which could be understood as a sort of ``pure'' consciousness stripped off all personal human features, should be called differently to avoid confusion. E.g. C. Stoica suggests the word ``sentience'' in this context, while the word ``proto-consciousness'' might be suitable but is sometimes used in a relation to panpsychism and thus with a somewhat different meaning.}

This relation of internal and external perspectives in the context of subjective experience obviously begs for further explanation. We must consider it first on the low, quantum mechanical level and clarify how is it possible that the wavefunction collapse can occur internally, but does not happen from the external viewpoint. Next, we must discuss the same relation on a higher, cognitive level by addressing how and why agent's subjective experience of the world and of her own actions coincide with the external description of that agent, when seen as a complex particle system. The latter problem is particularly nontrivial if we want to renounce the epiphenomenal views on the consciousness and allow for the existence of free will -- which is philosophically certainly more preferable position if can be logically reconciled with the laws of physics. After that, we will discuss in which way the expounded framework might potentially lead to an explanation of consciousness capable to pass the test of the no-go theorem.

\subsection{Reconciling different perspectives -- quantum mechanical level}

When we consider human-like agents as physical systems (i.e.\ as seen externally), they possess many internal physical degrees of freedom (located in agent's memory) which correlate with the external degrees of freedom (in the environment). If represented by a wave function this correlation is modeled as in (\ref{Alice2}), where $|{\rm Alice}_\pm\rangle$ corresponds to the state of the agent's internal memory degrees of freedom which are entangled with the external degrees of the measurement device $|\rm device\pm\rangle$. (In general, there can be more terms in superposition, with varying multipliers, or just a single term.) We also note that this type of entanglement between the agent and the environment is not only initiated via agent senses, but also the other way around: via behavioral responses and via different forms of involuntary leaking of information from the agent, her internal states leave a certain imprint on the environment. However, once we have adopted views of the first class, we should not see the superposition and entanglement of (\ref{Alice2}) as any real entity -- it is just a part of mathematical procedure helpful to calculate probabilities of future measurement outcomes.

This was the external description and, in addition, we must also incorporate the postulated non-derivable internal perspective. There seems to exist a direct correlation, a sort of mapping between the physical degrees of freedom localized in Alice's brain (as seen externally) and Alice's internal subjective experience.  We may say that Alice is in some sense \emph{aware} of her brain state (or of a part of it).We will return shortly to the question of how the existence of such a correlation could be explained. Via this correlation of subjective experiences and brain states, where the latter are further correlated with the environment, Alice forms a subjective internal representation of the external world. As these subjective experiences (constituting the internal representation) are by definition well defined, real and do never exist in any sort of ``superposition of qualia'', by the virtue of this chain of correlations also certain external degrees of freedom (such as the device pointer position and the z projection of particle spin) become well defined and real from Alice's perspective. This she recognizes as the wavefunction collapse.

It is worthwhile to stress a few points. First, it is wrong to apply classical intuition about the time ordering of the cause and effect when it comes to collapse: in the spirit of Wheelers delayed-choice effect \cite{DelayedChoice}, there is nothing contradictory in the fact that Allice can, from her vantage point, actualize events in her arbitrarily remote past by experiencing them in the present. Secondly, the collapse here is not objective, in the sense that for another agent, e.g.\ Bob, neither Alice nor the device has to be in a well defined classical state -- they become well defined and actualized for him only as these states get entangled with his internal degrees of freedom and he subjectively experiences them (this is in clear contrast with the objective collapse in the Penrose-Hameroff orchestrated objective reduction \cite{OrchOR}, or in the Wigner's consciousness causes collapse hypotheses \cite{Wigner}).

Due to this latter reason, perspectives of different agents on the very occurrence of the collapse can be directly conflicting and reconciling them is a nontrivial problem -- and one that potentially leads to profound implications. The problem most seriously surfaces in the form of the Wigner's friend type of paradoxes \cite{Wigner}, which, according to some authors, jeopardize the very foundations of QM and question its consistency \cite{RennerNature}. We share the majority view that there are many ways to consistently deal with this ostensible contradiction (without altering the mathematical formalism of QM), as we have discussed in detail in \cite{ToTheRescue}. Here we will only summarize some main points from that paper which deals with the Wigner's friend paradox, to the extent that will be needed later in our analysis.

Alice, in the initial state described by (\ref{Alice1}), can have the role of the Wigner's friend if we assume that she is enclosed in an ideally insulating box. Let us denote that initial moment corresponding to state (\ref{Alice1}) as $t_0$, while at some future moment $t_1$ Alice performs the spin measurement and learns its z-axis projection. Bob, who here plays the role of the Wigner in the original setup \cite{Wigner}, is outside of the box and from his vantage point, after $t_1$, the system in the box (including Alice) must be in the entangled state given by (\ref{Alice2}) -- we do not consider objective collapse models, so there cannot be other options. At moment $t_2$ Bob performs in-practice-impossibly-delicate and yet in-principle-feasible quantum interference measurement on the entire box content (which includes measurement of Alice's brain quantum state), by which he experimentally proves that Alice is indeed in the superposition state given by (\ref{Alice2}). Note that since the state (\ref{Alice2}) is an eigenstate of his measurement, the measurement actually does not alter at all this prior state of the system in the box. On the other hand, had the collapse occurred \emph{objectively} at $t_1$, then from his and every other perspective at time $t_1$ state in the box would have turned either into $|{\rm Alice}_+\rangle|{\rm device}_+\rangle|Z+\rangle$ or into $|{\rm Alice}_-\rangle|{\rm device}_-\rangle|Z-\rangle$, and he would have mere 50 percent chances to measure the superposition state (\ref{Alice2}). Finally, in \cite{ToTheRescue} we argued that of utmost importance for the understanding of this scenario is also the moment $t_3$ when Alice eventually leaves the box and tells Bob which outcome she perceived.

The views of different interpretations about this scenario largely diverge. The tricky part is how to address the question ``what did Alice experience in between $t_1$ and $t_2$?'' (at $t_2$ Bob must perform the measurement on Alice's brain, so it is beyond the point to discuss what that part of the experience might be like, or whether Alice has to be sedated for this hypothetical operation). Namely, both classes of interpretations agree what the description of the events and the outcome of the experiment must be from Bob's point of view: quantum evolution of the system in the box is unitary (for him), state in the box is after $t_1$ described by (\ref{Alice2}), which he experimentally confirms with certainty at $t_2$ and, finally, at $t_3$ he is told by Alice that at $t_1$ she either measured spin up or spin down (with equal probabilities). However, the outcome that she finally reports (together with the content of her memory regarding the experiment that she had carried out at $t_1$) was somehow decided at $t_3$, since at $t_2$ she was confirmed to be in the superposition (\ref{Alice2}), and the interference measurement must have erased any trace in the brain of any particular outcomes occurring at $t_1$, if these existed. The only thing that is under dispute is what she has subjectively experienced after $t_1$.\footnote{Already at this point we find the Wigner's friend type of paradoxes to be very peculiar in this context, as such paradoxes actually completely depend on the existence of subjective experience: if one refrains from questioning how the events look from the first-person viewpoint (of the Wigner's friend), the problem entirely disappears.}

According to the MWI, there are simply two Alices in two worlds (branches), each experiencing one of the spin outcomes at $t_1$. These branches recombine in a predictive way during the interference experiment at $t_2$ and finally split again once for all at $t_3$. Of course, in this view, no collapse occurs, apart from as subjective illusions of the ``two'' Alices.

The interpretations of the first class at this point do not concur even among themselves. The QBist view might be seen as a sort of agnostic in this regard.\footnote{C.\ A.\ Fuchs in \cite{QBism} explains: ``It is just that quantum theory provides a calculus for gambling on each agent’s own experiences -- it doesn’t give anything else than that. It certainly doesn’t give one agent the ability to conceptually pierce the other agent’s personal experience.''} Relational and Brukner's interpretation \cite{Rovelli, Brukner, Brukner2} would imply that at $t_1$ Alice performs the measurement and thus the collapse occurs, but only from her perspective. She reads out a single well-defined outcome and, unlike in MWI, there is no other Alice in another branch obtaining the other outcome. Accordingly, she also experiences the perception of a single outcome, in spite that no outcome is ever singled out at $t_1$ from Bob's perspective. In Bob's reality neither ``up'' nor ``down'' in any sense happen at this point, none of these two possibilities is preferred or realized, no matter that this bit of ``up-or-down'' information was indeed created in Alice's world.\footnote{Note that whatever she perceived at $t_1$ has nothing to do with the outcome she finally reports at $t_3$ -- in the meanwhile she has undergone the interference measurement in the basis which is incompatible with well-defined outcome at $t_1$, and any outcome she reports at $t_3$ must have been ``decided'' anew somewhere in between $t_2$ and $t_3$. However, if the outcomes she perceived at $t_1$ and $t_3$ are not necessarily the same, then it is unclear what and when causes the possible switch in the perception since the interference measurement, as noted before, does not alter the Alice quantum state.} In the view of these interpretations, the paradox here is only superficial since ``there are no general facts of the world'' \cite{Brukner2} and talking about what had or had not ever happened makes sense only given the precise reference system. This further implies that there is a huge number (probably infinitely many) of different worlds corresponding to different (informational) reference systems, each with its own ``facts'' being real (which slightly reminds of the infinite number of realities in MWI).

Finally, in \cite{ToTheRescue} we have pointed to a specific reading of the standard textbook QM interpretation which allows us to retain the comfort of single reality and of generally valid facts of the world. Of course, some price has to be paid, since we have pointed out many times that the collapse in the interpretations of the first class is necessary to some extent subjective (i.e.\ observer relative). But instead of seeing as relative \emph{what} has become real (as according to the ``no objective facts of the world'' motto), the other option is to accept as relative only \emph{when} events become real. This should come at a lesser price since we are already used to the relative nature of time, thanks to Einstein.

The key element is to consistently apply the idea that the events become decided and in the full sense real only upon the collapse. In particular, this means that, from Bob's perspective, the entire Alice's spin measurement activity and her accompanying subjective experience can not be decided and real prior to the final moment ``$t_3$'' when he became correlated with Alice and the outcome. For anyone outside the box, including Bob, it makes no sense to ask what were Alice's experiences at the moment $t_1$: the answer can be neither that she perceived one of the outcomes (as in Rovelli/Brukner), nor both outcomes (as in MWI). None of these hypothetical noncorrelated and informationally disconnected events have become a part of Bob's reality until $t_3$.

To the extent that it makes sense at all to speak about Alice's first-person subjective experience from the third-person perspective of Bob, we may say that these experiences become real at $t_3$, together with the corresponding events that happened in the box and collapsed into reality just at that moment. This does not introduce any inconsistencies since these are experiences of Alice having seen the well-defined outcome when her wristwatch showed time $t_1$ -- so Alice's memory and personal recollections, though decided and actualized at $t_3$ (from Bob's perspective), are consistent with the time-line in which she performed the measurement at $t_1$. Rather than seeing this explanation as if Alice was, from Bob's perspective, in a sort of ``suspended animation"\footnote{The term used by Wigner himself in his paper \cite{Wigner}.} from $t_0$ until $t_3$, we may understand it by arguing that it is pointless to speak about time synchronization of events in total absence of any physical flow of information which could maintain this synchronization. In this case, there is no, by assumption, any flow of information outwards correlating events inside the box with the outside timeline -- such a flow would make events inside both real and truly positioned in time (and space) also from the outside perspective (discussion of the experimental setup in presence of partial informational flow, such as in Deutch variant, can be also found in \cite{ToTheRescue}).

From Alice's perspective, she subjectively perceived the spin measurement outcome at the moment when her watch showed time $t_1$, without experiencing any weird gaps, and the outcome she perceived is the same one that she later reports to Bob at $t_3$ (the same outcome that was, from Bob's perspective, undecided all the time until $t_3$). Of course, she experienced the measurement only once, and that outcome is the one she remembers from then on (which does not seem to be necessarily the case in the Brukner's view, where perceived outcome at $t_1$ and the one remembered after $t_3$ may differ). Nothing in her subjective experience was altered by the fact that she was enclosed in the box and that Bob performed the measurement on her brain (we remind that this measurement did not introduce any change in the wavefunction). She would have had exactly the same experience of the spin measurement with or without the box surrounding her. Neither was her subjective experience somehow ``objectively'' postponed until the moment $t_3$ -- such assertion would be impossible to prove by any measurement and thus would be unphysical.\footnote{This aspect can be compared to the situation of an observer falling into a black hole: in spite of the fact that the infalling observer will cross the horizon only after infinity of time as measured by a distant observer, it by no means implies that he is bound to ``wait'' or be ``suspended'' for infinity until he experiences the crossing.} For all these reasons nothing changes even if we enclose Bob, together with Alice, in yet another much larger box, or if we ask ourselves if there is an ``exterior'' observer keeping us in the box right now.

In this way, the laws of QM are preserved both from Alice's and from the Bob's viewpoint and the confusion can arise only if Alice mistakenly applies the perceived subjective collapse as something objective, which should then pertain also to Bob's measurements (and thus erroneously conclude that he should obtain superposition state (\ref{Alice2}) outcome with only 1/2 probability). Intuitively, this view perhaps goes best together with the Wheeler's idea of the participatory universe (Participatory Anthropic Principle): observers exist and through their mutual interaction they co-create a single shared reality, or, in Wheelers own words ``the observations of all the participators, past, present and future, join together to define what we call `reality' '' \cite{Wheeler2}. It is the logical consistency constraints that give rise to the seemingly objective and persistent properties of the external world -- the idea which Wheeler used to illustrate by likening reality to his variant of the ``twenty questions'' game: one player tries to guess the imagined object by posing (up to 20) yes-no questions, but in Wheeler's version, the players who answer are playing a trick on him -- there is actually no well-defined, decided object at the beginning, it is only the questions and the required logical consistency of otherwise random answers that gradually define and shape the object \cite{TwentyQuestions}. This Wheeler's view is inseparable from his ``it from bit'' doctrine, where information replaces concept of matter at the heart of reality: ``Otherwise put, every `it' -- every particle, every field of force, even the spacetime continuum itself -- derives its function, its meaning, \emph{its very existence entirely} -- even if in some contexts indirectly -- from the apparatus-elicited answers to yes-or-no questions, binary choices, bits. It-from-bit symbolizes the idea that every item of the physical world has at bottom -- a very deep bottom, in most instances -- an \emph{immaterial source and explanation}; that which we call reality arises in the last analysis from the posing of yes-no questions and the registering of equipment-evoked responses; in short, that all things physical are information-theoretic in origin and that \emph{this is a participatory universe.}"\cite{Wheeler1} Thus, the information the observers obtain from measurements -- which in this context means the information which is, in the end, subjectively perceived -- is the only true reality, unlike the objects themselves which only derive their existence and properties from that information. Since reality is here unique (as opposed both to MWI and ``no facts of the world'' views) and there is a presumed symmetry between the observers (i.e.\ participants) of which all observe the same laws of QM, it seems to us that the above understanding of the Wigner's friend scenario is the only one that fits the premises. Besides, it seems more natural in this framework to accept that, in the absence of any mutual interaction, subjective timelines of different participants simply are not commensurable at all (e.g.\ if there was never interaction between two sets of observer-participants, these would effectively exist in two unrelated universes, and it would be completely pointless to ever ask what the participant from one group experiences at a given time measured by the time-scale of the other universe).

As already noted, a more detailed analysis of Wigner's friend scenario along these lines can be found in \cite{ToTheRescue}. In this context, however, we just needed to underline this possibility to preserve observer-independence of the ``facts of the world'' while relativizing and making subjective only the ``when'' of the collapse -- as this option will be later important in the relation to the hard problem.

\subsection{Reconciling different perspectives -- cognitive level}

On the higher, cognitive level, once we have recognized the essence of consciousness to be fundamental and non-derivable from the external viewpoint, there is a problem of explaining the intimate connection between the internal subjective experiences and the externally visible behavior of the conscious agents. On one side, this external behavior is (at least it certainly seems to be) entirely explainable in terms of dynamics of matter (i.e.\ of the brain functioning). When seen as a complex particle system, as a ``quantum-automaton", probabilities for different responses of an agent to a given stimulus should be predictable by computation (i.e.\ we should be able, in principle, to simulate unitary evolution of the entire agent-system state and calculate probabilities of the various responses). The functioning of ``quantum automaton'' Alice, as seen from external, e.g.\ Bob's perspective, should obey laws of physics, in particular, it should not violate Born rule.\footnote{We will take this for granted, though, in principle, this is yet to be verified by the neuroscience. On the other hand, if the Born rule is violated in the brain, than we need new physics to explain why a lump of matter called brain defies a basic physical law.} And these explicable and predictable (at the level of probabilities!) external properties should somehow perfectly correlate with subjective properties, i.e.\ experiences. The root of this connection on the brain level is called ``neural correlates of consciousness'' (NCC). Historically first and natural hypothesis about NCC was that the physical dynamics of neurons not only correlates with, but directly and per se produces the subjective experiences. Yet, we have seen that this idea inevitably faces the explanatory gap which seems to be impossible to bridge (to explain how motions turn into qualia), and this problem we tried to formalize as much as possible via the no-go theorem. Alternatively, by taking the path tacitly suggested by QM interpretations of the first class, the neural correlates of consciousness should be taken to be exactly that -- merely correlates, expressing the compatibility of internal and external perspectives, in other words, as a sort of consistency constraints.

And in spite of the presence of perfect correlations, we stress once more that the very existence of the internal perspective does not logically follow from the external description of an agent, and cannot be derived from it. The external description provides a consistent, self-sufficient account of the external behavior of any considered agent, which is neither sufficient to predict detailed qualities of internal experiences (e.g.\ if subjective experience of ``red'' is the same among different agents), nor even to prove that these internal experiences exist at all (external description is entirely consistent also with the philosophical zombie hypothesis). Thus, left only to external observation, we can never positively and rigorously prove even that any other minds apart from our own exist (and we have to resort to various different philosophical arguments to discard the solipsist option \cite{OtherMinds}). As noted before, this is known as the ``other minds problem'' and is, as the denomination ``problem'' suggests, usually seen as an obstacle which requires some sort of solution (e.g.\ the hope that pursuing investigation of neural correlates of consciousness could someday prove that subjective experiences indeed arise as a consequence of neural functions).

Instead of seeing this perfect correlation as a problem that awaits for an explanation, we believe it should be  pronounced a basic principle, in other words, elevated to the level of ``other minds postulate'' pivotal for the existence of the Universe as we know it.\footnote{This type of twist can be seen as borrowed from Einstein's reasoning, where we try to mimic the way he turned the problem of how to explain the apparent constancy of the light speed in Michelson-Morley experiments into the basic axiom of its constancy.} The postulate might be formulated as: ``In a Universe inhabited by conscious beings, laws of physics must be such to allow complete consistency of internal (i.e.\ subjective) and external perspectives". In other words, it must be possible to explain externally visible behavior of agents solely on the grounds of dynamics of physical systems, and, at the same time, that behavior also has to make sense from the internal side, that is, it has to be subjectively experienced as stemming from conscious decisions. This ``other minds postulate'' necessitates the ``other minds problem'': the physical laws must allow different minds (observers) to fit into reality (or each others realities) so seamlessly, consistently and without violating physical dynamics, that it \emph{must} be impossible to ``objectively'' prove existence of the subjective experience.

At this point we must make an important distinction: imposing the ``other minds postulate'' does not introduce anything quite new if we take consciousness to be an epiphenomenon. In that case it is trivial to reconcile subjective experience with the external description of the agent, simply because then it is the external dynamics (in combination with chance or not) which alone drives agent's responses to stimuli, where, on the other side, the subjective experience merely follows the events as they unfold and what we perceive as agency is just an illusion.

On the other hand, the idea of epiphenomenal consciousness was never a natural one: it contradicts our subjective perception that we are influencing our behavior and thus it must be artificially introduced. It was a logical hypothesis in the settings where also the entire subjective experience was a sort of illusion itself, but once we were forced to give up the latter idea, this assumption about free will lost much of its plausibility. Assuming epiphenomenalism, if not by necessity, would be against the logic of Occam's razor. Indeed, the rationale behind epiphenomenalism was always the strong belief that the laws of physics simply leave no room for anything else. First it was the determinism of classical physics that ruled out any possibility of free will. Nowadays, it is still generally held that nothing has changed in this regard in spite of quantum mechanics. Indeed, if indeterminism of QM is understood superficially, as a mere introduction of chance into physical dynamics, then it is true that brute randomness no more reflects any ``free will'' than the unwavering necessity of determinism: either random or predetermined -- it is not free in any case.

It is therefore exceedingly important to notice that our current paradigm is very much different in this sense. The recognition of this internal perspective, which is fundamental and not something that can be derived from external properties, in combination with inherent indeterminism of physical laws used to predict the agent's external behavior, invalidates the above arguments and naturally allows for the free will. Namely, randomness observed in external behavior of an agent no longer implies that the same behavior had to be also purely random from the subjective perspective too. Actions of an agent can be free from her internal perspective (within the boundaries set by Born's rule), and thus indeed truly free, while from the external perspective this freedom would be mathematically interpreted as the intrinsic randomness of the quantum mechanics. It should be obvious that adhering to the probabilistic Born rule in principle still leaves immense freedom for the agency: there are myriads of fairly probable outcome sequences of tossing a coin thousand times, and each of these sequences reasonably corresponds to a fair coin. Similarly, at almost any moment there is a huge number of possible behavioral responses of an agent, where all of them would be externally indistinguishable from being random and in compliance with the Born rule -- just choosing among these leaves more than enough room for true agency and will.\footnote{Another way to see this is to presume the opposite: that obeying the statistical Bohr's rule leaves no freedom at all. However, that could only mean that Bohr's rule leads to fully predetermined outcomes, which is obviously false.} (Here we implied parting with the naive idea that brain responses are deterministic in the classical sense: not only that brain functioning seems to involve processes that amplify low-level quantum randomness \cite{RandomnessAmp}, but also external stimuli certainly have quantum origins -- something to which we will return below.) On the other hand, while the freedom may exist, it must be limited in certain sense: we can imagine a human exposed only to extremely and artificially controlled stimuli and hooked up to an advanced MRI machine capable of predicting probabilities of behavioral responses -- the free will would be clearly constrained as the sequences of responses are not expected to systematically violate predictions in statistical sense. For the same reason, our behavior and choices have a very large predictable component, in spite of the free will. However, in relation to the hypothetical MRI machine, it is also important to note the following: the ``amount'' of free will an agent possesses is not necessarily the same in the highly controlled conditions of the MRI setup and in common, real-life conditions. The QM taught us that the world is highly contextual, and that many things depend on whether something is being measured or not -- there is no reason to presume that the free will is an exception.

Hence, since this paradigm does not prohibit free will, it would be here very difficult to motivate the epiphenomenal stance. And this interplay of the randomness (from the external perspective) and free will (from the internal one), can be now seen as a nontrivial consequence of the ``other minds postulate''. Such postulate here automatically put constraints both on the internal and external ``dynamics": the external dynamics cannot be (fully) deterministic in order to leave a room for the subjective freedom of conscious decisions (i.e.\ free will), while the physical laws governing the external side of dynamics do also constrain this free will to some extent (e.g.\ within the bounds of satisfying probabilistic Born law). These consistency constraints which allow for the free will of the observers and, at the same time, maintain that internal and external perspectives are in full accord, are necessarily much less trivial than in the case of epiphenomenalism. It could be argued that something very much like the laws of quantum physics is necessary to accomplish this, and, in this sense, one could hypothetically seek to derive QM laws from the ``other minds postulate''.\footnote{It is here difficult to avoid the question of why there are then any physical laws at all? Without further elaborating, we just note that the existence of law(s) can be potentially explained as an unavoidable emergent effect, resulting from the fact that, otherwise unconstrained, wills of different agents would be, in general, mutually conflicting. Mathematically, this could be translated into a requirement that newly created information is minimal, from which it should be possible to derive statistical Born's rule.} On the other hand, universes with laws that cannot satisfy the ``other minds postulate'' could maybe exist, but could accommodate just one sentient perspective, like the solipsistic universe of a dream.

The ``other minds postulate'' also resonates well with Wheeler's participatory universe idea: it only extends the basic idea that the ``answers'' we obtain from measurements must be consistent, however not only when we observe particle properties and inanimate matter, but also when we observe other observers. And there is little doubt that Wheeler, when speaking about ``observers'', had in mind conscious agents in the fullest sense - i.e.\ those capable of exerting their free will (e.g.\ when participating by posing the questions to the universe). Epiphenomenalism actually makes no sense in his essentially idealist world view.

Therefore, we have concluded that postulating the non-derivable internal perspective has automatically opened a door for (re)introducing the free will in a way that is consistent with laws of physics. Alternatively, we can start from the ``other minds postulate'' alone (which already incorporates the presumption about the internal perspective) and conjecture that laws of QM are necessary to provide the consistency of internal and external perspectives, required by the postulate. Here we have implied that conscious beings in the postulate must also possess free-will. This is not only a natural attribute of consciousness, but we will soon argue that closely associating free-will with the consciousness is one of the rare options to defy the conclusions of the no-go theorem in a logically satisfactory way.

\subsection{Defying the no-go theorem}

So far we have discussed how the internal and external viewpoints may be reconciled within the interpretations of the first class, both on the lower level of quantum events (discussing the Wigner's friend type of scenarios) and on the higher level (the interplay which naturally allows for the free will, albeit of a limited type).

But the key question, to elucidate the announced potential of the first class interpretations to defy the conclusions of the no-go theorem and in this way to consistently alow for the existence of consciousness, is yet to be addressed. Clearly, by presupposing the internal perspective of systems as an irreducible and non-derivable concept, required for existence of subjective experience, this paradigm explicitly violates the Neuroscience postulate, thus avoiding the troubling conclusion of the no-go theorem. Nevertheless, it is very instructive to consider more precisely where this modification of the basic axioms breaks the logical chain of the theorem. This is also interesting from the perspective of implications for strong AI.

One option is to immediately dismiss the conclusion of the Behavioral theorem (Theorem 1). Our paradigm now encapsulates consciousness as an ontologically real, irreducible notion, so this becomes a tricky but logically valid possibility. In spite that two agents behave in qualitatively indistinguishable way and both similarly report conscious states, we are this time logically allowed to assign consciousness only to one of them, guided by any additional principle (e.g.\ possessing of a human-like brain or not, or judging by the value of Phi as according to Integrated Information Theory \cite{Tononi}) or even in an entirely ad hoc manner. In this way we can simply preclude the problems implied by the no-go theorem, breaking its logical chain by insisting that a computer simulating brain responses would not, unlike brain itself, cause emergence of any subjective experience, no matter that both the brain and the computer produce the same behavior (for example, this is exactly a conclusion of IIT approach \cite{IITUnconsciousAI}). In the context of first class interpretations, we may further claim that only agents (physical systems) to which we assign consciousness are then ones entitled to internal perspective, or at least that they are the only ones in which this internal perspective has given rise to consciousness in the standard sense.

However, as already discussed, abandoning the conclusion of the Behavioral theorem comes with a certain price. We may imagine a situation in which a being, that behaves and reasons like a human in all relevant respects, implores us to understand that it is genuinely conscious and that it feels the emotions and pain very much, and yet we are adamant in our decision that it is just a senseless automaton, simply because it does not fulfill some formal criterion (such as possessing high enough Phi number). Such a scenario would be very much reminiscent of a literal witch-hunt, even more so if the verdict that it ``does not possess consciousness'' would imply that this agent is expendable, for example, to serve in some hazardous conditions instead of humans. The problem is that such a verdict cannot be checked nor proved even in principle (due to the other minds problem, as discussed in more details in Appendix) and thus abandoning of the Behavioral theorem calls for ethical and legal decisions which are not only dubious, but also based on hypotheses which are of unscientific nature. And worse than that, in the Appendix we have argued that giving up the Behavioral theorem effectively means sacrificing entire scientific method and choosing over it hypotheses which are ``cognitively unstable'' \cite{BoltzmannBrain}.

While the option that Behavioral theorem is violated cannot be entirely logically discarded (similarly as, for example, the solipsist option), it would be far more satisfactory if we could have a consistent framework where the Behavioral theorem is preserved and yet the consciousness can be accounted for without absurd implications that the no-go theorem brings. However, this is not trivial, even with accepting consciousness (or internal perspective) as fundamental: once we acknowledge behavior as the good criterion of consciousness, we can always simulate brain on a Turing machine to reproduce the equivalent behavior and we are seemingly back on the track all the way down to copying of the sand-symbols and writing of a number that produces subjective experiences. Even the randomness of QM does not help much by itself. To begin with, it is unclear whether this randomness plays any significant role in the functioning of the brain. And even if it does, we need more than that, once we have accepted the behavioral criterion. First, we would have to insist that artificial intelligence, based on the classical computing, is in principle and ever incapable of reproducing the behavior that brain produces -- a claim that is not easy to substantiate at this moment. And secondly, we have the problem that brain, even as a quantum system, can be anyhow simulated by a Turing machine (at least in principle, with sufficient precision) and at the end the behavioral output to stimuli could be chosen randomly in accordance with the Born rule -- it is very difficult to motivate why making this choice based on a computer-generated pseudo-random number could not result in the same final behavior as by randomness of the collapse.

In spite of these problems, we will demonstrate that it is still possible to keep the conclusion of the Behavioral theorem and nevertheless leave a room for consciousness (of course, we do not claim that our proposal is the only such possibility, but is sufficient to demonstrate that the first class of QM interpretations has the potential to get us closer to the solution of the hard problem even while obeying the Behavioral theorem).

To this end, we first note that in all the views which emanate from the subjectivistic understanding of the collapse postulate, the wavepacket reduction which selects actual outcomes from the possible ones is the key trait of the presumed internal perspective. This process with non-deterministic outcomes is also vital from the perspective of the introduction of fresh informational content into the universe. Therefore, within this paradigm, it seems natural to seek origins of consciousness in more tight relation to the collapse. Another layer of our conjecture would be to essentially equate the subjective collapse and the subjective experience -- in the sense to assume that subjective experience arises together with the collapse. In other words, to conjecture that not only the collapse requires the existence of this inner perspective, but also vice-versa, that the inner perspective, or at least its higher-level subjective experience phenomenon, requires the collapse. That we are, somehow, aware of, or thanks to, the process of creation of new information. In such view, the consciousness is inseparable from the active exploitation of this inherent freedom left for the agency, from taking choices or, seen externally, from the actualization of possibilities. That is, the consciousness is inseparable from the free will.

It is probably not immediately clear how this assumption makes the consciousness and behavior potentially compatible in the light of the no-go theorem. For, by the Behavioral theorem, a classical and deterministic computer must nevertheless be conscious if it is capable to reproduce the human-like behavior -- which again threatens to logically lead all the way to sand-symbols and numbers that possess subjective experiences. However, let us consider more closely a hypothetical android robot driven by a classical computer. In order to reproduce a qualitatively indistinguishable human-like behavior, the android must be equipped with sensors (e.g.\ cameras, microphones) which sample external stimuli. These stimuli then must be processed in an extremely complex and nonlinear way (e.g.\ by layers of neural networks at least to say) which makes it highly difficult to trace the cause-effect chain from the combination of the initial stimulus and the current state of the computer to the final behavioral response. While the entire process remains deterministic in principle, it potentially amplifies the fluctuations of input stimuli in obscure ways: it might be the color of only a few pixels in the visual field which will determine which of two (or more) chains of responses will be triggered. Now, it is sufficient that the environment is quantum-mechanical in nature (as it normally is), so that the entire system becomes non-deterministic: the two chains of behavioral responses become entangled with the color of the pixel. Externally, we will see the robot executing either one or the other of the two chains. We may say that the behavior was determined by the random color of the incoming light. But, as well, we may say that the color of the incoming few photons was a posteriori (in a logical sense) fixed in accordance with the \emph{chosen} response of the robot. From the internal perspective of the robot-system, the collapse of the photons' wavelength has occurred, and this new information of the definitive photon color is created in the universe. In this paradigm, this is sufficient, combined with the complexity of the android's cognitive system, to be perceived by the robot as his own choice, and sufficient as the basis for subjective experience. In a similar way, we ourselves can exploit quantum randomness and have subjective experience even if (which seems unlikely) the functioning of our brain is entirely deterministic in the classical sense and has no internal sources of randomness.

In the context of our theorem chain, this means that the second theorem, stating that a sufficiently advanced android controlled by a classical type of computer may possess the same level of subjective experience as a human, still holds, even in the QM setting. But the chain definitely breaks up at latest with the Theorem 5, which asserts the possibility to isolate the android from the environment inside a virtual reality simulated on a classical computer. At that point there is no more room for creating information or taking choices, as everything is entirely deterministic. From an informational perspective, everything becomes frozen. And thus, under this conjecture, from that point on in the sequence of equivalences analyzed in the no-go theorem, there is no longer room for any subjective experience.

A few comments are due at this point. Naively, one could object that our environment is effectively classical and that the incoming photons in the above android example probably have decohered long ago.\footnote{Decoherence is often poorly understood and used as a general buzzword to dismiss any significance of the QM nature of the universe, while in truth it is nothing else but the process of entanglement with a larger system in which everything remains just as quantum as before.} It is important to stress that whether the photon was in a pure state superposition of two wavelengths, or in a mixed quantum state nontrivially entangled with some other physical systems (i.e.\ decohered) is irrelevant for the argument. Namely, having the android response correlated merely with a photon, or with a somewhat larger or more physically distant part of the environment, makes no difference.

A more delicate question is whether this conjecture is equally applicable to all QM interpretations of the first class. We might imagine that the android's CPU is enclosed in an informationally sealed box (as in the Wigner's friend experiment) and that the only information allowed from the outside are the sensory inputs arriving from, for example, a camera. If we follow relational (or Brukner's) interpretation, then it is unclear how much speaking about the external world has any sense from the internal perspective of the enclosed system in the box (having in mind the ``no facts of the world'' doctrine). In particular: what difference could it make whether the data incoming from the camera is a genuine picture of the environment or just a virtual reality perpetually simulated by a classical computer? It seems to us that insisting on the collapse and the change of informational content as the source of consciousness requires a framework with a unique reality and ``objective facts'' (even if our understanding of subjective timeline has to be modified a bit).

However, if the environment can play some sort of role in consciousness, as argued above, where is this subjective experience ``located'' (if not entirely in the agent's body/brain) and how that could be intuitively understood? And also ``when'' the subjective experience occurs?

So far, we tried to address such questions of intuition through the Wheeler's idea of participatory universe. But, Wheeler's idea in its basic form cannot help us a lot with dilemmas related to the Behavior theorem since it seems to imply a fixed, discrete (integer) number of participant-observers (admittedly, we are not aware if Wheeler has even considered this aspect of the problem). In such framework it looks impossible even to discuss the question at what point an external physical system (whose reality and properties are only a consequence of information that participant-observers possess) can become an observer itself. In general, if we want to keep the idea that consciousness is present on a fundamental level and yet to allow that any system which behaves indistinguishably from Wheeler's observers can be an observer itself, there are essentially two roads to follow. One is panpsychism. As it basically assigns some level of consciousness to any physical (sub)system, it might be seen as an extension of the relational QM interpretation where the internal perspective (from which the collapse occurs), granted to any physical system, is recognized also as the basic form of consciousness. Panpsychism comes in many flavors, but the basic versions again assume some level of reductionism and ontological existence of matter (with added experiential properties). Due to the former, not only that panpsychism faces the serious ``combination problem", but also, in our context, it does not seem able to support, in the light of the no-go theorem, our hypothesis that quantum randomness and free will are at the core of the consciousness phenomenon (thus it is unclear to us how this approach could reconcile the no-go and the behavioral theorems).

The other option, opposite to splitting consciousness into tiny atomic pieces as in panpsychism, is to assume that, on a deeper level, there is somehow only one single underlying consciousness, stripped of all personal prerogatives, instead of a multitude of separate ones. Individual consciousnesses in such picture would be merely particular expressions of that single consciousness, mutually interacting as in the Wheeler's idea and observing the whole from distinct internal perspectives. Intuitively, this idea could be, in its most simplified form, likened to a collective dream (or a sort of "schizophrenic" dream) which derives its consistency and strict laws of physics (at macroscopic level) from the very fact that it is collective -- that there is symmetry between different observers (i.e.\ perspectives) and that each observer is entitled to have a consistent account of the world external to him (in effect, from the ``other minds postulate"). On the other hand, any system capable to process information and to respond to stimuli becomes another perspective through which the underlying consciousness can exert its will (while remaining consistent with the ``other minds postulate") and become aware on a different, ``localized'' level.

As space and time itself would have to be the constructs within this large consciousness, there is nothing inherently troubling if we cannot pinpoint where or when the individual consciousness actually occurs (it should be clear that no conflict with relativity can occur, just as it does not occur in any wavefunction collapse). The following analogy might help us to cope with non-intuitive time ordering of events as seen from different subjective vantage points. We may imagine writing of a novel, where the content is being made up as the letters are being written. But the novel is special, as each of the involved fictional characters becomes conscious and aware of the parts that pertain to him, as soon as these are written. Moreover, each character subjectively experiences that he is taking part in the writing of these segments. Needless to say, the book must be sensible and consistent in every detail. Alternatively but similarly, we may imagine that the author of the novel suffers from a form of ``creative schizophrenia": as he writes a part about some character he puts himself into that role so literally that he effectively becomes that character for a period of time, experiencing what the fictional character does and writing down in the book the decisions that the character ``made". The order in which the book sections are written need not be in any strict relation with the time order which characters in the book experience (as long as everything is logically consistent). If a new character is introduced in the third chapter, but is introduced ``retrospectively", describing some events years before the plot of the second chapter, the character will have experienced these events in the past, no matter that these were decided and written ``now". As the matter of fact, there are numerous (or even infinite) orderings in which the same story might be told (for example, the plot as seen from the aspect of any of the characters), and in some sense, there is no meaning to claim which is the ``true'' one. Time has sense only from the subjective perspectives of the characters in the book and even for them it has some common meaning only to the extent to which their stories interleave and logically affect each other. Following this analogy, in this view there is a single underlying consciousness and a single reality (effectively ``made of'' information) that is being created by that consciousness and through the perspectives of all of us (again, information content is not frozen as is in deterministic frameworks, which according to Wheeler deserves the use of the word ``created"). And being a part of this creative process is what we subjectively experience.

Developing further this idea would be out of the scope of the present paper. We only note that emphasis on randomness as opposed to determinism as a prerequisite of subjective experience might have some curious implications. The closer an agent is to obeying deterministic rules (from the external viewpoint) and the more restricted possibilities for her responses are, the more she is lacking aspects of subjective experience, i.e.\ the closer she is to a philosophical zombie. It is an intriguing though an entirely speculative idea that this relation could be extrapolated to every-day life: those individuals who lead a free life, full of opportunities and bold taking of choices might have stronger subjective experiences (and in some sense be ``more conscious") than those always sticking to a routine.

Returning briefly to the no-go chain of theorems, it should be noted that although the breakdown of the Theorem 5 was more essential, Theorems 3 and 4 also no longer necessarily hold in the context of the first class interpretations. Theorem 3 (on temporal correlation) concludes that subjective experiences must arise synchronously with the evolution of the physical system to which they correlate. However, in the above discussion of the Wigner's friend scenario, we have pointed to a variant of explanation in which Wigner friend's experiences (i.e.\ Alice's), from the reference point of the Wigner (i.e.\ Bob), cannot be said to appear synchronously with the evolution of the friend's brain wavefunction.

And in the context of our latest proposal that the actualization of possibilites (which appears during the collapse) is a necessary prerequisite for consciousness, this assumption affects the conclusion of the Theorem 4 (the Repetition theorem). First of all, the very idea of an ideal repetition of the system's evolution is problematic in the quantum context. If we understand that as the identical repetition of the evolution of the system wavefunction and speak in the context of the first class interpretations, then obviously having the wavefunction evolve in the same way does not mean that the sequence of collapses will repeat either. Alternatively, if we somehow force the system to repeat the same sequence of classical states that the system underwent in the first run (as a result of the initial spontaneous evolution and collapses), then what were ``free'' choices and randomness in the first run are no longer such in the second. In turn, if the first run was accompanied by subjective experiences, the second will not -- according to the conjecture.

The bottom line is: once we give up the tendency to derive the inner perspective from the externally visible dynamics, we see that possibilities for explaining the hard problem open up, at least in principle. In particular, the options are limitless if we are ready to accept the violation of the Behavioral theorem (which becomes a logically valid option in this new context), but also there seem to exist consistent ways out even if we decide to keep the conclusions of this theorem.  Of course, there are many remaining problems to be encountered along any of these alleys, but these do not seem even remotely as desperate as attempts to solve the hard problem in physicalist terms, without invoking QM.

\subsection{Common criticism}

There are some common objections to interpretations of QM which extend in this and similar directions.

One is that incorporating subjective experience (i.e.\ consciousness) as a fundamental, non-derivable concept, requires a huge extension of the ontological basis of reality, which is, a priori, certainly less preferred by the Occam's razor criterion than the monistic ontology of physicalism. As we initially presented it, the duality of the two QM postulates, one related to system's subjective-internal and the other to the objective-external viewpoint, is indeed reminiscent of the Cartesian dualism. But, while it might be formally possible to uphold such dualistic reading of quantum mechanics, it is by no means a natural one.

Namely, we should note that the very concept of a well-defined, truly external perspective is a dubious extrapolation, certainly not supported by experiment. In the bottom line, every experiment ever performed yielded only a subjectively perceived outcome. Thus the external perspective we were frequently referring to eventually boils down to a subjective perspective again, only from a viewpoint of a system different (and external) to the one being considered. Therefore, in our opinion, the most natural reading of quantum mechanics which has the potential to deal with the hard problem is again monistic: the true ontological reality should be attributed only to these internal perspectives of systems (that is, to subjective experiences). To consistently reconcile all these internal perspectives, that is, subjective experiences of all coexisting and ``participating'' observers, there must be an ideal correlation of the internal and external accounts, expressed by our ``other minds postulate''. For observers Alice and Bob to coexist and interact, with both having subjective experiences and some extent of free will, it is necessary that Bob is represented in the referential system of Alice and that Alice is able to find an explanation for observed Bob's behavior from her vantage point (and vice-versa). Furthermore, Alice's and Bob's accounts of events in their common exterior should be also consistent. It is the mathematical consequences of these consistency requirements and of the symmetry between different observers (internal perspectives) that is, in this view, behind the ostensible objectivity of the world exterior to us. Postulating existence of independent external reality (e.g.\ matter) is likely to be entirely superfluous and thus should be avoided. What we call matter is here seen as a consequence of participation of (infinitely) many observers or points of view in defining the reality and there should be no need to resort to dualism.

Another common objection is that non-physicalist views like this one are unscientific, at least due to being unfalsifiable. While the charge of being unscientific per se often draws solely from personal philosophical biases and from misunderstanding of concept of a rational explanation\footnote{"Rational explanation'' is often used as a misnomer for ``hard-line materialism", while in truth it denotes an explanation that successfully accounts for all the data, with resort to as little starting assumptions as possible (Occam's razor).}, the problem of falsifiability deserves to be addressed in more detail.

First of all, we must keep in mind that even being unfalsifiable is much better for a theory than being falsified. And it is this latter danger, that seems to have befallen physicalist views in the context of the cognitive hard problem. The entire idea of physicalism (especially in its purest form expressed in the Neuroscience postulate) hinges on a never proved assertion that the (illusion of) consciousness can be derived from the mere motion of matter. Here, ``never proved'' is actually a serious understatement -- there is no a hint of a true idea how the gap from motions to (illusion of) qualia could be ever bridged (wishful thinking aside), to contest numerous arguments (of which the no-go theorem is only one) suggesting that this was an irrational expectation from the outset. If it was about any other philosophical view but physicalism, such a view would have been long ago discarded as a non-serious contender for the theory that could explain the universe, in other words, would have been taken as already falsified by its inability to account for consciousness. When seen in this light, even the non-falsifiable but consistent-with-data positions should have a clear edge over physicalism.

But, we do not even think that views presented in this section are truly unfalsifiable, at least because they do have certain verifiable logical implications or hints, if not full-fledged predictions. We first note that non-physicalist views were, to a large degree, thought to have been falsified, by the temporary historical success of mechanistic determinism of classical physics. This fact, together with the consecutive breakdown of the classical paradigm, showed that these ideas could be falsified in principle, but are not. In fact, the sudden departure from the determinism of classical physics into the randomness of observer oriented quantum mechanics can be thus taken at least as a slight argument in their favor. More specifically, indeterminism itself can be taken as a prediction stemming from non-physicalist considerations (e.g.\ more than two millennia before Planck, the Greek philosopher Epicurus has in some way anticipated QM by predicting random ``swerving'' of atoms, to account for the free will).

But there are also some more concrete predictions and implications than the indeterminism of quantum mechanics. A different form of fundamental laws of physics (of the particle physics and of the hypothetical theory of everything) can be expected if the universe is built of matter, and if it is ``made of'' information shared by participating observers. It is hard to tell how exactly different these should be, but in the former case we could expect dynamics to be governed by very specific ad-hoc formulas, given as predefined properties of matter, much like it was the case in Newton's dynamics. Also, it is arguable that a hypothetical success of string theory could be more suggestive that the universe is based on matter-like constituents (strings) than it would hint towards universe made out of information (at least intuitively). Contrary to that, if the matter is only a derived notion, its dynamics should rather follow from principles of symmetry and information. A hypothetical grand unified particle model with action (almost) entirely derivable from space-time and internal symmetries, could be a tell-tale sign of a non-physicalist ontology. In such view, particles would most naturally \emph{be} unitary irreducible representations of the corresponding symmetry (super)groups, instead of being ``pieces of matter'' with symmetry properties that merely \emph{correspond} to such representations. In turn, irreducible representations could be mathematically seen as tiniest bits of information that can separately exist in a universe with a given symmetry. The unitary state evolution, governed by such a dynamics derived from symmetry principles, would then merely express conservation of information: new information is being created via indeterministic collapse in measurements, while in between nothing ``new'' happens; Born rule can be seen as prescribing the maximal probability to events that correspond to the least change of information content; the existence of the probabilistic law can be seen as a way to regain determinism and thus conservation of information on large scales (i.e.\ after averaging over an ensemble), while maximally suppressing creation of new information and constraining this essential process to the events in the micro-world. Furthermore, much like the anthropic principle, the ``other minds postulate'' automatically solves all fine-tuning/hierarchy problems: the values of masses and coupling constants must be appropriate to attain consistency of internal and external perspectives (which is not possible if natural constants do not support a reductionistic explanation of existence of intelligent agents). However, the most natural mechanism to establish the proper values of natural constants in a participatory universe paradigm would be by invoking the collapse postulate. This further indicates that a formulation of the Standard Model (or a GUT model) might exist where all of the constants (couplings, mixing angles, etc) could be expressible in terms of some field values -- this would allow for the possibility that we have collapsed the constants to proper values by mere observation of the external world (instead that these values are miraculously, by chance, tuned to favor our existence). We remind that, in this view, it is correct and logically justified to derive physical properties, including the values of natural constants, from the fact of our subjective existence, rather than the other way around, as is usually attempted in physicalist frameworks. For this reason, the ``other mind postulate'' can also naturally provide solution for many other problems of low probability, without need to postulate existence of infinitely many universes (as requires the usual combination of multiverse and anthropic principle): e.g.\ the extremely low probabilities for all conceivable origin-of-life scenarios turn into substantial values when recognized as conditional probabilities. Namely, given that the agents with subjective experiences exist, the wavefunctions in the past can collapse only into states consistent with this fact. In practice, this means that when exploring the evidence of the life origin, we must always encounter (i.e.\ obtain as measurement outcomes) findings which are compatible with our existence. If more then one history is compatible with all information up to a given moment, then according to QM all such histories must be in superposition; once we find a way to distinguish between these possible histories by a measurement, and only then, the history becomes indeed defined, as in the Wheeler's delayed-choice phenomenon \cite{DelayedChoice}. This seemingly strange order of logical causation (i.e.\ events in present defining causes in the past) is surely counterintuitive for those used to classical reasoning (as is the Wheeler's delayed-choice experiment counterintuitive itself), which is reflected through another common objection: the erroneous assertion that these interpretations require existence of an (intelligent) observer in the past to collapse the wavefunction, at the time before emergence of life. Here, as always when dealing with matters of quantum mechanics, one should avoid the logical fallacy of confusing contradictory with merely counterintuitive but otherwise rational ideas.

We are aware that most of the statements above require more detailed explanations and possibly some nontrivial research, which is entirely out of the scope of the present paper. Our present motive was just to illustrate that each of the two ontological viewpoints (materialist and idealist) leads physics research in a different direction, and hence both proposals are scientific and to some extent falsifiable. Certain aspects, on the other hand, cannot be addressed by the scientific method by definition, and are thus indeed unfalsifiable in the standard sense. Scientific method presupposes the third-person perspective and is therefore incapable to deal with the internal perspective and its properties. Besides, our idea of scientific "explanation" usually implies a cause-and-effect way of thinking and a belief that everything can be reduced to more fundamental principles. However, such reasoning seems to be applicable only in a narrow, macroscopic domain and when considering average values of certain properties definable from the third-person view. Expectation that the same approach should be applicable also to properties of subjective experiences and to entities like qualia, could be simply pointless and wrong, in a similar matter as it is pointless and wrong to ask why a radioactive particle has decayed precisely now and not a bit later, or to ask why the particular photon has passed a half-silvered mirror while the other has not. In turn, this means that qualia themselves will likely never be accessible for scientific studies (only our reports of qualia), but this by no means implies that we should deny their existence.

\section{Summary}

In this paper, we dealt with the hard problem of consciousness from the perspective of physics. We explored how the basic physical premises about the universe influence prospect of solving the hard problem, and, vice-versa, what the difficulties encountered en-route to address the hard problem can tell us about the nature of our universe. The paper was naturally divided into three parts.

The first part concerned the hard problem within the framework of classical physics. Analyzing this perspective of the hard problem is nowadays of huge importance simply because it is still the predominant view that the phenomenon of consciousness can be explained entirely within the physicalist paradigm based on classical physics. We made an attempt to scrutinize this stance by providing a (reasonably) formal framework for the analysis of the problem. We have formulated the Neuroscience postulate, which summarizes the prevailing contemporary position on consciousness, and argued that it inevitably leads to the statement of our ``no-go theorem": that a human brain cannot have any greater power to generate subjective experience than a process of writing of a certain big number. A lengthy logical argument was given in the form of a sequence of intermediary ``theorems", accompanied by the corresponding ``proofs". Once the conclusion of the theorem was established as apparently true, we went to further analyze its impact on the likelihood that the Neuroscience postulate could be ever reconciled with our subjective experience. We focused on the mapping between the brain states and the content of the subjective experience, and used the ``no-go theorem'' to translate the entire issue to the relation between the ``number being written'' and this content of the experience. Rephrasing the question in this new context allowed us to show that none of the available options (under the premise of NP) can be compatible with our introspection: neither assuming that consciousness is nothing -- in which case we realize that notions as ``red'' or ``pain'' would have never existed, nor assuming that subjective experience truly emerges (either as an illusion or not) -- in which case a well-defined mapping must exist and be experienced when writing a number, which then not only bears a clear magical connotation but also eludes consistent description. Finally, we interpreted the conclusions reached in the first part as a clear sign that the Neuroscience postulate must be abandoned and explanation of the consciousness sought for in a broader ontological context, where a natural candidate was, of course, quantum physics.

The presented no-go theorem has some elements of a more elaborate and formal Chinese room argument (in spite of the final conclusion which is contradicting Searle's personal convictions). While we cannot hope for the impact Searle's thought experiment had -- in part a result of the concise form of his simple and yet profound argument -- we firmly believe that our formal and detailed approach has some advantages on its own. While the Chinese room argument revolves around the notion of ``understanding'' which relies on its intuitive meaning, we, in the end, concentrate on the concept of mapping between a number and the content of the subjective experience. Since such mapping must be obviously nontrivial (if exists) and is hardly a natural construction, the conclusions are much harder to ignore (or to dismiss either by hand-waving arguments or by pointless analogies). The formal exposition in the form of ``theorems'' facilitates analysis and helps identify the pivotal points of the entire logical chain. Moreover, as it is often the case with no-go theorems, even plausibly invalidating the theorem claim could be highly insightful and of immense importance, as it would necessarily pin-point some essential prerequisite for consciousness. For these reasons, we believe that properly addressing the issues raised by this no-go theorem has certain potential if not to bring the discussion on this matter to some conclusion, then at least to better classify different views with respect to the part of the proof they disagree with.

Since the conclusion of the first part of the paper was the necessity to go beyond classical physics to explain the consciousness, in the second part we turned to quantum mechanics. We noted that different interpretations of QM have different potentials to bypass the constraints of the no-go theorem and to offer a framework that could encompass the phenomenon of subjective experience. In this, we concentrated only on pure interpretations, i.e.\ the ones that are mathematically equivalent to the standard QM formulation (and left the rest to the experimental test). We pointed out that, for the evaluation of different interpretations, of great importance is a natural mapping which exists between particular interpretations of QM and predominant positions on the hard problem of consciousness. We explored this mapping in detail, revealing that the class of interpretations which acknowledges the necessity of the collapse postulate on equal footing as the postulate of unitary evolution corresponds in a straightforward manner to cognitive science viewpoints which fully acknowledge the hard problem and insist that third-person description is not sufficient per se. On the other hand, interpretations that find the collapse postulate to be superfluous (e.g.\ MWI), are natural counterparts of the cognitive science views maintaining that once all objective functions of the brain are understood there is nothing left to be explained and the subjective experiences (if it makes sense to speak of these) must arise as a direct logical consequence of the physical brain functioning.

Not surprisingly, we concluded that these two classes of interpretations bear the impact of the no-go theorem implications: the problems present in the attempts to explain consciousness based on the classical physics still hamper the interpretations of the second class, whereas the interpretations of the first class, due to already involving the non-derivable internal perspective from which the collapse occurs, have certain potential to account for the subjective experience (we dubbed the latter as ``hard-problem solution-ready").

In the third part of the study, we argued that the set of possible solutions to the hard problem is no longer empty if we accept the paradigm of the quantum physics and uphold the interpretations of the first class. However, even in this broader ontological context, when we give up determinism and introduce some subjective elements as fundamental in the universe, it turned out to be difficult to evade, in a truly satisfactory way, the problems brought up by the no-go theorem. The available options depended crucially on whether we insisted that Behavioral criterion for consciousness has to be still valid or not. The conclusion was that if we give up this criterion then a multitude of solutions appears, but not only that none of them is falsifiable -- worse than that, we lose logically consistent framework to evaluate any hypothesis about consciousness (i.e.\ we lose any hope of scientific or even merely coherent treatment of the problem). On the other hand, if we still require behavioral criterion to be the final arbiter of consciousness, it turns out to be relatively difficult to consistently reconcile this requirement with the rest of the no-go theorem in a way that nonetheless leaves some room for subjective experience. We pointed out that one such possibility seems to be to relate subjective experience with creation of new information that occurs during the collapse (and see that as an expression of free will, constrained by the Born rule). We proposed the ``other minds postulate'' as the basic principle of this paradigm, which intuitively well corresponds to Wheeler's idea of the participatory universe. Of course, a full-fledged solution to the hard problem could not have been expected on this level, but we hope to have done enough to illustrate the potential of quantum mechanics to open the door towards the understanding of the consciousness. In particular, one seemingly viable alternative to physicalism of the classical physics is quantum mechanics interpreted in a subjectivist way, with a special emphasis given to the emergence of new informational content which accompanies the collapse.

There are a few general conclusions that deserve to be singled out and emphasized once more. First of all, the problems formally explicated by the no-go theorem either require its reasoning to be logically invalidated with the same formal precision, or require us to abandon the (never quite substantiated) hope that consciousness can be accounted for in the usual physicalist framework. Furthermore, the no-go theorem suggests that we might have recognized deficiency of the paradigm of mechanical determinism irrespectively of the experiments revealing the quantum nature of the universe, simply on the basis of its inability to accommodate consciousness. The weakness of mechanical determinism is, paradoxically, in the fact that it is too well-defined and too constrained, allowing for strong and almost mathematical arguments which seem to dismiss it as a plausible option for an explanation of the subjective experience (and, consequently, of our universe). And quite generally, the more elements of mechanical materialism and determinism a paradigm has, the more there is to support the proof of the no-go theorem and less likely it is to find any hypothetical room for consciousness. This realization might turn upside-down some intuitive ideas about what are preferable properties for a QM interpretation. Deterministic and mechanically-materialistic ideas of the pilot-wave interpretation (irrespective of its non-local features) and determinism of MWI have, as a result, that these views become seemingly as implausible as is the classical physics when we take into account the existence of (illusion of) consciousness. On the other hand, certain subjective elements and some sort of vagueness in definitions which go in pair should not be seen as a merely tolerable feature of a QM interpretation (or formulation), but also as a necessary one, if we are to have an all-encompassing understanding of the phenomena in the universe.

But maybe the main takeaway from this analysis is that the deepest problems in cognitive science and the most profound questions in physics should not be treated independently, since they are inherently interconnected. Because the assumptions about the ontology of the universe and about the main principles guiding the evolution of matter (both of which are the subject of physics) influence the available explanations of consciousness, therefore also the empirical fact of the existence of subjective experience (or of the illusion of it) can serve as an indirect test of our hypotheses in physics and inescapably also affects plausibility of different interpretations of quantum mechanics.

\bigskip
\bigskip

\noindent
\textbf{Acknowledgments}

\medskip

\noindent
This paper would never have been written without inspirational discussions with Vladimir Prelovac, Cristi Stoica, Liviu Coconu, Adal Chiriliuc, Robert Ruxandrescu, Vladimir \v Ceperkovi\'c, and Marko Vojinovi\'c. Their constructive criticism (and, in these discussions, most of them tended to disagree with the above conclusions) led to the immense improvement of the presented arguments.

\newpage
\section{Appendix}
\setcounter{theorem}{0}
\setcounter{lemma}{1}

Here we will discuss in more detail arguments supporting the statements of the theorems.

\begin{theorem} (Behavioral theorem). Under the premise of NP, two agents behaving effectively in the same way when it comes to reacting to external stimuli, reporting about internal conscious/emotional states, having the ability to autonomously introspect and recognize its subjective experience and engage into related discussions -- must possess the same qualitative level of subjective experience.
\end{theorem}

\begin{proof}
First, we will consider this theorem in a broad ontological setting (i.e.\ ignoring the NP presumption) and only after that we will take into account the Neuroscience postulate and discuss how it further constrains our conclusions.

We begin by noting the ``other minds problem": we can be strictly sure only of existence of our own mind (or of illusion of it), and there is no, by definition, any objective experiment by which we could verify presence of ``subjective experience'' in any other (external to us) agent. This is a direct consequence of the subjective nature of the phenomenon and, in principle (i.e.\ without assuming any additional presumptions), all other agents we perceive in the external world could be ``philosophical zombies'' -- i.e.\ there is no chance to refute this possibility with certainty or by any objective means. Thus, when attributing the subjective experience to anyone else but ourselves, we must be aware that we cannot seek solid proofs to support our conclusions, but can proceed only by analogy.\footnote{While we are aware of additional arguments in favor of existence of other minds, such as the argument of best explanation \cite{OtherMinds},  all of them have, as a starting point, some elements of reasoning by analogy -- for, we would never tend to assign subjective experience to others, had not we been aware of our own experience.} (Strictly speaking, there is also a logically valid option to deny the existence of any subjective experience, including our own, in which case the theorem statement is trivially true.)

To go further from our own mind and the solipsist stance, one can pursue the following reasoning: ``I see people that look similar to me, they seem capable of similar things I am capable of, thus I'll assume that they have also similar inner, subjective experiences as I do.'' Now, there are a couple of main directions from here.

First, we may proceed to make a generalization based on the ``capable of similar'' part of the above statement -- in other words, to take the behavior as the criterion. The behavioral criterion stands out, since we experientially relate our behavior with our subjective experience (via our mental actions, capabilities and properties, appearing either as the precursors of the related behavior or accompanying it). This is true both for our deliberate actions (before or along with saying something we experience a related stream of thoughts) and for most of our involuntary responses (we feel anger or shame as we blush). We \emph{know} that our behavior reflects our inner life. In spite that this connection cannot be established via objective experiments, it is the only type of correlation we can expect to have with a subjective phenomenon - from the first-person type of knowledge. It is thus natural to judge the presence of subjective experience in others also by looking for signs of such cognitive activity which is correlated with these forms of behavior. Besides, there is a relatively natural mapping from different aspects of behavior to different aspects of consciousness: for example, if android Data does not express feelings, we need not entirely deny its consciousness, but we would just deny him the emotional component of it. In this way, based on experimentally (by introspection) established correlation of our behavior with our internal experiences, we have established a potentially consistent criterion to determine what is conscious and what is not -- possibility to judge, in principle, the existence of an agent's consciousness by analyzing the information exchange (response vs. stimuli) between the agent and the rest of the universe.\footnote{This formulation emphasizing the information exchange is capable also to encapsulate some pathological cases, e.g.\ when an agent is still conscious but with totally impaired motoric abilities: in such case, we might need special devices (e.g.\ NMR) to read his responses to stimuli, but that would still fit into the ``information exchange'' definition. It also covers the cases when an agent intentionally subdues behavioral responses so that these become invisible to the naked eye.}

Another option is to take the ``look like'' path. It means to make a generalization based on likeness, of course not only of the external sort, but also on the similarity of the inner workings. In this way, presented by another (living) human, which is both from the outside and the inside quite similar to us, there might be no problem to conclude that he/she is conscious.

But we cannot simply redefine ``conscious'' to become a synonym for ``human", it certainly must have an independent meaning. However, without behavior as a criterion, it is not clear how is it possible to establish what is the relevant aspect of this similarity in the context of the ``consciousness". To start with a trivial example: is a human still conscious if someone cuts his leg off? How do we know that the amputee is still sentient? Shell we ask him to tell us, or consider how he behaves? Certainly not, as we have now chosen to dismiss these behavioral responses as credible criteria and follow the other path. The amputee might have all of the sudden become a philosophical zombie, now speaking about inner experiences merely as an automaton. Indeed, if we do not start from the premise that consciousness is manifest through behavior, then not only we cannot prove what is conscious, but we cannot even have a consistent framework to deal with the problem. There is no way to establish even that the brain is any more related to subjective experience than, for example, a spleen.

As an illustration, let us formulate a couple of competing hypotheses about consciousness, which do not take behavior as a criterion. Let Hypothesis 1 be that some specific property of the neural system, having to do with the way it processes information -- for example, something about the way the neurons are interconnected -- is what is responsible for the emergence of consciousness. Moreover, let the conjecture be that this particular sort of interconnectedness in processing information is, in general, a prerequisite for the appearance of subjective experience\footnote{The Integrated Information Theory \cite{Tononi} is one example of such hypothesis.}. Let the Hypothesis 2 be that only agents with two legs possess subjective experience, and Hypothesis 3 that being born with a spleen\footnote{There is a rare congenital condition that a baby is born with the missing spleen, but with no other developmental abnormalities. We included a congenital anomaly as the third option, as in principle one could cut off his own leg and subjectively check if he is still conscious -- though his discovery would be again non-verifiable and thus meaningless for everybody else.} is the only thing that determines whether an agent is conscious. First, note that within general ontological setting (dualistic ones included) all three hypotheses could be true in principle: if the consciousness is ontologically real, it might be simply that some things/agents possess it and some not, and there is no any way to discover how this fact is correlated (if at all) to any objectively measurable property. And not only that all of the three hypotheses are experimentally irrefutable, but it is also that none can be deemed more plausible than the others. Ostensibly, Hypothesis 1 sounds most plausible, sort of scientific. However, what at best can be proved for conjectures of this type is that such-and-such property of the neural system is required for certain type of (cognitive related) \emph{behavior} (e.g.\ in humans), but there is not even in principle any possibility to prove that the property in question is correlated with \emph{subjective experience} (one can only \emph{postulate} that such-and-such dynamics of the system correlates, or causes accompanying subjective experience, as this part can be never proved nor objectively tested due to definition of subjectivity). Besides, it is only the fact that neural system is crucial for our behavior and especially for the parts of behavior that we intuitively see as traits of consciousness, which makes the first hypothesis \emph{sound} more plausible than the others. If we truly ignore the behavioral criterion then the neural system is no more related to consciousness than the digestive or immune one, and all these hypotheses are on equal footing. Even partial use of behavioral criterion is not sufficient: let us assume that we decide to use behavior to test our hypotheses on a set of agents and realize, on one hand, that the absence of a limb or a spleen does not seem to relevantly influence behavior, while on the other hand, that each of the agents who exhibited conscious-like behavior also met the requirements of the first hypothesis. Could we prove the first hypothesis in this way? Unfortunately not, since the first time we encounter an agent that behaves in a conscious manner but does not fulfill the formal requirement of Hypothesis 1, it will be impossible to know whether the hypothesis was faulty (incomplete) or the agent is really a philosophical zombie.

In this sense, by discarding the behavioral criterion we lose any consistent framework to tackle the problem of consciousness and any hope to evaluate remaining hypotheses. In spite that some of the non-behavioral hypotheses on how to recognize conscious agents might be actually true, such could be never proved so and not even seriously argued to be so, and in this sense such approaches could be dubbed as unscientific (or ``cognitively unstable"\footnote{In \cite{BoltzmannBrain} S.\ Carroll introduced this term to qualify hypotheses which ``cannot simultaneously be true and justifiably believed''.}). In other words, if we live in a scientifically accessible universe, the behavioral theorem must be valid, and if we do not, then why bother to discuss consciousness or formulate any hypothesis. The reason why we cannot invert the argument and discard  behavioral type of criteria on the same basis (of not being independently and objectively verifiable), is the one already discussed: the behavior is inevitably singled out as the criterion that we experience and introspect as relevant.

There is another type of problem, related to any attempt to correlate subjective experience with something else than behavior, directly stemming from the previous. Let us consider two agents of identical external appearances whose behaviors seem to indicate the same levels of subjective experience (as observed during a prolonged period of time). They give qualitatively indistinguishable responses to external stimuli and reports of their inner life: for example, they express pain and joy in the same manner, speak similarly of their perceptions and emotions and discuss with the same passion and understanding about the philosophical problems of consciousness. Moreover, as this is also a form of report, both are in principle equally capable to, without external suggestion, independently realize their own subjective experience and arrive at the ``hard problem".\footnote{Some researchers attribute great significance to this ``test'' of being conscious \cite{ConsciousnessTest}.} On the inside, however, these two agents might function differently.

Assume now that we have adopted an exact and objective but non-behavioral criterion for consciousness, and that, unfortunately, it turns out that one of these two agents does not pass the test\footnote{For example, does not possess high enough value of Phi as defined by IIT.}. We might have even devised an instrument, a ``consciousmeter", based on this objective principle -- and the device clearly indicates that only one of the agents possesses subjective experience. Unless our criterion always identically coincides with the conclusions of the behavioral one (in which case these criteria are identical and the theorem is confirmed), such a situation must occur sooner or later. Trusting our hypothesis, we insist that one of the agents is a mere unconscious automaton and treat him accordingly, in spite of the fact that we cannot see any difference in the looks or behaviors of the two. And the poor agent begs us to understand that he nonetheless has feelings and implores us in despair to understand his position, but we stay adamant. Worst of all, in all that we are aware that our criterion is a mere hypothesis that we have actually never proved, and, moreover, which cannot be ever proved even in principle.

Obviously, such hypothetical literal witch-hunt would be ridiculous and ethically unacceptable in practice (legally even less so). Such a prospect should serve to clearly illustrate the pointlessness of formulation of this type of inapplicable criteria, which do not coincide with the behavioral one -- irrespectively of the ontological context!

Apart from making generalization along ``capable of similar'' logic, and generalization according to ``look like'' criterion, there could be an option to require both, or to require neither. But none of these is promising. In the first case, we would anyhow inherit bad traits of the generalization by the mechanism similarity approach: the witch-hunt of the agents that behave exactly in a human-like manner but have different inner workings (or are missing a spleen) would remain. In the latter case -- requiring nothing in order to conclude that two given agents have the same or similar level of subjective experience -- would imply that a toaster can go through emotional crises just as a human.

All this reasoning renders criteria differing from behavioral one dubious, inapplicable and sort of unscientific, irrespectively of philosophical presumptions (and, in particular, irrespectively of the Neuroscience postulate). Nonetheless, in a general ontological context as was considered above, the existence of meaningful rules about what possesses and what lacks consciousness is not a logical necessity, and the converse of the theorem statement still can be true in principle (i.e.\ agents without a spleen, amputees or AI androids behaving identically as humans potentially might lack human level of consciousness). However, in spite that rules of logic do not strictly prohibit such possibility, our conclusion is that discarding of the Behavioral criterion comes with a very high price tag, in any framework. (Although not crucial for the proof of the no-go theorem, this conclusion will be valuable in the third part of this study.)

Does this mean that we are doomed to redefine ``consciousness'' as a synonym for certain patterns of behavior, and that we should not seek deeper understanding of the phenomenon nor try to identify ``mechanisms'' essential for its occurrence? Not necessarily. A hypothesis about the underlying ``mechanism'' essential for consciousness that would, on the other hand, always effectively coincide with the behavioral criterion, could be not only of great importance for our understanding of consciousness, but also a valuable guideline in the context of AI research. However, the behavioral theorem would have to be the final arbiter for all such hypotheses.

Next we will consider what the introduction of the Neuroscience postulate and of the paradigm of classical physics brings on top of the previous conclusions.

Within physicalist framework of classical physics, objectively real are only matter and its properties/configurations. Agents themselves are merely certain complex configurations of atoms, and all their properties and their behavior (responses, reports) must be, in principle, understandable in terms of the properties and dynamics of these atoms. Indeed, these atom configurations exhibit very complicated patterns of behavior -- for example, emergent behavioral patterns that we call ``expressing of emotions", or, sometimes, these configurations report stuff like ``I see red", or ``I subjectively feel that I exist". But, in this paradigm, there can no longer exist any additional, ontologically real, element of reality that we call consciousness, which might accompany these processes and behavior. Whatever is that what we call ``subjective experience'' it must boil down to some motions and arrangements of atoms. And nothing else, since nothing else truly exists in this setting. This fact further constrains our previous conclusions.

To start with -- the solipsism is no longer an option. If the system of atoms that constitutes ``me'' can result in ``subjective experience'' (or an illusion of it) -- whatever that could be -- then an almost equivalent particle system constituting ``you'' must produce the same effect. Next, it is no longer logically allowed to have consciousness attached to agents/things in an arbitrary manner. Moreover, nothing can be ``attached'' at all to physical systems we call ``agents'' -- as the matter is all there is. But this ``attaching'' option used to be crucial for all of our previous example hypotheses that followed the reasoning of mechanism similarity. Such possibility was allowing these hypotheses to be potentially true in spite of being not only unprovable in principle, but also in spite of being completely arbitrary.

Not surprisingly, hypotheses of this ``unscientific'' type are not viable in the context of classical physics. We return once more to the comparison of the two identically behaving and externally looking agents, now in the NP context. To simplify the discussion, let one of the agents be human, and the other of a different inner structure, e.g.\ an android. In both cases all details of how the agent functions and processes the stimuli into responses are either understood, or understandable in principle. The classical physics is causally closed and cause-effect chain can be followed all the way through -- it is inessential if we maybe do not know, at present, all parts of the mechanism and even if we do not know about some basic matter constituents (e.g.\ some involved types of particles or fields), we can nevertheless be confident that all that happens is a causal chain of motions and interactions of particles within the agent that, like a clockwork mechanism, produce behavior. Behavior itself is again just specific motion and rearrangement of the matter constituents and we can be sure that nothing but possibly some other motions and rearrangements of matter can be produced in the entire process. (Thus we find quite unrealistic and in the domain of wishful thinking arguments that further exploration of brain functions could alter these conclusions and show how consciousness is ``secreted", without changing the entire philosophical paradigm -- unless, of course, ``consciousness'' itself is defined as a type of motion.) Now, what is the prospect for the claim that one of these agents possesses subjective experience while the other one does not?

Since the only difference is in the details of motions inside the agents, and there is nothing else to it, the only clear way to substantiate such claim is to \emph{define} subjective experience as some particular type of these internal motions which occurs in one but does not in the other. But this is hardly satisfactory: we may also redefine ``subjective experience'' to mean ``afternoon nap'' or ``a load of bananas", but these options clearly do not encapsulate the common meaning of the notion. In spite of the vagueness in the definition of the subjective experience, we do not know of any researchers or philosophers who claim anything of the sort that such-and-such motion of carbon atoms in the vicinity of some nitrogen atoms \emph{is}, per se, the subjective experience of ``red", or that the motion towards a slightly different rearrangement of carbon, nitrogen, hydrogen (and a few other types) of atoms \emph{is} the ``sensation of joy". Instead, more common is the assertion that such motions and/or arrangements of matter \emph{cause} these experiences to emerge. However, the physicalist framework of classical physics is very strict and conservative in this regard: the only thing these motions can cause are again other motions and neither ``red'' nor ``joy'' can appear in the formalism of physics. Allowing such-and-such motion of some carbon to also directly cause ``joy'' to appear, demands, strictly speaking, an extension of our framework at least into the paradigm of property dualism, in conflict with NP.

Though property dualism is no longer physicalism in its purest form and thus does not fit into our basic definition of the Neuroscience postulate, it makes sense also to consider an appropriate extension of NP which would allow for dualistic properties of physical matter. Then it becomes conceivable that certain types of matter motions or arrangements yield subjective experiences, but even in that case we are faced with the earlier charge of the unscientific approach: it is absolutely impossible to discover which of these motions correlate with what experiences, and any such mapping is doomed to be totally arbitrary, neither provable nor disprovable, even in principle! Unless, of course, we allow reports of the experiences to guide us in construction of the mapping. The rationale for this choice of guidance is again the same as before and obvious: we feel that it is the experience itself that causes us to report that experience and we introspectively know that our reports of experiences are correlated with our experiences. It is the experienced pain which prompts us to verbally complain or to (involuntarily) make a corresponding facial expression. And it is this thing which we experience, which we later report or in some other way express, that we call ``subjective experience'' and want to explain in the first place (and not to explain any motions per se). If we let us be guided by reports, only then we have a mean to correlate certain type of motion with a certain experience. But in that case we have already accepted the behavior as the criterion.

And what if we accept the behavior as the criterion from the outset? In the strict physicalist framework, behavior is again a type of complex motion and cannot cause anything to emerge apart from other motions. But, unlike in the previous case, there are notable positions which assert that behavior and reports of the experiences are indeed all that there is, the rest is an illusion. Whether the introduction of this ``illusion'' again implies excursion into the property dualism is open to debate. Even if it does, the mapping of different aspects of behavior into different subjective experiences is no longer arbitrary as was in the case of motions of internal mechanism -- here it is natural and direct (e.g.\ there is little dilemma whether a terrifying scream correlates with subjective experience of ``pain'' or of ``seeing red"). In any case -- property dualism or not -- the statement of our present theorem holds either way if we choose behavioral criterion: the two behaviorally indistinguishable agents cannot have different levels of subjective experience, no matter if the ``subjective experience'' actually denotes nothing, something (i.e.\ if we extend NP into property dualism) or an illusion.

To summarize: the above analysis seems to clearly indicate that the only way to allow two agents of indistinguishable behavior to have different levels of consciousness is to both: extend the paradigm of physicalism (at least into property dualism) and to deliberately accept non-scientific (i.e.\ cognitively unstable \cite{BoltzmannBrain}) positions on consciousness.
\end{proof}

\begin{theorem}
Under the NP premise, a sufficiently advanced android controlled by a classical type of computer may possess the same level of subjective experience as a human.
\end{theorem}

\begin{proof}
In the light of the Theorem 1, all that has remained to prove is that a classical type of computer can generate behavior qualitatively indistinguishable from a human. That this must be possible follows directly from NP assumption that the human brain operates as a deterministic machine governed by laws of classical physics. Since it is a biological machine of finite size, it is, in principle, possible to model it with arbitrary precision, simulate its operation on a classical computer and calculate the same responses to stimuli that a biological brain would produce. In the worst case, the simulation could, in principle again, simulate the motion of each particle in the brain. However, present understanding of brain functioning indicates that such a low-level simulation would not be even necessary to predict all details of behavior: it seems likely that modeling each neuron in the brain (accompanied maybe with modeling of some other biological systems, such as endocrine) would be sufficient for the task. Besides, note that, according to classical physics, there is also nothing that prevents us to scan the brain of a given individual as precisely as required and use that data as initial conditions for the simulation. Also, the part of the neural system responsible for sampling stimuli could be simulated just as well, providing that the simulated brain has the proper inputs.

Strictly speaking, to fulfill the behavioral theorem in its basic form, it is also required that the calculations can be performed in real-time. Though incredibly technically demanding task altogether, there are no indications that anything could preclude its accomplishment in theory -- and being possible in principle is the only thing that matters here. It might be that a classical computer capable of performing such simulation could not fit in the size of a human skull, but the theorem 1 does not necessitate this (after all, the computer can control an android body remotely).

Another indication of the validity of this theorem is the contemporary development of AI, which is interpreted by many to encourage our belief that even a practical construction of a classical computer matching the brain capabilities could be just a matter of time.

It is worthy of noting that the premise of a fully deterministic brain evolution can be relaxed without changing the conclusion of the theorem. Namely, if there are elements of chance inherent in the brain functioning, these could be also simulated, either on a classical computer using pseudo-random numbers, or maybe with the aid of some pre-supplied set of true random numbers. As a matter of fact, it is very difficult to motivate why good enough pseudo-random numbers would not be capable to reproduce qualitatively the same behavior as a brain functioning in part indeterministically. This further affects the case of real-life brain which is, as everything else, a quantum system. In principle, a classical computer should be capable to also simulate, with arbitrary required precision, the quantum evolution of a given brain state. After that, the output response could be selected from the resulting superposition of classical brain states by some (pseudo-)random number generator. This task would likely be so enormously computationally demanding that it is not clear if it would be feasible even in theory (given the finite resources available in the universe) -- but it is nevertheless hard to see why would these constraints be relevant for the discussion of consciousness. A potentially bigger obstacle is to obtain the initial condition for the brain state -- could such state be measured (in a relevant basis) and would quantum mechanical ``no-cloning'' theorem disqualify this prospect even in principle. However, the contemporary neuroscience position suggests that such fine-details of quantum states of the brain do not influence the computation of behavioral responses and that a mere high-precision scan on the classical level should provide good enough initial data. Therefore, it seems that the conclusion of the Theorem 2 follows from the Behavioral theorem even in a more general context of a stochastic or a quantum mechanical brain.
\end{proof}

\begin{theorem}
(On temporal correlation): Subjective experiences (or illusions of) that accompany dynamical evolution of a certain system (e.g.\ brain) are synchronized with that evolution, and do not depend on evolution of the system either before or after the considered period of the dynamical evolution (apart from due to the contents of memory).
\end{theorem}

\begin{proof}
This should be an immediate consequence of the NP and the framework of classical physics. Namely, subjective experience is here understood to be a direct product and a function of the brain operation (or of the functioning of the agent as a whole). Thus, there is no way in which subjective experiences could precede the corresponding brain dynamics and brain states. Similarly, it is also impossible for the experiences to appear with a (significant) lag, since there is nothing that could ``memorize'' brain state from its occurrence until the appearance of the corresponding experience. For example, were the subjective experiences entity with independent existence, this would not be necessarily so straightforward conclusion.
\end{proof}

\begin{theorem}
(Repetition theorem): If a physical system that gives rise to consciousness undergoes identical dynamical evolution more than once (fed by the identical stimuli and starting from the same initial state each time), each time the same subjective experience arises. Alternatively, if another system identical in every detail to the first one, undergoes the same dynamics, both evolutions must be accompanied by identical subjective experiences.
\end{theorem}

\begin{proof}
The theorem statement can be seen either to follow from the Theorem 3 (since experience arrises coincidentally with the brain dynamics and none of the evolution runs cannot ``know'' about the previous ones, nor can influence the repeated runs afterward), or by noting that if the consciousness is a consequence of system dynamics, the fact that a period of identical dynamics in the same system state has already occurred before cannot preclude or influence the emergence of the same experience again. Same goes if we consider a different system, but identical with the first in all relevant properties.

Surprisingly, the theorem statement runs opposite to some of the widespread expectations (not the majority, however). Namely, there is a quite popular subculture of people that take very seriously prospects of science-fictional type of teleportation or ideas of uploading one's personality into a computer. But, the theorem poses the following problem in this regard: a hypothetical identical copy of an agent (produced in an attempt to teleport him) would also then produce another instance of his subjective personality (with all accompanying experiences). This further means that copying would not affect the original agent, who might be as well unaware that was copied: the original agent will obviously still remain with his own personality and subjective experiences, correlated solely to his original body. If the original was to be disposed afterward (as the idea of teleportation requires), the original instance of the agent would experience that simply as his demise (instead of as being efficiently transported), with only a feeble consolation in the fact that his copy, which he cannot subjectively experience and is no more than a distant twin-sibling to him, persists. (The same problem plagues the even more far-fetched idea of digital ``uploading".) As a remedy, it is sometimes conjectured that the two identical instances of the physical agent would correspond to a single personal identity (i.e.\ single consciousness with two physical bodies), and that, thus, subjective experiences would be lived through only once (instead of each time again with every new copy or a new run of the identical evolution). This is also sometimes motivated by the claim that if all properties of the system that give rise to the subjective experience are the same in the agent and in the copy, then they must correspond to one and the same experience, and not to a multitude of experiences with the same content. However, in our opinion, this argument ignores the fact that not all properties are identically the same: position of the system in time and space is different for original and the copy. There are also attempts to advocate this position based on broader ontological presumptions and leading to more radical conclusions \cite{Zuboff}. However, it is easy to identify many conceptual problems with such conjecture which, at least, make it incompatible with the NP. To start with, there are problems related to non-locality. One aspect of this problem is related to the Theory of Relativity, and it is analyzed in \cite{Knight}. But probably more acute in our context is that such nonlocality seems to defy ideas of strict mechanical causality assumed by NP. For example, if the copy of a given agent is artificially constructed at a large spatial and/or temporal distance, it is purported that agent's consciousness nevertheless becomes immediately linked to both bodies, without any physical connection. This seems to be quite acausal for the standards of classical physics. Besides, it is quite unclear how precisely identical the copy must be to attain such connection. Matters become only more so hazy if we consider (illusion of) agency and its potential action over distance. The very idea could be probably more interesting in a broader ontological setting (e.g.\ where information is fundamental entity instead of matter, and where it might be argued that repeating of the same informational content does not give rise to a new instance of subjective personality). However, we find it irrelevant in the context of NP and thus of no influence to the validity of the theorem.
\end{proof}

\begin{theorem}
(Simulation theorem): Simulated android (or androids), existing only in a virtual reality realm simulated on a classical computer, can be, in principle, programmed to be as conscious as humans.
\end{theorem}

\begin{proof}
We wil first consider a case when there are two actual Data ``units": Data1 and Data2. Since each of the CPU's can be located in the external computers, nothing prevents these two computers to actually be a single computer powerful enough to do the necessary calculations for both. Obviously, no need to stop at two -- sufficiently powerful computer cluster would be able to control an arbitrary number of Datas (a sort of ``Data'' cloud), each of which would be fully sentient in that cluster as it would have been if its CPU was situated in the android's head.

As the next step, we note that the data ("data'' in lowercase denotes the common meaning of the word, not the android) originating from androids' sensors and arriving as the input into the cluster is, from some point on, necessarily digital (at least at the cluster input point). A powerful enough cluster should be thus perfectly capable of simulating this data as well, i.e.\ by generating some rich enough virtual reality for these androids. How can we be sure that the consciousness has not disappeared in this transition? Well, let's take, at first, that the computer is simulating the actual environment of these robots, and only replacing, from some point on, the data arriving from android's sensors with the calculated data. For some brief interval of time at least, the discrepancies between the computed data, and the actual one should not be significant. Thus, the androids should notice no difference, at least no significant difference: even if/when they are able to eventually tell the simulation from the true reality (e.g.\ by the level of detail) that should not affect the fact that they must remain conscious in the simulated surrounding if they were so earlier when their inputs were sampled by some external sensors (more precisely, in the NP context, there seems to be no room for any difference in this regard). Besides, we can, for a while, keep signals from the simulated Data's minds drive the actual robotic bodies ensuring the behavioral condition -- but this motion of the bodies would now become just an epiphenomenon, as nothing from the outside world is any longer entering the simulation. These behaviors of the physical android bodies can only serve us as an indicator of the sentience of agents that drive the bodies, and since the bodies do not influence the agents anymore, and that the agents have to be just as sentient with us watching as without, we may safely turn the bodies off (as a direct consequence of the NP, act of observation cannot influence reality). It seems that there can be nothing, within the paradigm set by the NP, that could invalidate the above chain of reasoning leading to the conclusion of the theorem.
\end{proof}

\begin{lemma}
The computer running a closed simulation containing conscious agents can be arbitrarily slow.
\end{lemma}

\begin{proof}
Once we have isolated Datas' world in a computer simulation, the following becomes clear: if anyone clicks a button ``pause simulation'' and then, after 10 seconds presses ``resume'' button, androids living inside of the computer will not notice any gap. Their lives and experiences will just take up where they have stopped upon the press of pause. Similarly, they cannot notice if the powerful computer slows down a bit -- they can measure the passage of time only relative to events within the simulation -- so if everything slows down, that must be unobservable to them. Even if the computer becomes disastrously slow, so slow that it takes a thousand years of real-time to simulate a single second of the virtual world, the androids would not notice that. In other words, in spite of widespread prejudices that computers performing such hypothetical simulations which include sentient agents must be futuristically advanced and of immense computational speed, we see that actually the computer cluster does not need to be fast at all, only of pretty huge capacity.
\end{proof}

\begin{theorem}
A system comprising a pattern of X/O symbols written in sand and a wooden Turing head mechanism performing certain automatic operations (of ``reading'' the symbol beneath, flattening out the sand, inscribing a new symbol and moving left/right one step) according to the instructions embedded also in the sand, can contain conscious agents living through human-like subjective experiences (i.e.\ agents belonging to the same equivalence class of consciousness level as humans).
\end{theorem}

\begin{proof}
Combined conclusions of the Theorem 2 and Lemma 2 indicate that the technical details of the computer realization cannot play any significant role, as long as the computer is capable of performing the required behavioral calculations. This arbitrariness certainly includes particular computer architecture (e.g.\ RISC or CISC), so we are also free to choose one from the variety of universal Turing machines, with a binary-alphabet. Note that we are here not interested in halting problem or anything similar.\footnote{For these reasons we also need not bother to discuss accepting/rejecting states of the machine.} Instead, we are concerned with a specific well defined finite problem: running the simulation of the Datas and their environment for a finite virtual time (the time might be defined beforehand, e.g.\ one virtual hour) and with finite precision. The program is not supposed to halt on any condition apart from reaching the end of the simulation time, which is again a well-defined condition which can be translated into a finite limit for the number of computational steps (therefore, it eventually halts for any given initial conditions). It is quite obvious that such task can be performed on an extremely powerful but finite-resources computer cluster, and thus it is also doable by using a finite length tape Turing machine (besides, finite number of execution steps implies that a finite tape is sufficient).

Turing machines are abstract mathematical idealizations, and as a such does not per se exist in the real world (especially not one governed by classical physics). We are free to choose any proper physical realization, and thus the one with a mechanical wooden head writing X/O symbols in the sand (instead of on tape) is also fine. We use universal Turing machine, where the simulation program itself is also written on the tape (i.e.\ in the sand), while the head itself is fairly simple: for binary alphabet it was shown that 18-state head is sufficient \cite{SmallestBinaryUTM} (or, we could group X/O symbols into n-tuples, mimicking in this way machine with bigger alphabet while choosing a head with only 2 internal states). Without changing the essence of the no-go argument, we might have also chosen Rogozhin's machine with 6 symbols and 4 states, which uses only 22 instructions \cite{SmallestUTM}. In addition to the simulation program, the initial state of sand symbols also properly encodes the initial state of the simulation, that is the state of Data(s) together with the virtual environment, at some given time $t_0$.

After the initialization, the head is put in operation: according to the instruction table and taking into account only the symbol directly below the head (it reads the symbol by sensitively checking sand level at a few points), it either replaces that symbol in the sand (by firstly flattening out the sand at the spot and stamping the new one) or just moves one step to the left or right -- and proceeds to the instruction indicated as the next one. It is well known and established with mathematical rigor that this procedure can be made absolutely computationally equivalent to the intricate operation of the futuristically advanced super-computer, for which we have already concluded that must be giving rise to the subjective experience of the Data-agents within the simulation. This seems to confirm the assertion of the theorem.
\end{proof}

\begin{theorem}
(The no-go theorem). The Neuroscience postulate implies that a human brain cannot have any greater ability to induce subjective experience than a process of writing a certain sequence of digits (writing can be done in an arbitrary way and on the arbitrary surface).
\end{theorem}

\begin{proof}
The outline of the proof is the following. First, we will modify the wooden Turing machine from the Theorem 6, so that it no longer ever deletes any information -- each new state of the tape is now written in a new line of sand, leaving the entire log of the calculation as the result of the operation. Then we will consider and compare two scenarios. In the first one we will re-run the same simulation, with identical initial conditions (i.e.\ identical initial state of both Turing tape and head), right below the previous run -- and this must give rise to a new occurrence of Data's subjective experiences. In the second scenario, we will consider a wooden head that merely copies the log of the previous run, and argue that there cannot be any difference of the two scenarios with respect to the elicited subjective experiences. The reasoning will employ a sequence of ``hybrid heads'' that gradually bridge the difference between computing and copying in this setup. Once we have established that the copying must have an equal effect on subjective experiences as computing anew, it is a trivial part to reinterpret copying of X and O symbols as the writing of a given (binary) number.

Therefore, let us consider a slightly modified version of the (Post-)Turing machine where nothing is ever erased: instead of replacing a sand symbol somewhere in a line, as the result of each such instruction the head writes again the entire row -- writing this time the new value at that position while simply copying all the other values. The new line is written just below the previous one, and so on. Furthermore, we will include in each written row also the entire information about the state of the head (the label of the current instruction and the position of the head). Mathematically, it can be accomplished in many ways: the standard methods for encoding the current state of the entire machine ("the complete configuration") were discussed already by Turing himself, and those most straightforward include writing some auxiliary symbols (i.e.\ apart from X and O) \cite{TMgeneral}. If necessary, we can devise for this a supplementary machine (which can be later seen as a part of the same mechanism) that will be activated after every operation of the main head, with the task only to traverse the current line from start to the end and to perform the symbol-copying, plus to add these few extra symbols describing the main head state and position. And, in order not to spoil the basic idea of having solely X and O symbols written in the sand, we can encode any auxiliary symbols by grouping X and O symbols in n-tuples: for example, we can simply pair all symbols two-by-two, giving us now effectively an alphabet of four symbols, two of which can be used as special markers and have a head that operate on pairs, instead on single symbols (obviously, we can combine more than two X/O symbols if necessary).

Though certainly not optimal performance-wise, this new calculation protocol allows that the entire information of the system current state is encoded in each line in the sand and, since nothing is ever deleted, the log containing each step of computation remains in the sand once the process is over. (We can decide to rewrite the line also each time the head moves, even if no symbol is altered -- but the reasoning is independent of these details.) These modifications do not change the computational abilities of the machine: the evolution of the simulated virtual reality and the of the Datas inside is being calculated just as before. As before, we can confirm that by allocating a certain number of X/O symbols for the ``display'' and thus be reassured that androids look as conscious as before, even in this now 2-dimensional giant tic-tac-toe like pattern.

Next, below the rows of the just-finished simulation, we will repeat the same simulation again, starting with a freshly written initial machine state in the first row. If we let the head do its job once more, the identical pattern of the computation log will be written all over again. To facilitate analysis, let us enumerate lines of the first run as $A_1$, $A_2$, $A_3$..., and of the current run $B_1$, $B_2$, $B_3$... The resulting sand pattern, after the two runs, would then be something of the following sort (though immensely bigger and more complex):
\begin{verbatim}
A1: XXOOOOXXOOXOXXXXOOXXOOOO...
A2: XXOOOOXXOOXOOOXXOOXXOOOO...
A3: XXOOOOXXOOXOXOXXOOXXOOOO...
A4: XXOOOOXXOOOOXOXXOOXXOOOO...
...
B1: XXOOOOXXOOXOXXXXOOXXOOOO...
B2: XXOOOOXXOOXOOOXXOOXXOOOO...
B3: XXOOOOXXOOXOXOXXOOXXOOOO...
B4: XXOOOOXXOOOOXOXXOOXXOOOO...
...
\end{verbatim}

The line identifiers $A_i, B_i$ do not appear in the sand (they are only written for our convenience) and we assumed (as proposed in the above modification) that the head operates on pairs of symbols as if being a single symbol. (The schema above is merely a sketch and should not be scrutinized. We did not feel as necessary to provide a too formal definition of our modified machine as it suffices to establish that such modification is possible.)

According to the Repetition theorem, as the $B$ rows are written the Datas will certainly go through the same subjective experiences once more, as they did while generating the $A$ rows. More precisely, while the wooden head was generating line $A_{n+1}$ based on the line $A_n$ it was already established that the sentience had to appear. Similarly, when the same head calculates row $B_{n+1}$ based on symbols in $B_n$, the same subjective experiences also occur (Repetition theorem).

Next, note that if the same head calculates the row $B_{n+1}$ based on the row $A_{n}$ instead of $B_{n}$, the consciousness must emerge in the same way: symbols in rows $A_{n}$ and $B_{n}$ are identical, so which one it is going to read is inessential (though it is time-consuming to go up and down those few billion of lines in order to write each symbol, we have concluded that CPU clock is irrelevant). More precisely, in doing this the head moves up to the corresponding symbol of row $A_{n}$, reads it and based on its value and the current state of the head (m-configuration) ``calculates'' the next symbol to be written at $B_{n+1}$. Even more precisely, according to the symbol just read at $A_{n}$ and the current head state, the gears in the head rotate so as to prepare either the X-die (that is, piece of wood, a plank, with X embossed) that will imprint the X symbol, or a die with a symbol O (or pairs of symbols together, if we take into account that technical modification). This preparation of the embossed die we will call the operation ``prepare the computation result". Then it moves back down to the row $B_{n+1}$ and impresses the prepared symbol in the sand. This is nothing else but the standard computational procedure, so the sentience of simulated androids must be there.

Next, we will devise a head that does, in some sense, both ``computing the next symbol'' and ``copying the next symbol'' operations at once. It first goes to the row $A_{n+1}$ and engages in the following simple ``prepare copy'' operation: it reads the corresponding symbol at $A_{n+1}$ to be copied and prepares the wooden die with the same symbol it just read. If it now goes back to the row $B_{n+1}$ and imprint this symbol it will be simply performing the symbol copying procedure. But the head will not return yet. Instead, after doing the ``prepare copy'' operation, it will immediately climb one row up, at $A_{n}$, read that symbol and also, in addition, do the ``prepare computation result'' operation, i.e.\ shift the gears so that the correct die for imprinting at $B_{n+1}$ gets prepared, now based on internal state and the ``computation". Only that this time actually nothing changes after the ``prepare computation result'' operation since the correct die was already prepared by the operation ``prepare copy'' (rows $A_{n}$ and $A_{n+1}$ are consistent). Next, the head moves back to the $B_{n+1}$ and imprints the prepared symbol. Again, the sentience of the simulated androids must emerge, as there is no way why the ``prepare copy operation", done before the calculation and in principle ``overwritten'' by the process of calculation could spoil the emergence. After all, this ``prepare copy'' operation can be seen as a part of the computation process, which does not actually do anything, but neither spoils the overall operation (as inserting a few NOP-s in the assembly code).

In the next version we will let the head do just the same, i.e.\ both operations before imprinting the prepared symbol at $B_{n+1}$, only this time we will deliberately break down/disable the last part of the ``prepare computational result'' mechanism -- the one that puts the correct die in position. However, since the correct die was prepared already by the just performed ``prepare copy'' operation, nothing changes in the overall operation. This now broken part of the mechanism was not actually doing anything also in the previous setup. So stopping it from doing something that anyhow did not have any effect cannot be detrimental for the emergence of anything, including sentience. So, if the sentience used to emerge before, it must also now. Only, this time, the reading of the symbol at $A_{n}$ and switching of the gears and internal states, while still being carried out, is absolutely uncorrelated with the symbols written in the $B$ rows. What the head does is merely copying the symbols from $A_{n+1}$ to $B_{n+1}$, followed by some clicking of gears that has nothing to do with the symbols written. In this way we have disentangled the switching of head internal states on one side and writing of the symbols on the other, making them uncorrelated in the cause-effect sense (which is the only sense that can matter in the framework of classical physics), while retaining the subjective experiences of the simulated agents. Except that there is no simulation any longer -- what is happening in the sand is mere copying of the symbols.

To arrive at this final conclusion, we essentially used the fact that it is very hard to explain by which magic would ``the way the head prepares the next symbol to be impressed into the sand'' be crucial in the feat of generating the subjective experience. Whether is the next die is prepared based on the symbol at row $A_{n+1}$, or based on the symbol at row $A_{n}$, cannot truly matter if the result is always the same. And, to be sure, we went through these intermediate steps involving the hybrid head.

Alternatively, we may see this combined mechanism of the hybrid head in the final stage as two separate mechanisms, one stacked over the other. Below is the copying head and it can fully perform its regular operation. Above is the computing head. It also performs its usual duty, only the copying head is always a few seconds ahead of it -- therefore, when the computing one attempts to imprint the correct symbol in the next row, actually nothing is done, since that very symbol has been imprinted by the copying head just a second before. Its die hits the empty space. However, that cannot matter, since we have successfully retained both the dynamics of the sand symbols and the evolution of the computing head (clicking of gears, changing internal states), so that any subjective experience that used to emerge must still be present in this setup. This is also clear since the combination of the two heads can be seen together as a single composite head, that head performs the computation (it can be easily seen if we imagine that there is no log of the previous run to copy from, or if the previous run was different -- anything that the copying head do will be overwritten by the computing mechanism).

On the other hand, there is no longer any causal effect of the computing head on the sand. This is equivalent to the effect of heads positioned somewhere in the midst of the simulation log: as discussed in the Part I, these heads would traverse the symbols but would alter nothing, since all the symbols they attempt to imprint are already there at that very spot. Indeed, we may again, in this latest scenario, position a few of the calculating heads at some position over the rows before, initialize them to the state proper for that line of log and put them in operation. Now, we have a few computing heads active over the sand rows: the original one (that goes in pair with the copying head), and a few additional ones. None of them affects sand in any way. Our point is that, since none of these heads has any causal effect on the sand, there is nothing that singles out the original head, which follows the copying head closely at the last row. And we have concluded already in the Part I that such computing heads, which do not move any sand but merely change their internal states, have nothing to do with the emergence of consciousness. Since our conclusion was that consciousness must emerge also in this setup (if it used to emerge due to action of a single computing head), then it must be solely due to the effect of the copying head, which does all writing in the sand. Therefore, we can safely remove all the computing heads and the remaining copying one must be sufficient to induce subjective experiences of Data.

There is yet another, quite different way to realize that complexity of those gear-motions that are going on in the wooden head (and this complexity is already very low in the simplest of the universal Turing machines) cannot be of significance for the emergence of the consciousness (as long as the correct symbols are generated). Namely, we could repeat the entire proof by using ``Rule 110'' cellular automaton \cite{Rule110} instead of a Turing machine realization. It is less well known that such a simple automaton, where in each step the binary value of every cell is replaced by the value dependent solely on the value of that cell and its two neighbors, is Turing complete, that is, is capable of performing any computation that can be performed using the most advanced computers. Rule 110 can be implemented by almost a trivial wooden device, that features no internal states and no instruction sets. The device needs to only traverse the line of symbols step by step, for each three consecutive symbols check which of the 8 possible patterns they form and according to that imprint a single symbol in the line below. (This checking can be again realized by sensitively ``touching'' the sand, while the patterns that are tested can be for example embossed on a cylinder, at equal angles. The cylinder can revolve, consecutively positioning each of the patterns above the sand surface and trying if that shape fits the relief below. Each of 8 orientations of the cylinder mechanically brings either X or O die over the row bellow. When the patterns match, the symbol is imprinted, the head moves a step to the right, then checks again all eight orientations of the cylinder and so on.) Since the automaton is Turing machine equivalent, and the simulation of one hour takes a finite number of steps on a Turing machine, so it must take finite number of rows of the automaton operation as well and a finite length initial pattern must suffice.

Now, if we go again through the ``repeated run'' versus ``copying run'' argument, it is far more difficult this time to insist that it is the functioning of the wooden machine -- being so trivial -- that is essential for the emergence of the subjective experience. The copying head, which differs only in comparing ``one-bit pattern'' instead of ``three-bit pattern'' is hardly anything mechanically and motion-vise different from the ``Rule 110'' wooden head.\footnote{Besides, copying head can be implemented as a ``Rule 204'' automaton with the otherwise identical design.} Again, we need to remind that ``Rule 110'' or ``automaton'' is here only our description for certain complex motion of wood and sand, and our task is to compare what type of motion has greater potential to elicit subjective experience. Since the motion of the wooden head is so similar and trivial in both cases, what remains as a hypothetical source of emergent consciousness are only the emergent symbols in the sand -- no matter if the correct sequence is inferred from the previous row, or from a previous run (by mere copying). In any case, the burden of a potential proof of the counter-hypothesis that one setup causes consciousness while the other one does not seems to be too heavy.

Once we have concluded that the copying process has the same chances to generate consciousness as the computing one, the rest of the proof is trivial: we can go again through entire reasoning just replacing X with 1 and O with 0, concatenate all rows together, and realize that we are copying a huge number. Whether it is done by a wooden mechanism, by hand or by some other means obviously cannot cause any relevant difference (besides, a human could have played the role of the wooden head from the start). This completes our proof of the no-go theorem.
\end{proof}


\begin{thebibliography}{10}

\bibitem{EverettThesis} Hugh Everett (1955) ``The Theory of the Universal Wavefunction'', Manuscript, pp 3–140 of Bryce DeWitt, R. Neill Graham, eds, The Many-Worlds Interpretation of Quantum Mechanics, Princeton Series in Physics, Princeton University Press (1973), ISBN 0-691-08131-X.

\bibitem{EverettThesis2} Hugh Everett III, Jeffrey A. Barrett (Ed.), Peter Byrne (Ed.), ``The Everett Interpretation of Quantum Mechanics: Collected Works 1955-1980 with Commentary'', Princeton University Press, Year: 2012, ISBN: 0691145075.

\bibitem{Heisenberg} Werner Heisenberg, {\it The representation of nature in contemporary physics}, Deadalus 87, 95 (1958).

\bibitem{Planck} Interview in 'The Observer' (25 January 1931), p.17.

\bibitem{Schrodinger} Erwin Schr\"odinger, ``Mind and Matter", University Press, Cambridge, University press, 1958, p.42.

\bibitem{Wigner} Wigner, Eugene; Henry Margenau (1967). ``Remarks on the Mind Body Question, in Symmetries and Reflections, Scientific Essays", American Journal of Physics. 35 (12): 1169–1170. E. P. Wigner, {\it Remarks on the Mind-Body Question}, The Scientist Speculates. Heineman. (1961).

\bibitem{GRW} Ghirardi, G.C., Rimini, A., and Weber, T., ``Unified dynamics for microscopic and macroscopic systems". Physical Review D. 34: 470 (1986).

\bibitem{OrchOR} S. Hameroff and R. Penrose, ``Consciousness in the universe: A review of the ‘Orch OR’ theory", Physics of Life Reviews 11 (2014) 39-78.

\bibitem{Wheeler1} Wheeler, John A. (1990). ``Information, physics, quantum: The search for links". In Zurek, Wojciech Hubert (ed.). Complexity, Entropy, and the Physics of Information. Redwood City, California: Addison-Wesley. ISBN 978-0-201-51509-1.

\bibitem{Wheeler2} Wheeler, J. A., ``At Home in the Universe'', New York: American Institute of Physics (1994).

\bibitem{QBism} C. A. Fuchs and B. C. Stacey, {\it QBism: Quantum Theory as a Hero’s Handbook}, arXiv:1612.07308v2

\bibitem{Rovelli} C. Rovelli, International Journal of Theoretical Physics 35, 1637 (1996).

\bibitem{Brukner} C. Brukner, {\it On the quantum measurement problem}, in ”Quantum [Un]speakables II”, Eds. R. Bertlmann and A. Zeilinger (The Frontiers Collection, Springer, 2016); arXiv:1507.05255

\bibitem{Brukner2} C. Brukner, {\it A no-go theorem for observer-independent facts}, arXiv:1804.00749v1

\bibitem{ToTheRescue} I. Salom {\it To the rescue of Copenhagen interpretation}, arXiv:1809.01746v1

% the Hard problem:
\bibitem{Chalmers95} D. Chalmers, {\it Facing up to the problem of consciousness}, Journal of Consciousness Studies 2(3):200-219, (1995).

% Cognitive science looking into QM example:
\bibitem{Eccles} J. C. Eccles, ``How the Self Controls its Brain'', Berlin: Springer-Verlag, (1994) ISBN 3-540-56290-7.

% Cognitive science looking into QM example:
\bibitem{Lockwood} M. Lockwood, ``Mind, Brain, and the Quantum: The Compound `I' '', Wiley-Blackwell, ISBN-10: 0631180311, 1992 (1989).

% Cognitive science looking into QM example:
\bibitem{Chalmers96} D. Chalmers, ``The conscious mind'', New York: Oxford University Press (1996).

% Cognitive science looking into QM example:
\bibitem{Hoffman} D. Hoffman, {\it Conscious Realism and the Mind-Body Problem}, Mind and Matter, 6 (2008) 87-121

% Cognitive science looking into QM example:
\bibitem{CognitiveToQuantumExample1} Schwartz Jeffrey, Stapp Henry P and Beauregard Mario, {\it Quantum physics in neuroscience and psychology: a neurophysical model of mind-brain interaction}, Phil. Trans. R. Soc. B
360 (2005). http://doi.org/10.1098/rstb.2004.1598

\bibitem{Chalmers97} D. Chalmers, {\it Moving Forward on the Problem of Consciousness}, Journal of Consciousness Studies 4(1):3-46 (1997).

%Panpsychism:
\bibitem{Koch} C.\ Koch, ``Consciousness: Confessions of a Romantic Reductionist'', MIT Press, 2012.

%Panpsychism:
\bibitem{Tononi} Giulio Tononi, {\it Integrated Information Theory of Consciousness: An Updated Account},  Archives Italiennes de Biologie, Vol. 150, No. 4, pages 293-329; (2012).

%Panpsychism:
\bibitem{PanpsychismOverview} Godehard Bruntrup and Ludwig Jaskolla (eds.), ``Panpsychism: Contemporary Perspectives'', Oxford University Press, 2017, ISBN 9780199359943.

\bibitem{Dennett} Daniel C. Dennett, ``From Bacteria to Bach and Back: The Evolution of Minds", W. W. Norton \& Company 2017, ISBN 978-0-393-24207-2

% Type A materialist example:
\bibitem{TypeA2} Patricia S. Churchland and T. J. Sejnowski, ``The Computational Brain'', (1992) . Cambridge, Massachusetts: The MIT Press.

% Type A materialist example:
\bibitem{TypeA3} Webb TW, Graziano MS (2015). ``The attention schema theory: a mechanistic account of subjective awareness". Front Psychol. 6: 500. doi:10.3389/fpsyg.2015.00500

% Type B materialist example:
\bibitem{TypeB1} Valerie Gray Hardcastle, ``Locating Consciousness'', John Benjamins Publishing Company (1995), ISBN: 978-1556191848

% Type B materialist example:
\bibitem{TypeB2} O'Hara, Kieron and Scutt, Tom, {\it There is No Hard Problem of Consciousness},  J. Shear (ed.), Explaining Consciousness: The Hard Problem. pp. 69-82 (1997).

% Type B materialist example:
\bibitem{TypeB3} John G. Taylor, The Race for Consciousness,  A Bradford Book (2001),ISBN: 978-0262700863

% Type B materialist example:
\bibitem{TypeB4} Hofstadter, Douglas R., ``I Am a Strange Loop'', (2008) New York, NY: Basic Books. ISBN 0-465-03079-3.

\bibitem{ChineseRoom} Searle, John. 1980a. ``Minds, Brains, and Programs.'' Behavioral and Brain Sciences 3, 417-424.

\bibitem{ColorblindMary} Jackson, Frank (1986-01-01). ``What Mary Didn't Know". The Journal of Philosophy. 83 (5): 291–295. doi:10.2307/2026143

\bibitem{PhilosophicalZombies} Kirk, Robert. ``Zombies". The Stanford Encyclopedia of Philosophy (2009), Edward N. Zalta (ed.).

\bibitem{ConsciousnessTest} Susan Schneider, Edwin Turner, ``Is Anyone Home? A Way to Find Out If AI Has Become Self-Aware'', Scientific American (2017), https://blogs.scientificamerican.com/observations/is-anyone-home-a-way-to-find-out-if-ai-has-become-self-aware/

\bibitem{OtherMinds} Avramides, Anita, ``Other Minds", The Stanford Encyclopedia of Philosophy (Summer 2019 Edition), Edward N. Zalta (ed.)

\bibitem{SearleOnRebellion} J.R. Searle, ``What Your Computer Can’t Know'', The New York Review of Books (October 9, 2014), https://www.nybooks.com/articles/2014/10/09/what-your-computer-cant-know/

\bibitem{SimulationHypothesis} Bostrom, N., 2003, {\it Are You Living in a Simulation?}, Philosophical Quarterly (2003), Vol. 53, No. 211, pp. 243-255.

\bibitem{TMgeneral} De Mol, Liesbeth, ``Turing Machines", The Stanford Encyclopedia of Philosophy (Winter 2018 Edition), Edward N. Zalta (ed.)

\bibitem{WoodenTuringMachine} Lara Grant, ``Mechanical wooden Turing Machine'', March 8, 2018, https://hackaday.com/2018/03/08/ mechanical-wooden-turing-machine

\bibitem{BoltzmannBrain} Sean M. Carroll, {\it Why Boltzmann Brains Are Bad}, arXiv:1702.00850 [hep-th]

\bibitem{DeWitt} Bryce S. DeWitt, ``Quantum mechanics and reality", Physics Today 23, 9, 30 (1970)

\bibitem{MWrequiresPostulate} Vaidman, Lev, ``Many-Worlds Interpretation of Quantum Mechanics", The Stanford Encyclopedia of Philosophy (Fall 2018 Edition), Edward N. Zalta (ed.)

\bibitem{MWchance} S.\ Saunders, {\it Chance in the Everett Interpretation}, In Many Worlds: Everett, Quantum Theory, and Reality, Oxford University Press (2010). DOI: 10.1093/acprof:oso/9780199560561.001.0001

\bibitem{ManyMinds} D. Albert and B. Loewer, Synthese 77 (1988) 195–213. https://doi.org/10.1007/BF00869434

% Many worlds require many minds, problem with emergence of probability
\bibitem{Squires} E. Squires, {\it One mind or many -- A note on the everett interpretation of quantum theory}, Synthese (1991) 89: 283. https://doi.org/10.1007/BF00413909

\bibitem{MWproblems1} Adrian Kent, ``One world versus many: The Inadequacy of Everettian accounts of evolution, probability, and scientific confirmation", In Saunders, Simon (ed.) et al.: Many worlds? Everett, quantum theory, and reality  (2010) 307-354;

\bibitem{MWproblems2} H. Greaves, {\it Understanding Deutsch’s probability in a deterministic multiverse}, Studies in History and Philosophy of Modern Physics 35 (2004) pp.423-456;

\bibitem{ActOfCreation} J. A. Wheeler, Law without law, in ``Quantum theory and measurement", eds. J.A.Wheeler and W.H.Zurek, Princeton University Press, Princeton (1983) 182.

\bibitem{NoSelfMeasurement} Breuer, T., 1993, {\it The impossibility of accurate state self-measurements}, Philosophy of Science, 62: 197-214.

\bibitem{DelayedChoice} Wheeler, J. A. {\it The `Past' and the `Delayed-Choice' Double-Slit Experiment},
In Mathematical Foundations of Quantum Theory , edited by A. R. Marlow.
New York: Academic Press, 1978, 9-48. Reprinted in part in WZ, 148 pp. 182-200.

\bibitem{RennerNature} D.\ Frauchiger, R.\ Renner, {\it Quantum theory cannot consistently describe the use of itself}, Nat Commun. 9, 3711 (2018), DOI: 10.1038/s41467-018-05739-8

\bibitem{TwentyQuestions} Gribbin, John; Gribbin, Mary; Gribbin, Jonathan, ``Q is for Quantum: An Encyclopedia of Particle Physics, (2000) Simon and Schuster. ISBN 9780684863153.

% Amplifiers of randomness:
\bibitem{RandomnessAmp} P.\ Jedlicka, {\it Revisiting the Quantum Brain Hypothesis: Toward Quantum (Neuro)biology?}, Front Mol Neurosci. (2017);10:366. doi:10.3389/fnmol.2017.00366

\bibitem{IITUnconsciousAI} G.\ Tononi and C.\ Koch, {\it Consciousness: here, there and everywhere}, Phil. Trans. R. Soc. B, 370 (2015) DOI: full/10.1098/rstb.2014.0167

\bibitem{Zuboff} Zuboff, A., {\it One self: The logic of experience}, Inquiry, 33(1), pp.39-68 (1990).

\bibitem{Knight} Andrew Knight, {\it Refuting Strong AI: Why Consciousness Cannot Be Algorithmic},	arXiv:1906.10177

\bibitem{SmallestBinaryUTM} Neary, T., and Woods, D., {\it Four small universal Turing machines}, Proc. 5th Int. Conf. on Machines, Computations, and Universality LNCS-4664, Springer-Verlag, 242–254 (2007).

\bibitem{SmallestUTM} Rogozhin, Y., {\it Small universal Turing machines}, Theoretical Computer Science, 168, 215–240 (1996).

\bibitem{Rule110} M.\ Cook,  {\it Universality in Elementary Cellular Automata}, Complex Systems. 15: 1–40 (2004).

\end{thebibliography}
\end{document}